\documentclass[12pt,letter]{article}
\usepackage{titling}
\usepackage{bibunits}
\usepackage{amsmath}
\usepackage{amsfonts}
\usepackage{graphicx,psfrag,epsf}
\usepackage{enumerate}
\usepackage{natbib}
\usepackage{url} % not crucial - just used below for the URL 
\usepackage[breaklinks=true,bookmarksopen=true,colorlinks=true,citecolor=blue]{hyperref}
\usepackage{fontawesome}
\usepackage{setspace}
\usepackage{ragged2e}   
\usepackage{float}
 
\usepackage{booktabs}   
\usepackage{multirow}   
\usepackage{float}      

\usepackage{threeparttable}

\newcommand{\blind}{1}

\usepackage{tikz-network}

\usepackage{amsmath}
\usepackage{kbordermatrix}

\usepackage{mwe}
\usepackage{subfig}

\usepackage{tikz}
\usepackage{pgfplots}

\usepackage{pgfplotstable}
\usepgfplotslibrary{groupplots}

\usepackage{pdflscape}

\newtheorem{theorem}{Theorem}

\newtheorem{lemma}{Lemma}[section]

\newtheorem{proposition}{Proposition}

\newtheorem{assumption}{Assumption}

\newenvironment{proof}[1][Proof]{\noindent\textbf{#1.} }{\ \rule{0.5em}{0.5em}}
\renewcommand{\baselinestretch}{1.3}
\makeatletter
\renewcommand{\thetheorem}{\arabic{theorem}}
\makeatother

\numberwithin{equation}{section}
\makeatletter
\begin{document}
\begin{bibunit}[jpe]

\def\spacingset#1{\renewcommand{\baselinestretch}%
{#1}\small\normalsize} \spacingset{1}

%%%%%%%%%%%%%%%%%%%%%%%%%%%%%%%%%%%%%%%%%%%%%%%%%%%%%%%%%%%%%%%%%%%%%%%%%%%%%%

\if1\blind
{
\title{Regression with Observational Multilayered Network Data}
\author{Juan Estrada\thanks{Analysis Group Economic Consulting, Washington, DC, USA. \faEnvelopeO: \href{mailto:juan.estrada@analysisgroup.com}{juan.estrada@analysisgroup.com}.} \and Kim P. Huynh\thanks{Department of Economics, Indiana University, 100 S Woodlawn, Bloomington, IN 47405, USA. The Laboratoire d’\'Economie d’Orl\'eans, Universit\'e d'Orl\'eans, Orl\'eans, France. \faEnvelopeO: \href{mailto:kim@huynh.tv}{kim@huynh.tv}.} \and David T. Jacho-Ch\'{a}vez\thanks{Corresponding Author: Department of Economics, Emory University, Rich Building 306, 1602 Fishburne Dr., Atlanta, GA 30322-2240, USA. \faEnvelopeO: \href{mailto:djachocha@emory.edu}{djachocha@emory.edu}.} \and Leonardo S\'{a}nchez-Arag\'{o}n\thanks{Facultad de Ciencias Sociales y Human\'{i}sticas, Escuela Superior Polit\'{e}cnica del Litoral, ESPOL, Campus Gustavo Galindo Km. 30.5 V\'{i}a Perimetral, P.O. Box 09-01-5863, Guayaquil, Ecuador. \faEnvelopeO: \href{mailto:lfsanche@espol.edu.ec}{lfsanche@espol.edu.ec}.}}
  \maketitle
} \fi

\if0\blind
{
  \bigskip
  \bigskip
  \bigskip
  \begin{center}
    {\LARGE\bf Regression with Observational Multilayered Network Data}
\end{center}
  \medskip
} \fi

\bigskip
\begin{abstract}
\noindent  A novel method to estimate social effect coefficients in the popular so-called linear-in-means regression model in the Social Sciences is presented here that utilizes non-experimental multidimensional network data. The procedure can accommodate social interactions that correlate with the error in the model by making use of a different set of network links among the same observations that are exogenous in the traditional sense. In particular, the full observability of a two-layered \emph{multiplex} network data structure is assumed here to propose a new Generalized 3-Stage Least Squares (G3SLS) estimator that is consistent,  asymptotically normally distributed, and also easy to implement using widely-used existing statistical software because of its closed-form definition. The underlying assumptions are general enough to accommodate common problems with observational data such as measurement error, simultaneity, and unobserved heterogeneity. Monte Carlo exercises confirm the good small sample performance of the proposed G3SLS estimator in these scenarios. An empirical application finds positive and significant peer effects in citations among research articles published in top general-interest journals in economics. 
\end{abstract}

\noindent%
{\it Keywords:}  Instrumental Variables; Linear-in-Means Models; Multidimensional Networks; Multilayered Networks; Multiplex Networks
\\
\noindent%
{\it JEL code:} A1, C21, C31, C51, I23, J24.  
\vfill

\newpage
\spacingset{1.45} % DON'T change the spacing!
\section{Introduction}\label{sec:intro}
In social science research, understanding the causal mechanisms behind individual outcomes is a central challenge. With network data, the outcome of a unit, such as the performance of a firm or the academic achievement of a student, is often influenced by its own characteristics and those of its peers (those to whom they are connected). The \emph{linear-in-means model} is a standard tool for quantifying these interactions, which are broadly categorized as direct ($\boldsymbol{\gamma}$), peer ($\beta$), and contextual effects ($\boldsymbol{\delta}$). The model is formalized as follows.

\begin{equation}
\label{intro}
  y_{i}=\alpha+\beta\sum_{j \neq i}w_{i,j}y_{j}+\sum_{j \neq i}w_{i,j}\mathbf{x}_{j}^{\top}\boldsymbol{\delta}+ \mathbf{x}_{i}^{\top}\boldsymbol{\gamma}+v_{i}\text{,}
\end{equation}

\noindent where $i, j \in \{1, \dots, n\}$ represents the nodes within the network, $\mathbf{x}$ is a vector of attributes associated with the units, and $w_{i,j}$ takes a value of 1 if the nodes $i$ and $j$ are connected (indicating a link or edge) and 0 otherwise. The term $v_{i}$ accounts for an unobserved random error, and $n$ is the total number of nodes (or observations) in the network. The entire network structure is represented by an adjacency matrix $n \times n$, $\mathbf{W}$, where the element at position $(i,j)$ corresponds to $w_{i,j}$. The parameters in \eqref{intro} are straightforward to identify when the network structure is exogenous; that is, it is not correlated with the unobserved error term of the model. The error term $\mathbf{v} = [v_1, \dots, v_n]^{\top}$ has an expected value of zero conditional on the observed covariates and the network structure, i.e., $E[\mathbf{v}|\mathbf{X}, \mathbf{W}] = \mathbf{0}$, where $\mathbf{X} = [\mathbf{x}_1, \dots, \mathbf{x}_n]^{\top}$ is the matrix of covariates; see \cite{Paula2017} and references therein.

However, when the assumption of exogeneity does not hold—meaning that $E[\mathbf{v}|\mathbf{X}, \mathbf{W}] \neq \mathbf{0}$—the situation becomes more complex. This occurs when the network structure or the error term is correlated with the covariates, which introduces endogeneity. Endogeneity can arise from various sources, such as simultaneity bias (where outcomes and connections influence each other), errors in the measurement of network links, or homophily (the tendency of similar units to form connections). This challenge remains a focal area of study in social science research, as highlighted by works such as \cite{Johnsson2019} and \cite{Chan_et_al_social_effects}.

As in \cite{Chan_et_al_social_effects}, our paper uses a type of multilayered (multidimensional) network data structure known as \emph{multiplex} networks to consistently estimate and perform correct inference on the structural social parameters in \eqref{intro} with a potentially endogenous network structure $\mathbf{W}$. In particular, it is assumed that the researcher observes another set of social ties between the original observations $\{w_{0;i,j}\}_{j=1,j\neq i}^n$ in the form of an adjacency matrix $\mathbf{W}_0$ $n\times n$, which is exogenous in the usual sense, i.e., $E_{\mathbf{X},\mathbf{W}_0}[\mathbf{v}]=\mathbf{0}$, instead of leaving the correlation structure between the errors $\mathbf{v}$ and the original social structure $\mathbf{W}$ in \eqref{intro} unspecified. This type of data structure is becoming increasingly popular in the Social Sciences; see, e.g., \cite{Jackson_Multiplex} in Anthropology, \cite{manta2021} in Economics, \cite{Aldasoro_Alves_JFinStab} in Finance, \cite{An_et_al_2026} in Sociology, and \cite{PoliSci_Multiplex} in Political Science. See \cite{Boccaletti2014} and \cite{Kivela_multilayer_network_2014} for an up-to-date survey of the mathematical underpinning of this type of data.

Our proposed estimator of parameters in \eqref{intro} is performed in three steps; we shall hereafter refer to it as the Generalized Three-Stage least squares estimator (G3SLS). Using the linearity of the model and simple linear projection arguments in each step, the resulting estimation procedure is very simple to implement. It is readily available in \texttt{Stata}, i.e., \cite{netivreg_stata_journal}. Furthermore, the estimator is shown to be asymptotically normal at the standard root--$n$ rate of convergence. The estimation effect from the multistage procedure is fully characterized, and a consistent asymptotic variance-covariance estimator is proposed as well.

Our approach is related to a growing literature that uses instrumental variables for identification in linear-in-means models, notably \cite{Kelejian_et_al_2014}, \cite{Konig_et_al_2019}, and \cite{Lee_et_al_2021}. These papers generally address network endogeneity by constructing instruments derived from observable dyadic attributes or estimated link formation probabilities. In contrast, the G3SLS estimator proposed here leverages the multiplex data structure by assuming that the researcher observes a distinct exogenous layer $\mathbf{W}_0$ that correlates with the endogenous layer $\mathbf{W}$. Unlike \cite{Kelejian_et_al_2014}, who achieves identification by assuming that the elements of the endogenous weighting matrix are linear functions of observable exogenous dyadic variables (such as distance or border length) and then estimating these weights to construct the instrument matrix, our identification strategy utilizes a linear projection onto a distinct, observed exogenous network layer $\mathbf{W}_0$. Furthermore, unlike \cite{Konig_et_al_2019}, who addresses endogeneity by explicitly modeling the strategic network formation process to construct instruments based on predicted link probabilities, our method does not require a specific description of the network formation process, allowing it to remain robust to different sources of endogeneity, such as simultaneity and measurement error. Additionally, while \cite{Lee_et_al_2021} constructs instruments by weighting peers' attributes with predicted link probabilities derived from a logistic regression on exogenous dyadic characteristics, our method directly utilizes the observed connections in the exogenous network $\mathbf{W}_0$ for identification, thereby bypassing the need to estimate a link formation model.

The Monte Carlo exercise showcases the versatility of the proposed estimator by analyzing its performance in three separate data generating processes that consider different scenarios, such as omitted variable bias, measurement error, and unobserved homophily in network formation. The simulation results show that the proposed G3SLS estimator performs well in terms of bias and root mean squared error, with sample sizes as small as 50 observations. In addition to the simulated data, this paper also presents an application to real data on publication outcomes in Economics \citep{netivreg_g3sls}. The use of web scraping and existing data on authors' research fields, education, and employment history allows for the creation of two types of professional ties among scholars, namely co-authorship and alumni connections. These multiplex networks are then used to uncover positive and significant peer effects in terms of citations among articles published by these scholars, as well as significant positive effects of research teams that are gender diverse on the quality of a paper, measured in terms of citation outcomes after controlling for other articles' characteristics, such as the number of pages, number of bibliographic references, co-authoring with current and previous editors, and various network fixed effects.

The structure of the paper is as follows. Section \ref{background} defines multiplex networks with an example and provides various scenarios in which the required exogeneity requirement for one of the network layers can be naturally satisfied. Section \ref{indet} introduces the model and identification conditions for the parameters of interest. Section \ref{est} describes the multi-step estimation procedure, asymptotic properties, and valid asymptotic standard errors. Section \ref{mc} presents the small sample properties of the proposed estimator in various Monte Carlo exercises, while Section \ref{emp} discusses an empirical application. Finally, Section \ref{conclusion} concludes. 

\ref{Appendix_A} and \ref{Appendix_B} contain, respectively, all mathematical derivations of the main results and the conditions used to establish the asymptotic properties of the proposed estimator, while intermediate steps are collected in \ref{Appendix_C}. The supplementary material reports further numerical experiments evaluating the proposed estimator across various scenarios, see \ref{Appendix_D}. Finally, \ref{Appendix_E} and the references therein provide details on the real data application and several robustness checks.

\section{Background\label{background}}
\subsection{Multiplex Networks}
The cornerstone of the identification of social parameters in \eqref{intro} in this paper is the complete observability of more than one type of social interaction between economic agents. These data structures are known as \emph{multi-dimensional} or \emph{multi-layered} networks; see, e.g., \cite{Boccaletti2014} and \cite{Kivela_multilayer_network_2014} for up-to-date comprehensive surveys and references therein. Following the definition of \citeauthor{Boccaletti2014} (\citeyear{Boccaletti2014}) a multilayer network is a pair $\mathcal{M}=(\mathcal{G},\mathcal{C})$, where $\mathcal{G}=\{g_{m};\quad m \in \{1,\dots,M\}\}$ is a family of $M$ graphs, and $\mathcal{C}$ is the set of interconnections between nodes of different layers $g_{\alpha}$ and $g_{\beta}$ with $\alpha\neq \beta$. When the same nodes are in each layer and there are no connections between different nodes in different layers except with themselves, these networks are called \emph{multiplex}; see, e.g., \cite{Jackson_Multiplex}. The latter is the specific type of data structure used in this paper for identification.

Figure \ref{F1}(a) presents an example of a two layer multiplex network, where there are two different types of connections between the same four nodes, i.e., the blue and red edges. The network $G_{0}$ depicted with the blue edges represents a (possibly predetermined) network, while the red edges in $G$ represent another (possibly endogenous) network. In this example, $M=2$, $g_{\alpha}=\{(1,2),(1,4),(3,4)\}$, and $g_{\beta}=\{(1,2),(1,3),(3,4)\}$ are described in the corresponding adjacency matrices $\mathbf{W}$ and $\mathbf{W}_0$ in panel (c). The relationship between the networks is given by the different types of edges that are shared by the same nodes; see, i.e., the flattened version of the multiplex network in Figure \ref{F1}(b). In this example, many of the red connections exist where the blue connections also exist, showing an important relationship between the networks. The matrix product $\widetilde{\mathbf{W}}\equiv\mathbf{W}_{0}\mathbf{W}$ in panel (c) encapsulates a measure of correlation between the networks that is relevant for the identification and estimation results presented in Sections \ref{indet} and \ref{est} below, respectively. The $(i,i)$th element of  $\widetilde{\mathbf{W}}$ represents the number of nodes $j \neq i$ that are connected to node $i$ by both types of edges (blue and red), e.g., the position $(1,1)$ of $\widetilde{\mathbf{W}}$ equals one because only node $2$ is connected with node $1$ by both types of edges. Additionally, the $(i,j)$th element of $\widetilde{\mathbf{W}}$ contains the number of two length paths that start with a blue connection from $\mathbf{W}_{0}$ and are followed by a red connection from $\mathbf{W}$. These types of paths are called \textit{inter-layer intransitive triads} (`a friend's relative in the family layer is not a friend in the friendship layer') hereafter. For example, the position $(1,4)$ in $\widetilde{\mathbf{W}}$ is equal to one because there is an inter-layer intransitive triad that changes the edge colors connecting nodes 1 and 4 via node 3.

\begin{figure}[!htb]
\caption{Example of a Two-Layered Multiplex Network and Their Adjacency Matrices}\label{F1}
    \centering
\begin{minipage}{0.33\textwidth}

\vspace{7mm}

        \begin{tikzpicture}[multilayer=3d]

\Vertex[x=-0.1, y = 0.2, label = 1,layer=1, color = red!60]{1b}
\Vertex[x=0.5, y = 0.8, label = 2,layer=1, color = red!60]{2b}
\Vertex[x=1.7, y = 0.1, label = 3, layer=1, color = red!60]{3b}
\Vertex[x=1, y = -1, label = 4,layer=1, color = red!60]{4b}

\Edge[color = red](1b)(2b)
\Edge[color = red](3b)(4b)
\Edge[color = red](1b)(4b)

\Vertex[x=0, y = -0.4, label = 1,layer=2, color = blue!60]{1a}
\Vertex[x=0.5, y = 0.4, label = 2,layer=2, color = blue!60]{2a}
\Vertex[x=1.7, y = -0.2, label = 3, layer=2, color = blue!60]{3a}
\Vertex[x=1.1, y = -1.6, label = 4,layer=2, color = blue!60]{4a}

\Edge[color = blue](1a)(2a)
\Edge[color = blue](1a)(3a)
\Edge[color = blue](4a)(3a)

\SetLayerDistance{-3}
\Plane[x=-0.6,y=-1.7,width=3,height=3,color=red!40,layer=1]
\Plane[x=-1.7,y=-0.7,width=3,height=3, color=blue!40,layer=2]

\Edge[style=dashed](1a)(1b)
\Edge[style=dashed](2a)(2b)
\Edge[style=dashed](4a)(4b)
\Edge[style=dashed](3a)(3b)

\Text[x = -0, y = -2.2, layer = 1, color = red, opacity = 2, rotation=30]{$G$}
\Text[x = -0.3, y = -1.5, layer = 2, color = blue, opacity = 2, rotation=30]{$G_0$}
\end{tikzpicture} 
        
\begin{center}
Panel (a)
\end{center}
\end{minipage}
\hfill
\begin{minipage}{0.33\textwidth}
\flushright

\vspace{16mm}

     \begin{tikzpicture}
    \Vertex[x=0, y=0, color = white, label = 1]{1}
    \Vertex[x=1,y=1, color = white, label = 2]{2}
    \Vertex[x=2, y=0, color = white, label = 3]{3}
    \Vertex[x=1, y=-1, color = white, label = 4]{4}
    
    \Edge[color = blue, bend = -30](1)(2)
    \Edge[color = blue](1)(3)
    \Edge[color = blue, bend = -30](3)(4)

    \Edge[color = red, bend = 30](1)(2)
    \Edge[color = red](1)(4)
    \Edge[color = red, bend = -30](4)(3)
    
    \Text[x = -0.5 , y = -0.2, layer = 1, color = red]{$\mathbf{G}$}
    \Text[y = 2.3, x = 0.5, layer = 1, color = blue]{$\mathbf{G}_{0}$}
\end{tikzpicture}
     
\vspace{18mm}

\begin{center}
Panel (b)
\end{center}
\end{minipage}%
\hfill
\begin{minipage}{0.30\textwidth}

\vspace{5mm}

\centering
\footnotesize

         \[
            \textcolor{red}{\mathbf{W}}=
            \begin{bmatrix}
            0 & \textcolor{red}{1} & 0 & \textcolor{red}{1}  \\
            \textcolor{red}{1} & 0 & 0 & 0  \\
             0 & 0 & 0 & \textcolor{red}{1} \\
            \textcolor{red}{1} & 0 & \textcolor{red}{1} & 0
            \end{bmatrix}
            \]
            
         \[
            \textcolor{blue}{\mathbf{W}_{0}}=
            \begin{bmatrix}
            0 & \textcolor{blue}{1} & \textcolor{blue}{1} & 0  \\
            \textcolor{blue}{1} & 0 & 0 & 0  \\
             \textcolor{blue}{1} & 0 & 0 & \textcolor{blue}{1} \\
            0 & 0 & \textcolor{blue}{1} & 0
            \end{bmatrix}
            \]

         \[
           \textcolor{blue}{\mathbf{W}_{0}}\textcolor{red}{\mathbf{W}}=
            \begin{bmatrix}
            1 & 0 & 0 & 1  \\
            0 & 1 & 0 & 1  \\
             1 & 1 & 1 & 1 \\
            0 & 0 & 0 & 1
            \end{bmatrix}
            \]
\normalsize

\vspace{5mm}
\begin{center}
Panel (c)
\end{center}
\end{minipage}

\vspace{0.3cm}

\begin{minipage}{1\textwidth}
\footnotesize
Note: Panel (a) displays an undirected two-layered \emph{multiplex} network structure among four nodes. The red edges showcase a possibly endogenous network, while the blue edges display exogenous connections among these agents -- with associate adjacency matrices in Panel (c). The figure in panel (b) displays the flattened version of the multiplex network. It facilitates the visualization of the inter-layer intransitive triads such as those between nodes 1 and 4 or nodes 2 and 4, i.e., matrix $\widetilde{\mathbf{W}}$.
\end{minipage}
\end{figure}
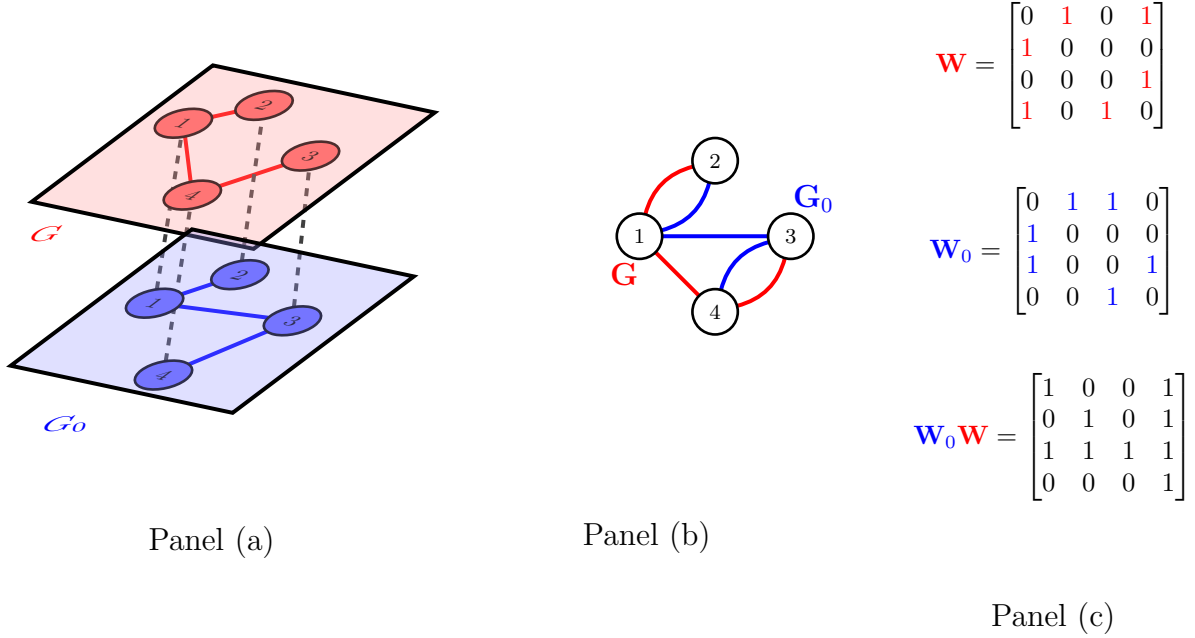

\subsection{Some Examples\label{some_examples}}

In economics, \cite{Goldsmith-Pinkham2013} and \cite{Chan_et_al_social_effects} have used multidimensional network data structures for causal inference. \cite{Goldsmith-Pinkham2013} used previously observed sets of connections of the same type as explanatory variables in their endogenous network formation model, while \cite{Chan_et_al_social_effects} used different types of peers, i.e., roommates, classmates, seatmates, studymates, and friends, to study peer effects on academic performance. The identification strategy put forward in this paper utilizes the latter type of multiplex networks by simply pointing out that some of these types of connections are exogenous in nature (blue edges in Figure \ref{F1}(b)), e.g., `roommates' \citep{Sacerdote_QJE}, `ethnic groups' \citep{Reza2019}, and `classmates' \citep{Ammermueller_Pischke_JLO}, since they are not directly chosen by the units of observation, while others are likely endogenous (red edges in Figure \ref{F1}(b)), e.g., `friends' \citep{Fruehwirth_RevStat} or `studymates' \citep{Chan_et_al_social_effects} due to potential homophily in unobserved characteristics.

Similarly, in applications about publication outcomes in Academia, see, e.g., \cite{Newman_2004a} in Physics or \cite{AoAS_coauthorship} in Statistics, co-authorship connections are potentially endogenous (red edges in Figure \ref{F1}(b)). Information relating to authors who received their Ph.D. from the same institution can be considered pre-determined (blue edges in Figure \ref{F1}(b)) and therefore exogenous. The empirical application in Section \ref{emp} below utilizes this observation in an application where the outcome variable in \eqref{intro} is the natural logarithm of citations of a peer-reviewed research article published in the top journals of general-interest in economics between 2002 and 2006.

Finally, it should  be noted that the best way to guaranty the exogeneity of the network $G_{0}$ would be by (quasi) randomizing individuals into groups. This has been done before. For example, \cite{carrell2013} randomly assigned freshmen students at the United States Air Force Academy to peer groups (squadrons), while \cite{falk2006} randomly assigned workers to groups in the lab. However, as pointed out by \cite{carrell2013}, these types of experimental designs do not consider the endogenous process of individuals sorting into more granular peer groups such as friends or study partners based on the intrinsic characteristics of the individuals. If the network of interest (for example, the friendship network) is observed by the researcher and can be represented by $G$, the proposed method here can be used to find the causal social effects generated by the network of interest, $G$.

\section{The Model\label{indet}}
This section presents the basic model used to quantify the peer effects in a linear-in-means regression model in \eqref{intro}. As is often the case with observational data, a set of sufficient conditions on the joint data generating process is presented here that guaranties the social parameters are uniquely recovered (identification) from the estimating sample. Let $\mathbf{W}$ be an $n \times n$ stochastic adjacency matrix corresponding to a random network. The matrix $\mathbf{W}$ may or may not be row normalized; the identification results apply to both types of adjacency matrices. The structural model \eqref{intro} in matrix form is given by:

\begin{equation}
\label{E1}
\mathbf{y}=\alpha\boldsymbol{\iota}+\beta\mathbf{Wy}+\mathbf{WX}\boldsymbol{\delta}+\mathbf{X}\boldsymbol{\gamma}+\mathbf{v}%
\text{,}%
\end{equation}

\noindent where $\alpha$ and $\beta$ are  scalar structural parameters, $\boldsymbol{\gamma}$ and $\boldsymbol{\delta}$ are $k \times 1$ vectors of structural parameters, $\mathbf{y}$ is an $n\times1$ vector of outcomes, $\boldsymbol{\iota}$ is a $n\times1$ vector of ones, and $\mathbf{X}$ is an $n\times k$ matrix of exogenous covariates. Under the assumption that $E_{\mathbf{X},\mathbf{W}}[\mathbf{v}]=\mathbf{0}$, regressors $\mathbf{X}$ and $\mathbf{WX}$ are exogenous, but $\mathbf{Wy}$ is not because of simultaneity. Since \eqref{E1} can be thought of as a spatial autoregressive model, \cite{Kelejian1998,Kelejian_Prucha_1999_ER}, \cite{Lee2003}, and \cite{Bramoulle2009}, as well as \cite{BL:2011}, proposed using $[\mathbf{X},\mathbf{WX},\mathbf{W}^{2}\mathbf{X}]$ as the matrix of instruments for the matrix of \emph{endogenous} regressors $[\mathbf{X},\mathbf{WX},\mathbf{Wy}]$ in a Generalized Two-Stage Least Squares (G2SLS) procedure. The use of powers of the adjacency matrix of the form $\mathbf{W}^{2}\mathbf{X}$ as instruments for $\mathbf{W}\mathbf{y}$ is common in the social science literature. The rationale behind the validity of those instruments is that the existence of intransitive triads (contained in $\mathbf{W}^{2}$) guaranties that if nodes $i$ and $j$ are connected, and nodes $j$ and $k$ (but not $i$ and $k$) are connected, then $\mathbf{x}_{k}$ affects $\mathbf{y}_{i}$ but only through its effect on $\mathbf{y}_{j}$. Notice that only one variable of this matrix is endogenous, i.e., $\mathbf{Wy}$, since $E_{\mathbf{X}}[\mathbf{v}]=\mathbf{0}$, and $E_{\mathbf{X},\mathbf{W}}[\mathbf{v}]=\mathbf{0}$ by the tower property of conditional expectations.

This paper assumes instead that $E_{\mathbf{X},\mathbf{W}_{0}}[\mathbf{v}]=\mathbf{0}$, where $\mathbf{W}_{0}$ is an exogenous undirected adjacency matrix that may or may not be row normalized as well. In this way, $\mathbf{W}$ is allowed to be \emph{endogenous}, i.e., $E_{\mathbf{X},\mathbf{W}}[\mathbf{v}]\neq\mathbf{0}$, which invalidates the instruments previously proposed, rendering the G2SLS estimator inconsistent. Notice that there is a set of $k+1$  endogenous regressors in equation \eqref{E1}, i.e., $\mathbf{WX}$ and $\mathbf{Wy}$, given that $E_{\mathbf{X}}[\mathbf{v}]=\mathbf{0}$ by the law of iterated expectations applied to $E_{\mathbf{X},\mathbf{W}_{0}}[\mathbf{v}]=\mathbf{0}$, but $E_{\mathbf{X},\mathbf{W}}[\mathbf{v}]\neq\mathbf{0}$ by the tower property of conditional expectations. Assumptions \ref{A1} and \ref{A2} below state the necessary conditions for $\mathbf{W}_{0}$ to be a valid instrument.

\begin{assumption}
There exists an $n \times n$ adjacency matrix $\mathbf{W}_{0}$ such that $E_{\mathbf{X},\mathbf{W}_{0}}[\mathbf{v}]=\mathbf{0}$.\label{A1}
\end{assumption}

From the original condition, notice that by applying the law of iterated expectations to $E_{\mathbf{X},\mathbf{W}_{0}}[\mathbf{v}]=\mathbf{0}$, the expression $E_{\mathbf{W}_{0}}[\mathbf{v}]=\mathbf{0}$ is obtained. Then, for all $i$ and $E_{\mathbf{W}_{0}}[v_i]=0$, this consequently implies that $E[\mathbf{w}_{0;i}v_i]=0$ by the conditioning theorem, where $\mathbf{w}_{0;i}$ represents the $i$th row of $\mathbf{W}_{0}$. Additionally, Assumption \ref{A1} also implies $E[\mathbf{x}_{i}v_i]=\mathbf{0}$ and $E[\mathbf{W}_{0}\mathbf{x}_{i}v_i]=\mathbf{0}$ by the same logic. The latter implication corresponds to Assumption 1 in \cite{Chan_et_al_social_effects}, and therefore Assumption \ref{A1} above is stronger. Note that equation \eqref{E1}, Assumption \ref{A1}, and the multiplex data structure together imply an exclusion restriction on the adjacency matrix $\mathbf{W}_{0}$. The exogenous adjacency matrix should only affect the outcome $y_{i}$ through its correlation with the adjacency matrix of interest, $\mathbf{W}$. A setup where individuals are (quasi-) randomized into groups is likely to guaranty the validity of the exclusion restriction on $\mathbf{W}_0$. In these settings, individuals are randomly selected into groups that do not necessarily have the potential to create network effects \citep[see, e.g.,][]{carrell2013}. On the other hand, in observational studies, a potential exclusion restriction argument could be based on predetermined networks with respect to the outcome. A network formed in the past is likely to affect outcomes only indirectly through its correlation with contemporaneous relevant networks, as in our real data application.

Consider the regressors formed with the endogenous matrix $\mathbf{W}$, i.e., $\mathbf{W}\mathbf{y}$ and $\mathbf{W}\mathbf{X}$. Let $\mathbf{S}$ be a $n \times (k+1)$ matrix given by $\mathbf{S}\equiv[\mathbf{y} \quad \mathbf{X}]$, and $\boldsymbol{\theta}\equiv(\beta,\boldsymbol{\delta}^{\top})^{\top}$ be a $(k+1) \times 1$ vector of parameters such that $\beta\mathbf{Wy}+\mathbf{WX}\boldsymbol{\delta}=\mathbf{WS}\boldsymbol{\theta}$. Therefore, equation \eqref{E1} can be written as 

\begin{equation}
\label{E2}
\mathbf{y}=\alpha\boldsymbol{\iota}+\mathbf{WS}\boldsymbol{\theta}+\mathbf{X}\boldsymbol{\gamma}+\mathbf{v}%
\text{,}%
\end{equation}

\noindent where, given Assumption \ref{A1},  the endogenous adjacency matrix $\mathbf{W}$, with $i$th row denoted as $\mathbf{w}_{i}$, can be instrumented by $\mathbf{w}_{0;i}$. The above implies that instruments can be constructed by combining the predetermined instrumental matrix $\mathbf{W}_{0}$ and the regressor defined by $\mathbf{S}$ in equation \eqref{E2}, i.e., this is formalized in the following assumption.

\begin{assumption}
\label{A2}
Let $\mathbf{\Pi}$ be the $(k+1)\times(k+1)$ matrix of coefficients from the system regression

\begin{equation}
\label{E3}
\mathbf{WS}  =\mathbf{W}_{0}\mathbf{S\Pi}+\mathbf{U}\text{,}
\end{equation}

\noindent where the $n \times(k+1)$ matrix of system errors $\mathbf{U}$ is such that $E_{\mathbf{W}_{0}\mathbf{y}, \mathbf{W}_{0},\mathbf{X}}[\mathbf{U}]=\mathbf{O}$ (a matrix of zeros), $E[\mathbf{S}^{\top}\mathbf{w}_{0;i}\mathbf{w}_{0;i}^{\top}\mathbf{S}]$ is positive definite, and \emph{rank}$(\mathbf{\Pi})=k+1$. Furthermore, the first row of $\mathbf{\Pi}$, $\boldsymbol{\pi}_{1}$, is such that $\boldsymbol{\pi}^{\top}_{1}\boldsymbol{\theta}< 1/\lambda_{\textup{max}}$, where $\lambda_{\textup{max}}$ is the largest eigenvalue of $\mathbf{W}_{0}$ and $\boldsymbol{\theta}\equiv(\beta,\boldsymbol{\delta}^{\top})^{\top}$.
\end{assumption}

Note that Assumption \ref{A2} implies that  $\mathbf{\Pi}$ is uniquely determined by the joint probability distribution of $(\mathbf{S}^{\top}\mathbf{w}_{i},\mathbf{S}^{\top}\mathbf{w}_{0;i})$. The rank condition is necessary for identification. Given that rank$(\mathbf{\Pi})\leq \min \{\text{rank}(E[\mathbf{S}^{\top}\mathbf{w}_{0;i}\mathbf{w}_{0;i}^{\top}\mathbf{S}]^{-1}),\text{rank}(E[\mathbf{S}^{\top}\mathbf{w}_{0;i}\mathbf{w}_{i}^{\top}\mathbf{S}])\}$, a necessary condition for rank$(\mathbf{\Pi})=k+1$ is that $\text{rank}(E[\mathbf{S}^{\top}\mathbf{w}_{0;i}\mathbf{w}_{i}^{\top}\mathbf{S}])=k+1$, which would be equivalent to the \emph{relevance} condition in the classical Instrumental Variable literature. Similarly, as is the case in this literature, equations \eqref{E2} and \eqref{E3} create a natural relationship between $\mathbf{\Pi}$ and structural parameters that is later used for estimation. 

As in Section \ref{background}, the condition that $E[\mathbf{S}^{\top}\mathbf{w}_{0;i}\mathbf{w}_{0;i}^{\top}\mathbf{S}]$ is positive definite imposes restrictions on the product matrix $\widetilde{\mathbf{W}}\equiv\mathbf{W}_{0}\mathbf{W}$ for a large enough sample size. Note that if the matrix $\widetilde{\mathbf{W}}=\mathbf{O}$ does not have full rank, this rank condition fails. Thus, identification requires the networks $G$ and $G_{0}$ to be somewhat correlated so that the rank condition imposed by Assumption \ref{A2} holds in the population. This correlation condition effectively requires the existence of edges that overlap in both layers and inter-layer intransitive triads. The second part of Assumption \ref{A2} is likely to hold when $\mathbf{W}_{0}$ is right-stochastic since $0<\lambda_{\textup{max}}<1$ is satisfied in this case. Given the reduced form relation in \eqref{E3}, a reduced form for $\mathbf{y}$ can be constructed by substituting equations \eqref{E3} into \eqref{E2}.

\begin{equation}
\label{E4}
\mathbf{y}=\alpha\boldsymbol{\iota}+\mathbf{W}_{0}\mathbf{S}\mathbf{\Pi}\boldsymbol{\theta}+\mathbf{X}\boldsymbol{\gamma}+\mathbf{e}
\text{,}%
\end{equation}

\noindent where $\mathbf{e}\equiv\mathbf{U}\boldsymbol{\theta}+\mathbf{v}$. Note that in \eqref{E4},  $E[\mathbf{S}^{\top}\mathbf{W}_{0}\mathbf{e}]\neq \mathbf{0}$ because of the simultaneity of $\mathbf{W}_{0}\mathbf{y}$ that still persists. To see that, note $E[\mathbf{S}^{\top}\mathbf{W}_{0}\mathbf{e}]=E[\mathbf{S}^{\top}\mathbf{W}_{0}\mathbf{U}\boldsymbol{\theta}]+E[\mathbf{S}^{\top}\mathbf{W}_{0}\mathbf{v}]=E[\mathbf{S}^{\top}\mathbf{W}_{0}E_{\mathbf{S},\mathbf{W}_{0}}[\mathbf{U}]\boldsymbol{\theta}]+E[\mathbf{S}^{\top}\mathbf{W}_{0}\mathbf{v}] =E[\mathbf{S}^{\top}\mathbf{W}_{0}\mathbf{v}]$, where $E[\mathbf{W}_{0}\mathbf{X}^{\top}\mathbf{v}]=\mathbf{0}$ by Assumption \ref{A1}; however, $E[\mathbf{W}_{0}\mathbf{y}^{\top}\mathbf{v}]\neq \mathbf{0}$, which implies $E[\mathbf{S}^{\top}\mathbf{W}_{0}\mathbf{e}]\neq\mathbf{0}$. Therefore, finding the reduced form for $\mathbf{y}$ requires decomposing the matrix $\mathbf{S}$ and the vector $\boldsymbol{\theta}$. Equation \eqref{E4} can be written as

\begin{equation}
\label{E5}
\mathbf{y}[\mathbf{I}-(\boldsymbol{\pi}_{1}^{\top}\boldsymbol{\theta})\mathbf{W}_{0}]=\alpha\boldsymbol{\iota}+[\gamma_1 \mathbf{I}+ (\boldsymbol{\pi}^{\top}_{2}\boldsymbol{\theta})\mathbf{W}_{0}] \mathbf{x}_1+\dots+[\gamma_k \mathbf{I}+ (\boldsymbol{\pi}^{\top}_{k+1}\boldsymbol{\theta})\mathbf{W}_{0}] \mathbf{x}_k+\mathbf{e}
\text{,}%
\end{equation}

\noindent where $\boldsymbol{\pi}_{j}$ is a $(k+1) \times 1$ vector containing the $j$th row of the coefficient matrix $\mathbf{\Pi}$, such that $\boldsymbol{\pi}^{\top}_{j}\boldsymbol{\theta}$ is a scalar representing the inner product of the two vectors of parameters. The parameter $\gamma_j$ is the $j$th entry of the $k \times 1$ vector $\boldsymbol{\gamma}$, the matrix $\mathbf{I}$ represents the $n \times n$ identity matrix, and $\mathbf{x}_l$ is an $n \times 1$ vector containing the $j$th column of the matrix $\mathbf{X}$. The second part of Assumption \ref{A2} implies that $\mathbf{I}-(\boldsymbol{\pi}_{1}^{\top}\boldsymbol{\theta})\mathbf{W}_{0}$ is invertible. Then, the reduced form for $\mathbf{y} $ is given by:
\begin{equation}
\label{E6}
\mathbf{y}=[\mathbf{I}-(\boldsymbol{\pi}^{\top}_{1}\boldsymbol{\theta})\mathbf{W}_{0}]^{-1}\{\alpha\boldsymbol{\iota}+[\gamma_1 \mathbf{I}+ (\boldsymbol{\pi}^{\top}_{2}\boldsymbol{\theta})\mathbf{W}_{0}] \mathbf{x}_1+\dots+[\gamma_k \mathbf{I}+ (\boldsymbol{\pi}^{\top}_{k+1}\boldsymbol{\theta})\mathbf{W}_{0}] \mathbf{x}_k+\mathbf{e}\}
\text{,}%
\end{equation}

\noindent where now $E[\mathbf{x}_{k}^{\top}\mathbf{e}]=E[\mathbf{x}_{k}^{\top}\mathbf{U}\boldsymbol{\theta}]+E[\mathbf{x}_{k}^{\top}\mathbf{v}]=0$ and $E[\mathbf{x}_{k}^{\top}\mathbf{W}_{0}\mathbf{e}]=E[\mathbf{x}_{k}\mathbf{W}_{0}^{\top}\mathbf{U}\boldsymbol{\theta}]+E[\mathbf{x}_{k}^{\top}\mathbf{W}_{0}^{\top}\mathbf{v}]=0$ for all $k$. The second part of Assumption \ref{A2} also implies that \eqref{E6} can be expressed as

\begin{align}
\label{E7}
    \mathbf{y}=&\sum_{r=0}^{\infty}(\boldsymbol{\pi}^{\top}_{1}\boldsymbol{\theta})^{r}\mathbf{W}_{0}^{r}\left\{\alpha\boldsymbol{\iota}+[\gamma_1 \mathbf{I}+ (\boldsymbol{\pi}^{\top}_{2}\boldsymbol{\theta})\mathbf{W}_{0}] \mathbf{x}_1+\dots+[\gamma_k \mathbf{I}+ (\boldsymbol{\pi}^{\top}_{k+1}\boldsymbol{\theta})\mathbf{W}_{0}] \mathbf{x}_k+\mathbf{e}\right\} \\
\label{E8}
    =& [\mathbf{I}-(\boldsymbol{\pi}^{\top}_{1}\boldsymbol{\theta})\mathbf{W}_{0}]^{-1}\alpha \boldsymbol{\iota}+\gamma_1 \mathbf{x}_1+[\gamma_1(\boldsymbol{\pi}^{\top}_{1}\boldsymbol{\theta})+\boldsymbol{\pi}^{\top}_{2}\boldsymbol{\theta}]\sum_{r=0}^{\infty}(\boldsymbol{\pi}^{\top}_{1}\boldsymbol{\theta})^{r}\mathbf{W}_{0}^{r+1}\mathbf{x}_1+\dots+\gamma_k \mathbf{x}_k \nonumber \\ 
    &+[\gamma_k(\boldsymbol{\pi}^{\top}_{1}\boldsymbol{\theta})+\boldsymbol{\pi}^{\top}_{k+1}\boldsymbol{\theta}]\sum_{r=0}^{\infty}(\boldsymbol{\pi}^{\top}_{1}\boldsymbol{\theta})^{r}\mathbf{W}_{0}^{r+1}\mathbf{x}_k+\sum_{r=0}^{\infty}(\boldsymbol{\pi}^{\top}_{1}\boldsymbol{\theta})^{r}\mathbf{W}_{0}^{r}\mathbf{e}\text{.}
\end{align}

When $k=1$, we have $\boldsymbol{\theta}=(\beta,\delta)^{\top}$, $\boldsymbol{\gamma}=\gamma$, and $\mathbf{\Pi}$ reduces to a $2 \times 2$ matrix. For that particular case, \eqref{E8} reduces to 

\begin{align}
\label{E9}
\mathbf{y}=&\alpha/(1-\pi_{1,1}\beta-\pi_{1,2}\delta) \boldsymbol{\iota}+[\gamma(\pi_{1,1}\beta+\pi_{1,2}\delta)+\pi_{2,1}\beta+\pi_{2,2}\delta]\sum_{r=0}^{\infty}(\pi_{1,1}\beta+\pi_{1,2}\delta)^{r}\mathbf{W}_{0}^{r+1}\mathbf{x} \nonumber \\ 
&+\gamma \mathbf{x} +\sum_{r=0}^{\infty}(\pi_{1,1}\beta+\pi_{1,2}\delta)^{r}\mathbf{W}_{0}^{r}\mathbf{e}\text{.}
\end{align}

\noindent From \eqref{E9}, note that $\gamma (\pi_{1,1}\beta+\pi_{1,2}\delta)+\pi_{2,1}\beta+\pi_{2,2}\delta \neq 0$ is a relevant condition to ensure that $\mathbf{W}_{0}^{2}\mathbf{X}$ is a valid instrument for $\mathbf{W}_{0}\mathbf{y}$ in the case when $k=1$. For the general case of $k$ covariates, the condition is generalized to $\beta(\gamma_k\pi_{1,1}+\pi_{k,1})+\sum_{l=1}^{k}\delta_l(\gamma_l\pi_{1,l+1}+\pi_{k,l+1}) \neq 0$ for all $k$. This is now formalized in the following main identification result:

\begin{theorem}
\label{T1}
Let Assumptions \ref{A1}, \ref{A2} hold and $\beta(\gamma_k\pi_{1,1}+\pi_{k,1})+\sum_{l=1}^{k}\delta_l(\gamma_l\pi_{1,l+1}+\pi_{k,l+1}) \neq 0$ for all $k$. If the matrices $\mathbf{I}$, $\mathbf{W}_{0}$, and $\mathbf{W}_{0}^{2}$ are linearly independent, then the parameters $\alpha$, $\beta$, $\boldsymbol{\gamma}$, and $\boldsymbol{\delta}$ in \eqref{E1} are identified.
\end{theorem}

Assumptions \ref{A1} and \ref{A2} do not require the layers in the multiplex networks to be undirected or unweighted. Thus, the results in Theorem \ref{T1} apply to the general case of potentially directed and weighted networks. Also, Theorem \ref{T1} is based on the existence of a set of regressors for which Assumption \ref{A1} applies, in particular, it requires the regressors in $\mathbf{x}_{i}$ to be exogenous. Note that only one such regressor is necessary to identify the peer effects parameter $\beta$. It is possible to control for other observable characteristics that are not orthogonal to the errors, as long as they are uncorrelated with the set of exogenous regressors.

The condition $\beta(\gamma_k\pi_{1,1}+\pi_{k,1})+\sum_{l=1}^{k}\delta_l(\gamma_l\pi_{1,l+1}+\pi_{k,l+1}) \neq 0$ on the structural parameters has a straightforward interpretation. From equation \eqref{E8} note that this condition involves the composed coefficients associated with the variables $\mathbf{W}_{0}^{r+1}\mathbf{x}_{l}$ for $r \in [0, \infty)$ and $l \in\left\{1,\cdots,k\right\}$. This variable can be interpreted as the reduced form social effects generated by the $l$th observable characteristic. Thus, this restriction on the parameters can be interpreted as a condition guaranteeing that the peer and contextual effects in the network space generated by $\mathbf{W}_{0}$ do not cancel each other out. The linear independence condition guaranties enough exclusion restrictions in the set of simultaneous equations generated by the linear model in \eqref{E4}. This condition can be numerically verified from the observed adjacency matrix $\mathbf{W}_{0}$. Section \ref{corr_effects} in the supplemental materials provides an extension of model \eqref{E1} that weakens the requirement that $\mathbf{W}_0$ is completely predetermined.

\subsection{Network Exogeneity Failure}\label{corr_effects}

In situations where it is difficult to find an adjacency matrix that is completely predetermined and, therefore, not correlated with the unobserved characteristics in the outcome equation, consider the following extension of model \eqref{E1}.

\begin{align}
\label{E25}
&\mathbf{y}=\alpha\boldsymbol{\iota}+\beta\mathbf{Wy}+\mathbf{WX}\boldsymbol{\delta}+\mathbf{X}\boldsymbol{\gamma}+\mathbf{v},\\ \nonumber
&\mathbf{v}=\boldsymbol{\lambda}_{0}+\mathbf{v}^{+}
\text{,}%
\end{align}

\noindent where $\boldsymbol{\lambda}_{0}$ represents the correlated effects in the exogenous network $\mathbf{W}_{0}$, the structure of $\boldsymbol{\lambda}_{0}$ is such that it allows for unobserved heterogeneity that is common for all individuals in the same group but varies across individuals in different groups. This specification explicitly accommodates scenarios where group assignment is non-random and correlated with unobservables—such as selection into schools or training programs based on latent ability—a challenge recently highlighted by \cite{Sheng_et_al_2025} in the context of endogenous group formation. The model in \eqref{E25} contains an important implicit assumption, i.e., the relevant characteristics that determine the structure of the exogenous network are common among individuals in the same group. 

The relevant assumption for the identification of the transformed model in \eqref{E25} is given by condition $E_{\mathbf{X},\mathbf{W}_{0},\boldsymbol{\lambda}_{0}}[\mathbf{v}^{+}]=\mathbf{0}$. However, $E_{\mathbf{X},\mathbf{W}_{0},\mathbf{W}}[\lambda_0]$ is allowed to be any function of the conditioning random variables. Note that the error $\mathbf{v}^{+}$ is allowed to contain unobserved heterogeneity that is correlated with $\mathbf{W}$ in any form. However, the correlation of $\mathbf{v}^{+}$ with $\mathbf{W}_{0}$ is only allowed through $\boldsymbol{\lambda}_{0}$. Applying group differences to equation \eqref{E25} eliminates the network-specific unobserved heterogeneity for the exogenous network. After applying the transformation, the structural model becomes

\begin{equation}
\label{E26}
(\mathbf{I}-\mathbf{W}_{0})\mathbf{y}=\beta(\mathbf{I}-\mathbf{W}_{0})\mathbf{W}\mathbf{y}+(\mathbf{I}-\mathbf{W}_{0})\mathbf{X}\boldsymbol{\gamma}+(\mathbf{I}-\mathbf{W}_{0})\mathbf{W}\mathbf{X}\boldsymbol{\delta}+(\mathbf{I}-\mathbf{W}_{0})\mathbf{v}^{+}.
\end{equation}

\noindent Analogously to the classic IV estimation in the panel data literature, equation \eqref{E3} is transformed following the same approach as in equation \eqref{E26} to obtain the expression

\begin{equation}
\label{E3s}
(\mathbf{I}-\mathbf{W}_{0})\mathbf{WS}  =(\mathbf{I}-\mathbf{W}_{0})\mathbf{W}_{0}\mathbf{S\Pi}+(\mathbf{I}-\mathbf{W}_{0})\mathbf{U}\text{.}
\end{equation}

\noindent Using the notation introduced earlier and considering the transformed model in \eqref{E3s}, the structural equation in \eqref{E26} can be re-written as

\begin{align}
\label{E27}
(\mathbf{I}-\mathbf{W}_{0})\mathbf{y}&=(\mathbf{I}-\mathbf{W}_{0})\mathbf{W}\mathbf{S}\boldsymbol{\theta}+(\mathbf{I}-\mathbf{W}_{0})\mathbf{X}\boldsymbol{\gamma}+(\mathbf{I}-\mathbf{W}_{0})\mathbf{v}^{+} \\
\label{E28}
&=(\mathbf{I}-\mathbf{W}_{0})\mathbf{W}_{0}\mathbf{S}\mathbf{\Pi}\boldsymbol{\theta}+(\mathbf{I}-\mathbf{W}_{0})\mathbf{X}\boldsymbol{\gamma}+(\mathbf{I}-\mathbf{W}_{0})\mathbf{e}^{+},
\end{align}

\noindent where $\mathbf{e}^{+}=\mathbf{U}\boldsymbol{\theta}+\mathbf{v}^{+}$. The next result establishes the identification in this setting through an application of Theorem \ref{T1} to the model described in equation \eqref{E28}.

\begin{proposition}
\label{P1}
Let $\beta(\gamma_k\pi_{1,1}+\pi_{k,1})+\sum_{l=1}^{k}\delta_l(\gamma_l\pi_{1,l+1}+\pi_{k,l+1}) \neq 0$ for all $k$ in model \eqref{E28}. Then the parameters $\alpha$, $\beta$, $\boldsymbol{\gamma}$, and $\boldsymbol{\delta}$ are identified if and only if the matrices $\mathbf{I}$, $\mathbf{W}_{0}$, $\mathbf{W}_{0}^{2}$, and $\mathbf{W}_{0}^{3}$ are linearly independent.
\end{proposition}

\subsection{Identification Failure}

In practice, even in situations where one has access to an exogenous network $\mathbf{W}_{0}$, identification can still fail. Specifically, identification requires $\widetilde{\mathbf{W}}\equiv\mathbf{W}_{0}\mathbf{W}\neq\mathbf{O}$, and this is guaranteed if there is an overlap between the two network layers through the existence of inter-layer intransitive triads. However, if $\mathbf{W}_0$ is extremely sparse, the number of these triads may be insufficient, leading to a problem where the rank of $\mathbf{\Pi}$ in the assumption \ref{A2} is technically full, but $\mathbf{\Pi}$ is numerically close to the zero matrix. Conversely, if $\mathbf{W}_0$ is very dense (e.g., a complete graph or a block-diagonal matrix with identical group sizes), the powers $\mathbf{I}$, $\mathbf{W}_0$, and $\mathbf{W}_0^2$ may become linearly dependent, failing the requirements of Theorem \ref{T1}.

If $\mathbf{W}$ represents close friendship ties in a school, for example, while $\mathbf{W}_0$ represents associations through participation in unrelated, non-overlapping extracurricular activities, the potential lack of shared edges or paths will likely cause $\widetilde{\mathbf{W}}$ to be the zero matrix, violating Assumption \ref{A2}. On the other hand, if individuals are assigned to groups of identical size with no connections between groups, $\mathbf{W}_0$ becomes a block-diagonal matrix where $\mathbf{W}_0^2$ is a linear combination of $\mathbf{I}$ and $\mathbf{W}_0$, precluding identification \citep[see, e.g.,][]{Lee2007}. This highlights a trade-off when applying the transformation proposed in Section \ref{corr_effects}. While highly regular group-based structures naturally motivate the condition $(\mathbf{I}-\mathbf{W}_{0})\boldsymbol{\lambda}_{0}=\mathbf{0}$ by allowing $\mathbf{W}_{0}$ to act as a perfect within-group averaging operator, this extreme regularity simultaneously undermines the linear independence conditions required by Theorem \ref{T1} and Proposition \ref{P1}. Consequently, successful identification in the presence of exogenous correlated effects requires a structural balance: the network must exhibit a group-based architecture to effectively differentiate $\boldsymbol{\lambda}_{0}$, but it must also possess sufficient irregularity—such as variation in group sizes or varied within-group connectivity—to maintain the linear independence of the network powers.

\section{Generalized Three-Stage Least Squares Estimation\label{est}}
Given the point identification in Theorem \ref{T1}, this section describes a multi-step procedure to estimate the parameters of interest in \eqref{E1}. An important feature of this proposed estimator is that it is already implemented in \texttt{Stata}  \citep[see, e.g.,][]{netivreg_stata_journal}. First, notice that \eqref{E4} can be re-written as

\begin{equation}
\label{En}
\mathbf{y}=\alpha\boldsymbol{\iota}+\mathbf{W}_{0}\mathbf{S}\boldsymbol{\theta}^{\ast}+\mathbf{X}\boldsymbol{\gamma}+\mathbf{e}
\text{,}%
\end{equation}

\noindent where $\boldsymbol{\theta}^{\ast}=\mathbf{\Pi}\boldsymbol{\theta}$. Note that \eqref{En} cannot be directly estimated due to the simultaneity in $\mathbf{W}_{0}\mathbf{y}$. Therefore, the idea of the G2SLS estimator of \cite{Kelejian1998,Kelejian_Prucha_1999_ER} is extended here to propose a Generalized Three-Stage least squares estimator (G3SLS), similar to \cite{BD:2015}, for the structural parameters $(\alpha,\boldsymbol{\gamma}^{\top},\boldsymbol{\theta}^{\top})$ as follows:\\

\noindent\underline{Step One:}

\noindent One starts by estimating the projection coefficient in equation \eqref{E3} by (system) Ordinary Least Squares, i.e.,

\begin{equation}
\label{E17}
    \widehat{\mathbf{\Pi}}=(\widehat{\boldsymbol{\pi}}_1,\widehat{\boldsymbol{\pi}}_2,\ldots,\widehat{\boldsymbol{\pi}}_{k+1})^{\top}=(\mathbf{S}^{\top}\mathbf{W}_{0}^{2}\mathbf{S})^{-1}\mathbf{S}^{\top}\widetilde{\mathbf{W}}\mathbf{S}\text{,}
\end{equation}

\noindent where $\widehat{\boldsymbol{\pi}}_1^{\top}=(\mathbf{S}^{\top}\mathbf{W}_{0}^{2}\mathbf{S})^{-1}\mathbf{S}^{\top}\widetilde{\mathbf{W}}\mathbf{y}$, and each $\widehat{\boldsymbol{\pi}}_j^{\top}=(\mathbf{S}^{\top}\mathbf{W}_{0}^{2}\mathbf{S})^{-1}\mathbf{S}^{\top}\widetilde{\mathbf{W}}\mathbf{x}_{j}$ for $j=1, \ldots ,k$. Lemma \ref{l1} in \ref{Appendix_C} states its consistency under the Assumptions listed in \ref{Appendix_B}. In this step, one saves the matrices of estimated coefficients, $\widehat{\mathbf{\Pi}}$, and residuals, $\widehat{\mathbf{U}}=\mathbf{WS}-\mathbf{W}_{0}\mathbf{S}\widehat{\mathbf{\Pi}}$. Matrices $\mathbf{W}_{0}^{2}$ and $\widetilde{\mathbf{W}}\equiv\mathbf{W}_{0}\mathbf{W}$ have a socioeconomic interpretation as explained in Section \ref{background}. The matrix $\mathbf{W}_{0}^{2}$ contains the   number of connections that individual $i$ has on the main diagonal, and the number of individuals that separate individuals $i$ and $j$ on the off-diagonal elements. Similarly, $\widetilde{\mathbf{W}}$ contains the number of connections that individual $i$ shares in both networks on the main diagonal, while the off-diagonal elements contain individual $i$'s number of connections from $\mathbf{W}_{0}$ that are connected with individual $j$ in $\mathbf{W}$ for all $i\neq j$. Hence, the matrix $\widetilde{\mathbf{W}}$ represents a measure of association between the endogenous adjacency matrix $\mathbf{W}$ and the exogenous matrix $\mathbf{W}_{0}$.\\

\noindent\underline{Step Two:}

\noindent Rewrite \eqref{En} as

\begin{equation}
\label{En2}
\mathbf{y}=\mathbf{D}_{0}\boldsymbol{\psi}^{\ast}+\mathbf{e},
\end{equation}

\noindent where $\mathbf{D}_{0}=[\mathbf{\iota},\mathbf{X},\mathbf{W}_{0}\mathbf{y},\mathbf{W}_{0}\mathbf{X}]$ is a $n\times (2k+2)$ matrix, $\boldsymbol{\psi}^{\ast}=(\alpha,\boldsymbol{\gamma}^{\top},\boldsymbol{\theta}^{\ast\top})^{\top}$ is a $(2k+2) \times 1$ vector of parameters, and $\mathbf{e}=\mathbf{U}\boldsymbol{\theta}+\mathbf{v}$. From equation \eqref{En2}, the set of parameters $\boldsymbol{\psi}^{\ast}$ can be estimated using 2SLS. Let $\mathbf{Z}=[\mathbf{\iota},\mathbf{X},\mathbf{W}^{2}_{0}\mathbf{X},\mathbf{W}_{0}\mathbf{X}]$ be the matrix of instruments for the variables in $\mathbf{D}_{0}$. For the general case where $k>1$, the model is overidentified and the standard two-stage least squares (2SLS) estimator is given by:

\begin{equation}
\label{E19}
\widehat{\boldsymbol{\psi}}^{\ast}_{\text{2SLS}}=(\mathbf{D}_{0}^{\top}\mathbf{Z}(\mathbf{Z}^{\top}\mathbf{Z})^{-1}\mathbf{Z}^{\top}\mathbf{D}_{0})^{-1} \mathbf{D}_{0}^{\top}\mathbf{Z}(\mathbf{Z}^{\top}\mathbf{Z})^{-1}\mathbf{Z}^{\top}\mathbf{y}.
\end{equation}

\noindent Under Assumptions in \ref{Appendix_B}, Lemma \ref{l2} in \ref{Appendix_C} establishes its consistency. Taking into account the estimator $\widehat{\boldsymbol{\psi}}^{\ast}_{\text{2SLS}}=(\widehat{\alpha}_{\text{2SLS}},\widehat{\boldsymbol{\gamma}}_{\text{2SLS}}^{\top},\widehat{\boldsymbol{\theta}}^{\ast\top}_{\text{2SLS}})^{\top}$, an estimator for the structural parameters of the social effects can be recovered from the relation $\boldsymbol{\theta}^{\ast}=\mathbf{\Pi}\boldsymbol{\theta}$, i.e., using $\widehat{\boldsymbol{\Pi}}$ from the first step, a consistent estimator for the social effects in equation \eqref{E1} is given by $\widehat{\boldsymbol{\theta}}=\widehat{\mathbf{\Pi}}^{-1}\widehat{\boldsymbol{\theta}}^{\ast}_{\text{2SLS}}$ -- Lemma \ref{l3} in \ref{Appendix_C} proves its consistency under the assumptions listed in \ref{Appendix_B} and its performance in our Monte Carlo designs below is reported in \ref{Appendix_D} of the supplement. In this step, the vector of the estimated coefficients is stored $\widehat{\alpha}_{\text{2SLS}}$, $\widehat{\boldsymbol{\gamma}}_{\text{2SLS}}$, and $\widehat{\boldsymbol{\theta}}$.\\

\noindent\underline{Step Three:}

\noindent This last step consists of constructing the optimal instrument for the regressor $\mathbf{W}_{0}\mathbf{y}$ in \eqref{En}, i.e., the optimal instrument (in the classical sense) is given by $E_{\mathbf{X},\mathbf{W}_0}[\mathbf{W}_0\mathbf{y}]$. Note that from the reduced-form equation in \eqref{E5} one has:

\begin{align}
\label{E20}
&E_{\mathbf{X},\mathbf{W}_0}[\mathbf{W}_0\mathbf{y}](\boldsymbol{\psi},\mathbf{\Pi})\\
&=\mathbf{W}_0[\mathbf{I}-(\boldsymbol{\pi}_{1}^{\top}\boldsymbol{\theta})\mathbf{W}_{0}]^{-1}\left\{\alpha\boldsymbol{\iota}+[\gamma_1 \mathbf{I}+ (\boldsymbol{\pi}_{2}^{\top}\boldsymbol{\theta})\mathbf{W}_{0}] \mathbf{x}_1+\dots+[\gamma_k \mathbf{I}+ (\boldsymbol{\pi}_{k+1}^{\top}\boldsymbol{\theta})\mathbf{W}_{0}]\mathbf{x}_{k}\right\}.\nonumber
\end{align}

\noindent A valid estimator of \eqref{E20} is then given by  $E_{\mathbf{X},\mathbf{W}_0}(\widehat{\boldsymbol{\psi}},\widehat{\mathbf{\Pi}})$, where $\widehat{\boldsymbol{\psi}}=(\widehat{\alpha}_{\text{2SLS}},\widehat{\boldsymbol{\gamma}}_{\text{2SLS}}^{\top},\widehat{\boldsymbol{\theta}}^{\top})^{\top}$ from the second step. Now rewrite \eqref{E1} as

\begin{equation}
\label{En3}
\mathbf{y}=\mathbf{D}\boldsymbol{\psi}+\mathbf{v}\text{,}
\end{equation}

\noindent where $\mathbf{D} = [\boldsymbol{\iota},\mathbf{X},\mathbf{W}\mathbf{S}]$ and $\boldsymbol{\psi}=(\alpha,\boldsymbol{\gamma}^{\top},\boldsymbol{\theta}^{\top})^{\top}$. Using $\widehat{\boldsymbol{\Pi}}$ from the first step, let the matrix $\widehat{\mathbf{D}}$ be defined as $\widehat{\mathbf{D}}=[\boldsymbol{\iota},\mathbf{X},\mathbf{W}_{0}\mathbf{S}\widehat{\mathbf{\Pi}}]=[\boldsymbol{\iota},\mathbf{X},\mathbf{W}_{0}\mathbf{S}]\widehat{\mathbf{\Gamma}}=[\mathbf{\iota},\mathbf{X},\mathbf{W}_{0}\mathbf{y},\mathbf{W}_{0}\mathbf{X}]\widehat{\mathbf{\Gamma}}=\mathbf{D}_{0}\widehat{\mathbf{\Gamma}}$, where

\begin{equation}
\label{E18}
\widehat{\mathbf{\Gamma}}=	\begin{bmatrix}

\mathbf{I}_{k+1} & \mathbf{O}_{k+1}  \\
\mathbf{O}_{k+1} & \widehat{\mathbf{\Pi}}  \\
\end{bmatrix},
\end{equation}

\noindent is a $(2k+2) \times (2k+2)$ matrix, $\mathbf{O}_{k+1}$ is a $(k+1)\times(k+1)$ matrix of zeros, and $\mathbf{I}_{k+1}$ represents the identity matrix of order $k+1$. Let $\widetilde{\mathbf{Z}}^{\ast}$ be the matrix of optimal instruments for $\mathbf{D}_{0}\mathbf{\Gamma}$, where $\widetilde{\mathbf{Z}}^{\ast}=\mathbf{Z}^{\ast}\mathbf{\Gamma}$, and $\mathbf{Z}^{\ast}=[\boldsymbol{\iota},\mathbf{X},E_{\mathbf{X},\mathbf{W}_0}[\mathbf{W}_0\mathbf{y}](\boldsymbol{\psi},\mathbf{\Pi}),\mathbf{W}_{0}\mathbf{X}]$. Analogous to the previous transformation, by additionally considering the estimator of the conditional expectation of $\mathbf{W}_{0}\mathbf{y}$ in the second step, an estimator of the optimal matrix of instruments is then given by

\begin{equation}
\label{EW0y}
\widehat{\widetilde{\mathbf{Z}}}^{\ast}=\widehat{\mathbf{Z}}^{\ast}\widehat{\mathbf{\Gamma}}.
\end{equation}

\noindent where $\widehat{\mathbf{Z}}^{\ast}=[\boldsymbol{\iota},\mathbf{X},E_{\mathbf{X},\mathbf{W}_0}[\mathbf{W}_0\mathbf{y}](\widehat{\boldsymbol{\psi}},\widehat{\mathbf{\Pi}}),\mathbf{W}_{0}\mathbf{X}]$. Therefore, the G3SLS in this just-identified case is

\begin{equation}
\label{E21}
\widehat{\boldsymbol{\psi}}_{\text{G3SLS}}=(\widehat{\widetilde{\mathbf{Z}}}{}^{\ast\top}\widehat{\mathbf{D}})^{-1}\widehat{\widetilde{\mathbf{Z}}}{}^{\ast\top}\mathbf{y}.
\end{equation}

\noindent Finally, from this step, one saves the resulting vector of estimated coefficients $\widehat{\boldsymbol{\psi}}_{\text{G3SLS}}$ and the residuals $\widehat{\mathbf{v}}=\mathbf{y}-\mathbf{D}\widehat{\boldsymbol{\psi}}_{\text{G3SLS}}$.\\

After defining $\mathbf{M}_{\mathbf{W}_0}\equiv\mathbf{I}_{n}-\mathbf{W_{0}}\mathbf{S}(\mathbf{S}^{\top}\mathbf{W}_{0}^{2}\mathbf{S})^{-1}\mathbf{S}^{\top}\mathbf{W_{0}}$, the following theorem establishes the asymptotic normality of the proposed G3SLS estimator.

\begin{theorem}
\label{dist}
Let Assumptions \ref{A1}, \ref{A2}, and \ref{D1}--\ref{D8} in \ref{Appendix_B} hold, then $\widehat{\boldsymbol{\psi}}_{\text{\emph{G3SLS}}}=\boldsymbol{\psi}+o_p(1)$ and $n^{1/2}(\widehat{\boldsymbol{\psi}}_{\text{\emph{G3SLS}}}-\boldsymbol{\psi})\overset{d}{\longrightarrow}N(\mathbf{0},\mathbf{V}_{\boldsymbol{\psi}})$, where
\[
\mathbf{V}_{\boldsymbol{\psi}}=(\mathbf{\Gamma}^{\top}\mathbf{Q}_{\mathbf{Z}^{\ast}\mathbf{D}_{0}}\mathbf{\Gamma})^{-1}\boldsymbol{\Omega}(\mathbf{\Gamma}\mathbf{Q}^{\top}_{\mathbf{Z}^{\ast}\mathbf{D}_{0}}\mathbf{\Gamma}^{\top})^{-1},
\]
$\mathbf{Q}_{\mathbf{Z}^{\ast}\mathbf{D}_{0}}=\lim_{n\to\infty}n^{-1}\mathbf{Z}^{\ast\top}\mathbf{D}_{0}$, $\boldsymbol{\Omega}=\lim_{n\to\infty}n^{-1}\sum_{i=1}^{n}\widetilde{\mathbf{z}}_{i}^{\ast}\widetilde{\mathbf{z}}_{i}^{\ast\top}E_{\widetilde{\mathbf{Z}^{\ast}}}[e_{i}^{\ast2}]$, $\widetilde{\mathbf{z}}_{i}^{\ast}$ represents the $i$th row of $\widetilde{\mathbf{Z}}^{\ast}$, and $e_{i}^{\ast}$ denotes the $i$th element of the $n\times 1$ vector $\mathbf{e}^{\ast}=\mathbf{M}_{\mathbf{W}_0}\mathbf{U}\boldsymbol{\theta}+\mathbf{v}$.
\end{theorem}

\noindent The asymptotic variance-covariance matrix, $\mathbf{V}_{\boldsymbol{\psi}}$, takes into account the estimation effects of the various steps, i.e., the error $\mathbf{e}^{\ast}$ associated with the estimator for $\boldsymbol{\psi}$ is composed of structural errors $\mathbf{v}$ and also errors in the first step $\mathbf{U}$. Note that $\mathbf{V}_{\boldsymbol{\psi}}$ is a $(2k+2) \times (2k+2)$ matrix containing the variances and covariances of all the structural coefficients in equation \eqref{E1}. To obtain the variance-covariance matrix of the parameters of the social effects $\boldsymbol{\theta}$, we need to extract the appropriate sub-component $(k+1) \times (k+1)$ from $\mathbf{V_{\boldsymbol{\psi}}}$. Let $\mathbf{J}$ be a $(k+1) \times (2k+2)$ matrix such that $\mathbf{J}=[\mathbf{O}_{k+1},\mathbf{I}_{k+1}]$. Then, the variance-covariance matrix of the social parameters of interest is given by

\begin{equation}
\mathbf{V}_{\boldsymbol{\theta}}=\mathbf{J}\mathbf{V_{\boldsymbol{\psi}}}\mathbf{J}^{\top}.\label{V_theta}
\end{equation}

\subsubsection*{Standard Errors Calculation}

\noindent A consistent estimator for $\mathbf{V_{\boldsymbol{\psi}}}$ can be calculated as

\begin{equation}
\label{E23}
\widehat{\mathbf{V}}_{\boldsymbol{\psi}}=(n^{-1}\widehat{\widetilde{\mathbf{Z}}}{}^{\ast\top}\widehat{\mathbf{D}})^{-1} (n^{-1}\sum\nolimits_{i=1}^{n}\widehat{\widetilde{\mathbf{z}}}{}^{\ast}_{i}\widehat{\widetilde{\mathbf{z}}}{}^{\ast\top}_{i}\widehat{e}_{i}^{\ast 2}) (n^{-1}\widehat{\mathbf{D}}^{\top}\widehat{\widetilde{\mathbf{Z}}}{}^{\ast})^{-1},
\end{equation}

\noindent where $\widehat{\mathbf{e}}^{\ast}=\mathbf{M}_{\mathbf{W}_0}\widehat{\mathbf{U}}\widehat{\boldsymbol{\theta}}+\widehat{\mathbf{v}}$. The residuals $\widehat{\mathbf{U}}$ are obtained from the first step, $\widehat{\boldsymbol{\theta}}$, and the residuals $\widehat{\mathbf{v}}$ are taken from the last step above. Therefore, a consistent estimator of \eqref{V_theta} is simply given by $\widehat{\mathbf{V}}_{\boldsymbol{\theta}}=\mathbf{J}\widehat{\mathbf{V}}_{\boldsymbol{\psi}}\mathbf{J}^{\top}$. The standard errors for the coefficients of interest are then the squared root of the main diagonal elements of this matrix after dividing them by $n$.

Note that, as with standard instrumental variable methods, the finite-sample moments of the proposed G3SLS estimator here may not exist under certain conditions. As established in the simultaneous equations literature \citep[see, e.g.,][]{mariano2001simultaneous}, the existence of their finite moments typically depends on the degree of over-identification, specifically, the difference between the number of valid instruments used and the number of endogenous regressors. When this degree of over-identification is low, the estimator may lack finite moments of higher orders.

\section{Monte Carlo Experiments\label{mc}}
In this section, we showcase the proposed estimator's strong performance and versatility across three distinct data-generating processes (DGPs). The endogeneity in these DGPs is generated by a classic omitted variable (Design 1), measurement error in the connections (Design 2), and an unobserved homophily with the simultaneous determination of network formation and outcomes (Design 3). Our Design 3 corresponds to the unobserved characteristics with homophily scenario (Design 1) presented in \cite{Chan_et_al_social_effects}. Similarly, our Design 2 is a modified version of the misclassified links scenario (Design 2) in that same paper, originally from \cite{Lewbel_Qu_Tang}. We include a detailed description of these scenarios here for completeness and to highlight their relevance in social science research. A total of 1,500 data sets $\left\{y_i,x_i,\{w_{i,j}\}_{j=1,j\neq i}^n,\{w_{0;i,j}\}_{j=1,j\neq i}^n\right\}_{i=1}^n$ with $n\in\left\{50,100,200\right\}$ are generated from \eqref{E1} by setting $k=1$, $\beta=0.7$, $\alpha=\delta=\gamma=1$, and drawing $\left\{x_i\right\}_{i=1}^n$ as a random sample from a normal distribution with mean zero and variance 3. The other data components are constructed as follows:

\subsubsection*{Design 1: Unobserved Heterogeneity\label{d1}}

As in \cite{Johnsson2019}, the network formation process follows \cite{Graham2017}, i.e., links are formed according to the rule $w_{i,j}(\psi)\equiv\mathbb{I}[z_{i}z_{j} + \psi(a_{i}+a_{j}) - u_{i,j} \geq 0]$, where $\mathbb{I}(\cdot)$ denotes the indicator function that equals one if its argument is true and zero otherwise. The scalar parameter $\psi$ acts as a switch to activate or deactivate the individual degree heterogeneity in the network formation process, so that the endogenous adjacency matrix is $\mathbf{W} = [w_{i,j}(1)]$, while the exogenous matrix is given by $\mathbf{W}_{0} = [w_{i,j}(0)]$. The exogenous variable $z$ takes on values -1 and 1 with a probability of 0.5 (these values imply a strong taste for homophilic matching; see \citealt{Graham2017}), $u_{ij}$ is drawn from a logistic distribution with mean zero and scale parameter 1, and $a_{i}=\alpha_{L} \mathbb{I}[z_{i}=-1]+\alpha_{H} \mathbb{I}[z_{i}=1]+\xi_{i}$, where $\alpha_{L}=-3/2$ and $\alpha_{H}=1$ are both parameters controlling the extent to which the degree heterogeneity $a_i$ is correlated with the observable exogenous characteristic $z_i$. The error $\xi_i$ is drawn from a re-centered Beta($1/4$,$3/4$) distribution (implying the often found right-skewed degree distribution in empirical applications, \citealp[see, e.g., ][]{Johnsson2019}). On the other hand, individual outcomes are constructed as in \eqref{E1}, where $\mathbf{v}=m\times \boldsymbol{a}+\boldsymbol{\varepsilon}_{0}$ and $m\in\{10,12\}$ measure the importance of degree heterogeneity; i.e., the higher it is, the higher the level of endogeneity will be. The error vector $\boldsymbol{\varepsilon}_{0}$ is drawn from a multivariate standard normal distribution independently of everything else. 

\subsubsection*{Design 2: Misclassified Links\label{d2}}

This is an adapted iteration of the Monte Carlo design originally presented by  \cite{Lewbel_Qu_Tang}. As in \cite{Chan_et_al_social_effects}, the true DGP includes a hidden adjacency matrix $\mathbf{W}_{0}^{\ast}=[w_{0;i,j}^{\ast}]$, which is derived from a typical random network model by \citeauthor{Erdos1959}'s \citeyearpar{Erdos1959}, featuring a density of 0.05 across a network of size $n$. However, it is assumed that the researcher only has access to an adjacency matrix $\mathbf{W}=[w_{i,j}]$ where links are randomly misclassified. This misclassification is represented as $w_{i, j}=w_{0;i,j}^{\ast}e_{1;i,j}+(1-w_{0;i,j}^{\ast}) e_{2;i,j}$ for $i \neq j$. Additionally, there is an exogenous adjacency matrix $\mathbf{W}_{0}=[w_{0;i,j}]$ with $w_{0;i,j}=w_{0;i,j}^{*}b_{1;i,j}+(1-w_{0;i,j}^{*}) b_{2;i,j}$ for $i \neq j$. The variables $e_{1;i,j}$, $e_{2;i,j}$, $b_{1;i,j}$, and $b_{2;i,j}$ are independently drawn Bernoulli random variables from one another for all $i \neq j$, with probabilities set at 0.75, 0, $1-\tau$, and 0.002, respectively. The design parameter $\tau\in\left\{0.01,0.05\right\}$ dictates the likelihood of misclassification in $\mathbf{W}_0$. It is crucial to note, consistent with \cite{Lewbel_Qu_Tang}, that non-existent links in $\mathbf{W}$ are never misidentified, whereas such misclassification is allowed in $\mathbf{W}_0$ but at a meager chance of 0.2\%. Nevertheless, this setup directly links the vector of individual outcomes to the proportion of misclassification in $\mathbf{W}$ for each $i$. The $n\times 1$ outcome vector $\mathbf{y}$ is constructed according to equation \eqref{E1}, where $\mathbf{v}=\boldsymbol{\varepsilon}_{1}+\boldsymbol{\varepsilon}_{2}$, $\varepsilon_{1,i}=\frac{1}{n}\sum_{j=1}^{n}w_{0;i,j}^{*}e_{1;i,j}$, and $\boldsymbol{\varepsilon}_{2}$ originate from a multivariate standard normal distribution that is independent of all other elements. 

\subsubsection*{Design 3: Unobserved Characteristics with Homophily\label{d3}}
 
Within this framework, the outcome variable for individual $i$, denoted as $y_i$, alongside connections $\left\{w_{i,j}\right\}_{j=1,j\neq i}^n$, are simultaneously determined by a shared idiosyncratic unobserved feature associated with homophily, $\varepsilon_{3;i}^{\ast}$. Initially, an exogenous adjacency matrix $\mathbf{W}_{0}=[w_{0;i,j}]$ is constructed from a \citeauthor{Erdos1959}'s \citeyearpar{Erdos1959} random graph featuring a density of 0.01, along with an $n\times 1$ vector $\boldsymbol{\varepsilon}_{3}^{\ast}=[\varepsilon_{3;1}^{\ast},\ldots,\varepsilon_{3;n}^{\ast}]^\top$ drawn from a multivariate standard normal distribution. Subsequently, the elements of the endogenous adjacency matrix $\mathbf{W}=[w_{i,j}]$ are computed as
\begin{equation*}
w_{ij}=
\begin{cases}
\mathbb{I}[|\varepsilon_{3;i}^{\ast}-\varepsilon_{3;j}^{\ast}|<\widehat{F}_{\varepsilon_3^\ast}^{-1}(0.99)]\times (1-w_{0;i,j}) + w_{0;i,j} & \text{; if $\varepsilon_{3;i}^{\ast}>\Phi^{-1}(0.99)$,} \\
\mathbb{I}[|\varepsilon_{3;i}^{\ast}-\varepsilon_{3;j}^{\ast}|<\widehat{F}_{\varepsilon_3^\ast}^{-1}(0.99)] \times w_{0;i,j} & \text{; if $\varepsilon_{3;i}^{\ast}<\Phi^{-1}(0.01)$,} \\
w_{0;i,j} & \text{; otherwise},
\end{cases}
\end{equation*}

\noindent where $\widehat{F}_{\varepsilon_3^{\ast}}^{-1}(0.99)$ denotes the 99\% empirical quantile of the components of the vector $\boldsymbol{\varepsilon}_3^{\ast}$, where $\varepsilon_{3;k}^{\ast}$ signifies its $k$th component, and $\Phi^{-1}(\cdot)$ is the inverse operation of the cumulative distribution function for a standard normal variable. This framework encapsulates the concept of homophily, suggesting that agents with higher values of $\varepsilon_{3}$ are inclined to forge or sustain relationships with others who also possess high $\varepsilon_{3}$ values, while disconnecting from those with lower values of this unique, unobservable trait. The $n\times 1$ outcome vector, $\boldsymbol{y}$, is derived from \eqref{E1} by defining $\mathbf{v}=m \times \boldsymbol{\varepsilon}_{3}+\boldsymbol{\varepsilon_{4}}$, where $m$ is selected from $\left\{1,3\right\}$, $\boldsymbol{\varepsilon_{4}}$ is sampled from a multivariate standard normal distribution, and the elements of $\boldsymbol{\varepsilon}_{3}$ are specified as

\begin{equation*}
\varepsilon_{3;i}=
\begin{cases}
\varepsilon_{3;i}^{\ast} & \text{; if $\varepsilon_{3;i}^{\ast}<\Phi^{-1}(0.01)$ or $\varepsilon_{3;i}^{\ast}>\Phi^{-1}(0.99)$,} \\
0 & \text{; otherwise}.
\end{cases}
\end{equation*}

\subsubsection*{Results\label{MC_results}}

\noindent Tables \ref{tab:design1_heterogeneity_main}--\ref{tab:design3_homophily_main} report the Monte Carlo results for the three estimators across the alternative network designs. In addition to the proposed G3SLS estimator in \eqref{E21}, we evaluate the standard Ordinary Least Squares (OLS) estimator -- which ignores the network endogeneity -- and the Generalized Two-Stage Least Squares (G2SLS) estimator -- which sets $\mathbf{W}=\mathbf{W}_0$ in our proposed estimator. All adjacency matrices are row-normalized prior to estimation, following \cite{liu2014}.

Tables \ref{tab:design1_heterogeneity_main}--\ref{tab:design3_homophily_main}  summarize the finite-sample performance of the estimators under each design. For the peer ($\beta$), contextual ($\delta$), and direct ($\gamma$) effects in \eqref{E1}, we report bias, standard deviation (SD), root mean squared error (RMSE), and inter-quartile range (IQR) across alternative values of the design parameters $m$ or $\tau$ and different sample sizes. Several regularities emerge from the Monte Carlo evidence when focusing on bias and RMSE. Across the three structural parameters — peer, contextual, and direct effects — the proposed G3SLS estimator systematically delivers lower bias relative to OLS and G2SLS, particularly as the sample size increases.  The improvements in bias directly translate into lower RMSE. For all three parameters, the RMSE of G3SLS declines steadily with larger samples and remains uniformly below that of OLS and G2SLS across designs. By contrast, the higher bias observed under naive OLS and conventional G2SLS contributes to persistently larger RMSE values, a pattern consistent with the link misclassification mechanisms discussed in \cite{Chandrasekhar_unpub_2016}. Overall, the Monte Carlo evidence indicates that G3SLS has larger Monte Carlo variance than both OLS and G2SLS in DGPs 1 across all sample sizes and parameters, and for the peer effect parameter in DGPs 2 and 3 across all sample sizes. This might be attributed to the potential lack of finite higher moments in these designs, rather than the lack of consistency; i.e., the G3SLS achieves a more accurate estimation of all social interaction parameters, as reflected in both the lower bias and the systematically smaller RMSE.

Finally, to assess the practical value of the third estimation step in Section \ref{est}, we compare the finite-sample properties of the efficient G3SLS estimator against the intermediate 2SLS estimator defined in equation \eqref{E19} in Step 2. Detailed results are provided in Appendix \ref{2SLS_MC} of the Supplementary  Material (Tables \ref{tab:design1_heterogeneity_2sls}-\ref{tab:design3_homophily_2sls}). The analysis confirms that while the 2SLS estimator is consistent, the G3SLS estimator yields substantial efficiency gains, particularly in designs with significant unobserved heterogeneity, thereby justifying the additional estimation step.

\begin{landscape}
   \begin{table}[!htbp]
\centering
\small
\begin{threeparttable}
\caption{Estimator Performance under Unobserved Degree Heterogeneity (Design 1)}
\label{tab:design1_heterogeneity_main}
\begin{tabular}{ccc c cccc c cccc c cccc}
\toprule
 &  &  &  &\multicolumn{4}{c}{Peer effects} &  &\multicolumn{4}{c}{Contextual effects} &  &\multicolumn{4}{c}{Direct effects} \\
\cmidrule(lr){5-8} \cmidrule(lr){10-13} \cmidrule(lr){15-18}
$m$ & Estimator & $n$ &  & Bias & SD & RMSE & IQR &  & Bias & SD & RMSE & IQR &  & Bias & SD & RMSE & IQR \\
\midrule
\multirow{9}{*}{10} & \multirow{3}{*}{OLS} & 50 &  & 0.359 & 0.046 & 0.361 & 0.072 &  & -1.119 & 0.875 & 1.420 & 1.357 &  & -0.449 & 0.616 & 0.762 & 0.983 \\
 &  & 100 &  & 0.352 & 0.031 & 0.353 & 0.049 &  & -1.104 & 0.576 & 1.245 & 0.897 &  & -0.441 & 0.417 & 0.607 & 0.643 \\
 &  & 200 &  & 0.348 & 0.022 & 0.349 & 0.033 &  & -1.110 & 0.404 & 1.181 & 0.633 &  & -0.472 & 0.287 & 0.552 & 0.442 \\
\noalign{\vskip 0.5em}
 & \multirow{3}{*}{G2SLS} & 50 &  & 0.691 & 0.317 & 0.760 & 0.306 &  & -1.612 & 2.233 & 2.754 & 2.917 &  & -0.588 & 1.087 & 1.235 & 1.502 \\
 &  & 100 &  & 0.742 & 0.268 & 0.789 & 0.278 &  & -2.143 & 1.786 & 2.789 & 2.500 &  & -0.816 & 0.811 & 1.150 & 1.158 \\
 &  & 200 &  & 0.810 & 0.272 & 0.854 & 0.289 &  & -2.532 & 1.551 & 2.969 & 2.203 &  & -1.034 & 0.714 & 1.256 & 0.957 \\
\noalign{\vskip 0.5em}
 & \multirow{3}{*}{G3SLS} & 50 &  & 0.074 & 0.986 & 0.989 & 1.283 &  & 0.250 & 6.687 & 6.689 & 8.261 &  & 0.500 & 1.594 & 1.670 & 2.372 \\
 &  & 100 &  & 0.027 & 0.773 & 0.773 & 1.032 &  & 0.214 & 4.136 & 4.140 & 5.281 &  & 0.562 & 1.079 & 1.217 & 1.629 \\
 &  & 200 &  & -0.005 & 0.549 & 0.549 & 0.791 &  & 0.245 & 2.486 & 2.497 & 3.500 &  & 0.562 & 0.718 & 0.912 & 1.124 \\
\midrule
\multirow{9}{*}{12} & \multirow{3}{*}{OLS} & 50 &  & 0.361 & 0.047 & 0.364 & 0.074 &  & -1.132 & 1.043 & 1.539 & 1.628 &  & -0.454 & 0.736 & 0.865 & 1.163 \\
 &  & 100 &  & 0.354 & 0.031 & 0.355 & 0.050 &  & -1.110 & 0.687 & 1.305 & 1.078 &  & -0.441 & 0.501 & 0.667 & 0.776 \\
 &  & 200 &  & 0.350 & 0.022 & 0.351 & 0.034 &  & -1.118 & 0.481 & 1.217 & 0.749 &  & -0.476 & 0.344 & 0.588 & 0.533 \\
\noalign{\vskip 0.5em}
 & \multirow{3}{*}{G2SLS} & 50 &  & 0.732 & 0.304 & 0.793 & 0.305 &  & -1.770 & 2.619 & 3.160 & 3.341 &  & -0.627 & 1.266 & 1.413 & 1.758 \\
 &  & 100 &  & 0.795 & 0.236 & 0.829 & 0.266 &  & -2.199 & 1.952 & 2.940 & 2.632 &  & -0.840 & 0.914 & 1.241 & 1.371 \\
 &  & 200 &  & 0.827 & 0.204 & 0.852 & 0.240 &  & -2.543 & 1.541 & 2.974 & 2.249 &  & -1.044 & 0.732 & 1.275 & 0.997 \\
\noalign{\vskip 0.5em}
 & \multirow{3}{*}{G3SLS} & 50 &  & 0.114 & 1.004 & 1.011 & 1.295 &  & -0.051 & 7.860 & 7.857 & 9.720 &  & 0.473 & 1.902 & 1.959 & 2.851 \\
 &  & 100 &  & 0.115 & 0.824 & 0.831 & 1.088 &  & -0.069 & 4.997 & 4.995 & 6.546 &  & 0.523 & 1.267 & 1.371 & 1.941 \\
 &  & 200 &  & 0.015 & 0.608 & 0.608 & 0.832 &  & 0.148 & 2.927 & 2.930 & 4.032 &  & 0.550 & 0.857 & 1.018 & 1.328 \\
\bottomrule
\end{tabular}
\begin{tablenotes}[para]
\footnotesize
\setstretch{1}
\item \textit{Notes:} This table reports Monte Carlo results under unobserved degree heterogeneity. The parameter $m$ controls the intensity of heterogeneity in node degrees and is common across individuals. Reported statistics include bias, standard deviation (SD), root mean squared error (RMSE), and inter-quantile range (IQR) for the peer effect ($\beta$=0.7), contextual ($\delta=1$), and direct ($\gamma=1$) Effects.
\end{tablenotes}
\end{threeparttable}
\normalsize
\end{table} 
\end{landscape}

\begin{landscape}
   \begin{table}[!htbp]
\centering
\small
\begin{threeparttable}
\caption{Estimator Robustness to Network Link Misclassification (Design 2)}
\label{tab:design2_misclassification_main}
\begin{tabular}{ccc c cccc c cccc c cccc}
\toprule
 &  &  &  &\multicolumn{4}{c}{Peer effects} &  &\multicolumn{4}{c}{Contextual effects} &  &\multicolumn{4}{c}{Direct effects} \\
\cmidrule(lr){5-8} \cmidrule(lr){10-13} \cmidrule(lr){15-18}
$\tau$ & Estimator & $n$ &  & Bias & SD & RMSE & IQR &  & Bias & SD & RMSE & IQR &  & Bias & SD & RMSE & IQR \\
\midrule
\multirow{9}{*}{0.01} & \multirow{3}{*}{OLS} & 50 &  & -0.387 & 0.129 & 0.408 & 0.201 &  & -0.425 & 0.314 & 0.528 & 0.497 &  & 0.411 & 0.220 & 0.466 & 0.346 \\
 &  & 100 &  & -0.525 & 0.087 & 0.533 & 0.137 &  & -0.628 & 0.153 & 0.646 & 0.244 &  & 0.203 & 0.098 & 0.225 & 0.154 \\
 &  & 200 &  & -0.555 & 0.062 & 0.559 & 0.098 &  & -0.778 & 0.091 & 0.783 & 0.149 &  & 0.092 & 0.047 & 0.103 & 0.073 \\
\noalign{\vskip 0.5em}
 & \multirow{3}{*}{G2SLS} & 50 &  & -0.538 & 0.222 & 0.582 & 0.305 &  & -0.152 & 0.469 & 0.493 & 0.718 &  & 0.475 & 0.251 & 0.538 & 0.387 \\
 &  & 100 &  & -0.649 & 0.133 & 0.663 & 0.184 &  & -0.457 & 0.215 & 0.505 & 0.326 &  & 0.237 & 0.112 & 0.262 & 0.179 \\
 &  & 200 &  & -0.688 & 0.092 & 0.695 & 0.140 &  & -0.623 & 0.122 & 0.635 & 0.191 &  & 0.114 & 0.051 & 0.125 & 0.080 \\
\noalign{\vskip 0.5em}
 & \multirow{3}{*}{G3SLS} & 50 &  & 0.739 & 0.466 & 0.873 & 0.668 &  & 2.406 & 1.961 & 3.103 & 2.783 &  & 0.051 & 0.092 & 0.105 & 0.138 \\
 &  & 100 &  & 0.249 & 0.188 & 0.312 & 0.294 &  & 0.995 & 0.648 & 1.187 & 0.976 &  & 0.023 & 0.059 & 0.063 & 0.092 \\
 &  & 200 &  & 0.043 & 0.175 & 0.180 & 0.274 &  & 0.235 & 0.337 & 0.411 & 0.535 &  & 0.007 & 0.038 & 0.039 & 0.059 \\
\midrule
\multirow{9}{*}{0.05} & \multirow{3}{*}{OLS} & 50 &  & -0.387 & 0.129 & 0.408 & 0.201 &  & -0.425 & 0.314 & 0.528 & 0.497 &  & 0.411 & 0.220 & 0.466 & 0.346 \\
 &  & 100 &  & -0.525 & 0.087 & 0.533 & 0.137 &  & -0.628 & 0.153 & 0.646 & 0.244 &  & 0.203 & 0.098 & 0.225 & 0.154 \\
 &  & 200 &  & -0.555 & 0.062 & 0.559 & 0.098 &  & -0.778 & 0.091 & 0.783 & 0.149 &  & 0.092 & 0.047 & 0.103 & 0.073 \\
\noalign{\vskip 0.5em}
 & \multirow{3}{*}{G2SLS} & 50 &  & -0.538 & 0.222 & 0.582 & 0.305 &  & -0.152 & 0.469 & 0.493 & 0.718 &  & 0.475 & 0.251 & 0.538 & 0.387 \\
 &  & 100 &  & -0.649 & 0.133 & 0.663 & 0.184 &  & -0.457 & 0.215 & 0.505 & 0.326 &  & 0.237 & 0.112 & 0.262 & 0.179 \\
 &  & 200 &  & -0.688 & 0.092 & 0.695 & 0.140 &  & -0.623 & 0.122 & 0.635 & 0.191 &  & 0.114 & 0.051 & 0.125 & 0.080 \\
\noalign{\vskip 0.5em}
 & \multirow{3}{*}{G3SLS} & 50 &  & 0.639 & 0.491 & 0.805 & 0.715 &  & 2.604 & 2.007 & 3.287 & 2.918 &  & 0.083 & 0.114 & 0.141 & 0.174 \\
 &  & 100 &  & 0.187 & 0.206 & 0.279 & 0.317 &  & 1.077 & 0.685 & 1.276 & 1.047 &  & 0.038 & 0.067 & 0.077 & 0.106 \\
 &  & 200 &  & 0.006 & 0.183 & 0.183 & 0.296 &  & 0.283 & 0.350 & 0.450 & 0.545 &  & 0.013 & 0.040 & 0.042 & 0.062 \\
\bottomrule
\end{tabular}
\begin{tablenotes}[para]
\footnotesize
\setstretch{1}
\item \textit{Notes:} This table reports Monte Carlo results under network link misclassification. The parameter $\tau$ controls the intensity of misclassification and is common across individuals. Reported statistics include bias, standard deviation (SD), root mean squared error (RMSE), and inter-quantile range (IQR) for the peer effect ($\beta$=0.7), contextual ($\delta=1$), and direct ($\gamma=1$) Effects.
\end{tablenotes}
\end{threeparttable}
\normalsize
\end{table} 
\end{landscape}

\begin{landscape}
   \begin{table}[!htbp]
\centering
\small
\begin{threeparttable}
\caption{Estimator Performance under Unobserved Homophily (Design 3)}
\label{tab:design3_homophily_main}
\begin{tabular}{ccc c cccc c cccc c cccc}
\toprule
 &  &  &  &\multicolumn{4}{c}{Peer effects} &  &\multicolumn{4}{c}{Contextual effects} &  &\multicolumn{4}{c}{Direct effects} \\
\cmidrule(lr){5-8} \cmidrule(lr){10-13} \cmidrule(lr){15-18}
$m$ & Estimator & $n$ &  & Bias & SD & RMSE & IQR &  & Bias & SD & RMSE & IQR &  & Bias & SD & RMSE & IQR \\
\midrule
\multirow{9}{*}{1} & \multirow{3}{*}{OLS} & 50 &  & 0.077 & 0.033 & 0.084 & 0.051 &  & -0.201 & 0.119 & 0.234 & 0.189 &  & -0.109 & 0.082 & 0.137 & 0.129 \\
 &  & 100 &  & 0.071 & 0.022 & 0.074 & 0.035 &  & -0.188 & 0.079 & 0.204 & 0.127 &  & -0.107 & 0.056 & 0.121 & 0.091 \\
 &  & 200 &  & 0.080 & 0.019 & 0.082 & 0.032 &  & -0.181 & 0.062 & 0.191 & 0.099 &  & -0.088 & 0.037 & 0.096 & 0.060 \\
\noalign{\vskip 0.5em}
 & \multirow{3}{*}{G2SLS} & 50 &  & 0.019 & 0.038 & 0.043 & 0.060 &  & -0.051 & 0.127 & 0.137 & 0.195 &  & -0.024 & 0.088 & 0.091 & 0.141 \\
 &  & 100 &  & 0.010 & 0.026 & 0.028 & 0.040 &  & -0.026 & 0.084 & 0.088 & 0.139 &  & -0.013 & 0.060 & 0.061 & 0.094 \\
 &  & 200 &  & 0.015 & 0.022 & 0.027 & 0.036 &  & -0.034 & 0.065 & 0.073 & 0.104 &  & -0.018 & 0.039 & 0.043 & 0.062 \\
\noalign{\vskip 0.5em}
 & \multirow{3}{*}{G3SLS} & 50 &  & -0.013 & 0.049 & 0.051 & 0.074 &  & 0.035 & 0.149 & 0.153 & 0.231 &  & 0.014 & 0.100 & 0.101 & 0.160 \\
 &  & 100 &  & -0.006 & 0.030 & 0.030 & 0.046 &  & 0.018 & 0.095 & 0.096 & 0.149 &  & 0.006 & 0.066 & 0.066 & 0.103 \\
 &  & 200 &  & -0.006 & 0.026 & 0.027 & 0.043 &  & 0.014 & 0.072 & 0.074 & 0.113 &  & 0.002 & 0.043 & 0.043 & 0.067 \\
\midrule
\multirow{9}{*}{3} & \multirow{3}{*}{OLS} & 50 &  & 0.118 & 0.040 & 0.124 & 0.064 &  & -0.304 & 0.152 & 0.340 & 0.236 &  & -0.168 & 0.124 & 0.209 & 0.198 \\
 &  & 100 &  & 0.098 & 0.026 & 0.101 & 0.040 &  & -0.258 & 0.095 & 0.275 & 0.149 &  & -0.149 & 0.077 & 0.168 & 0.118 \\
 &  & 200 &  & 0.113 & 0.023 & 0.115 & 0.037 &  & -0.256 & 0.075 & 0.267 & 0.122 &  & -0.124 & 0.052 & 0.135 & 0.081 \\
\noalign{\vskip 0.5em}
 & \multirow{3}{*}{G2SLS} & 50 &  & 0.056 & 0.050 & 0.075 & 0.078 &  & -0.143 & 0.174 & 0.225 & 0.265 &  & -0.078 & 0.137 & 0.157 & 0.209 \\
 &  & 100 &  & 0.032 & 0.034 & 0.046 & 0.055 &  & -0.082 & 0.110 & 0.137 & 0.171 &  & -0.047 & 0.086 & 0.098 & 0.135 \\
 &  & 200 &  & 0.048 & 0.027 & 0.055 & 0.041 &  & -0.109 & 0.080 & 0.135 & 0.123 &  & -0.053 & 0.056 & 0.077 & 0.091 \\
\noalign{\vskip 0.5em}
 & \multirow{3}{*}{G3SLS} & 50 &  & -0.019 & 0.072 & 0.074 & 0.103 &  & 0.051 & 0.215 & 0.221 & 0.324 &  & 0.026 & 0.141 & 0.143 & 0.222 \\
 &  & 100 &  & -0.007 & 0.039 & 0.040 & 0.062 &  & 0.024 & 0.128 & 0.131 & 0.202 &  & 0.007 & 0.083 & 0.083 & 0.130 \\
 &  & 200 &  & -0.007 & 0.034 & 0.035 & 0.054 &  & 0.017 & 0.094 & 0.096 & 0.142 &  & 0.003 & 0.057 & 0.057 & 0.091 \\
\bottomrule
\end{tabular}
\begin{tablenotes}[para]
\footnotesize
\setstretch{1}
\item \textit{Notes:} This table reports Monte Carlo results under unobserved homophily. The parameter $m$ controls the intensity of homophilous link formation and is common across individuals. Reported statistics include bias, standard deviation (SD), root mean squared error (RMSE), and inter-quantile range (IQR) for the peer effect ($\beta$=0.7), contextual ($\delta=1$), and direct ($\gamma=1$) Effects.
\end{tablenotes}
\end{threeparttable}
\normalsize
\end{table} 
\end{landscape}

\section{Application to Publication Outcomes in Economics\label{emp}}
\subsection{Background}
The availability of online scientific research repositories has resulted in a stable source of data to uncover scholars' professional connections. Additionally, when linked with scholars' biographical public information, other types of professional connections beyond observed Co-authorship can be uncovered \citep{Colussi2018}. However, a major challenge when performing causal inference is that co-authorship connections are inherently correlated with publication outcomes, such as citation counts. However, where scholars completed their graduate training, although related to the co-authorship network, does not a priori directly affect publication outcomes. This setting, therefore, directly lends itself to the use of the proposed estimator to calculate social effects in research citations.

\subsection{Empirical Example: Peer-effects in Publishing}

\noindent To illustrate the proposed method, a data set  of 1,628 peer-reviewed articles published between 2000-2006 in the top-four general-interest journals in Economics is used here to fit the following specific case of the estimating equation in \eqref{intro}:

\begin{equation}
\label{est_eq}
  y_{i,r,t}=\alpha+\beta\sum_{j \neq i}w_{i,j,t}y_{j,r,t}+\sum_{j \neq i}w_{i,j, t}\widetilde{\mathbf{x}}_{j, r,t}^{\top}\boldsymbol{\delta}+ \mathbf{x}_{i, r, t}^{\top}\boldsymbol{\gamma}+ \lambda_{r} + \lambda_{t} + \lambda_{0} + v_{i,r,t}\text{,}
\end{equation}

\noindent The dependent variable $y_{i,r,t}$ denotes the natural logarithm of the total citations received by the article $i$ in the journal $r$ within eight years after its publication. The term $w_{i,j,t}$ corresponds to the $(i,j)$ element of the adjacency matrix $\mathbf{W}$ that represents the \emph{Co-authorship} network at time $t$. A detailed description of this data set and the corresponding summary statistics can be found in \cite{netivreg_g3sls} and \cite{EstradaDingRosales2025}.

The vector of controls $\mathbf{x}_{i,r,t}$ includes indicators for whether current or former editors of the journal $r$ at time $t$ appear among the authors of the article $i$ (\texttt{Editor}) and for whether the co-authors of article $i$ are of different genders (\texttt{Different Gender}), which is set to zero for single-author papers. Additional article-level covariates comprise the total number of pages (\texttt{Number of Pages}), authors (\texttt{Number of Authors}), and bibliographic references (\texttt{Number of References}), as well as a dummy variable identifying articles that are isolated within the network (\texttt{Isolated}). Contextual or peer effects are computed only for the \texttt{Editor} and \texttt{Different Gender} variables, summarized in $\widetilde{\mathbf{x}}_{j,r,t}$. The specification also includes fixed effects for the journal ($\lambda_{r}$) and year ($\lambda_{t}$).  

The structural error term $v_{i,r,t}$ satisfies $E_{\mathbf{X},\mathbf{W}}[\mathbf{v}] \neq 0$, reflecting potential endogeneity in the \emph{Co-authorship} network. To address this, the \emph{Alumni} network is employed as the exogenous matrix $\mathbf{W}_{0}$ in the identification theorem (\ref{T1}). Additionally, to mitigate possible deviations from the exogeneity condition, the model incorporates the corresponding \emph{Alumni-based} components ($\lambda_{0}$), as specified in equation (\ref{E25}).

The model \eqref{est_eq} is estimated in a rolling-regression setting for $t=2002$, $2003$, $2004$, $2005$, and $2006$; i.e., the estimation sample each year includes those from previous years. The results for 2000 and 2001 are not included because they suffer from degrees-of-freedom problems given the specification \eqref{est_eq}. Table \ref{empt3} summarizes the results using the proposed G3SLS estimator, while Tables \ref{empt4} and \ref{empt5} in \ref{Appendix_E} report the OLS and G2SLS results, respectively. In all cases, the asymptotic standard errors are clustered at the corresponding network component level. This is a natural way of clustering when utilizing network data because each component corresponds to a portion of the network that is disconnected from the others, allowing for articles within each component to be correlated but not between disconnected components.

The results show that the estimated peer effect ($\widehat{\beta}$) is consistently positive and becomes statistically significant at the 5\% level from 2004 onward, peaking at 0.676 in 2005. This confirms the presence of significant positive citation spillovers among articles connected within the \emph{Co-authorship} network. Turning to direct effect estimates ($\widehat{\boldsymbol{\gamma}}$), the indicator of gender-diverse research teams exhibits a stable, positive, and statistically significant impact on citation counts in all accumulated samples. This result demonstrates that gender diversity within research teams positively influences the academic impact of a paper, an outcome that aligns with recent evidence \cite{RHJS:2024} showing that gender-diverse teams also produce more readable and accessible research. Furthermore, standard article characteristics such as the number of pages and bibliographic references remain robust, positive predictors of citations, while isolated articles consistently incur a significant citation penalty. In contrast, the estimated contextual effects ($\widehat{\boldsymbol{\delta}}$) related to the presence of the editorial and the diversity of the gender are statistically indistinguishable from zero in all specifications. Finally, empirical evidence supporting the underlying exclusion restriction and relevance conditions necessary for identification is provided in \ref{Appendix_E} in the Supplementary Material.

\begin{table}
\centering
\small
\caption{Estimation Results for Social and Direct Effects}

\begin{tabular}{lccccc}
\hline
& \multicolumn{5}{c}{\emph{Co-author Network}}  \\

\cline{2-6} 

                   & 2002 & 2003 & 2004 & 2005 & 2006   \\
\hline
     Peer Effects ($\widehat{\beta}$) &     0.520     &    0.476     &    0.542** &    0.676** &     0.570**   \\
                   &  (0.362)       &  (0.316)      &  (0.261)      &  (0.283)       &  (0.242)   \\ \hline
 Contextual Effects ($\boldsymbol{\widehat{\delta}}$) & & & & &   \\
 \hspace{0.2cm} \texttt{Editor} &     1.790      &   -0.971       &    -2.910      &   -5.355       &   -5.044      \\
                   &  (5.363)       &  (3.587)      &  (3.111)      &  (4.135)       &  (4.188)    \\
     \hspace{0.2cm} \texttt{Different Gender}  &  -2.096      &   -1.661       &   -1.922       &   -2.967      &   -1.559        \\
                   &  (2.829)      &  (1.793)      &  (1.838)     &   (2.65)     &  (1.737)      \\ \hline
                   
Direct Effects ($\boldsymbol{\widehat{\gamma}}$) & & & &  \\
    \hspace{0.2cm} \texttt{Editor} &    0.173      &    0.027      &   -0.037       &    0.025      &  0.057  \\
                   &  (0.116)       &  (0.124)       &  (0.131)      &  (0.116)     &   (0.130)       \\
    \hspace{0.2cm} \texttt{Different Gender} &    0.219* &    0.205* &    0.221** &    0.166* &    0.143*  \\
                   &  (0.131)       &  (0.108)       &  (0.092)       &  (0.085)      &   (0.080)    \\
    \hspace{0.2cm} \texttt{Number of Pages} &    0.029*** &    0.027*** &    0.023*** &    0.019*** &    0.018*** \\
                   &  (0.004)      &  (0.004)      &  (0.004)     &  (0.004)      &  (0.004)   \\
    \hspace{0.2cm} \texttt{Number of Authors} &    0.072       &    0.087* &    0.072       &    0.097*** &    0.076** \\
                   &   (0.060)      &   (0.050)      &  (0.044)     &  (0.038)       &  (0.031)   \\
    \hspace{0.2cm} \texttt{Number of References} &    0.012*** &    0.012*** &    0.011*** &    0.011*** &    0.012*** \\
                   &  (0.003)       &  (0.002)      &  (0.002)    &  (0.002)    &  (0.001)    \\
    \hspace{0.2cm} \texttt{Isolated} &   -0.223** &   -0.236*** &   -0.353*** &   -0.399*** &   -0.407***   \\
                   &  (0.132)      &   (0.110)      &  (0.103)      &  (0.092)       &  (0.088)      \\
                   
 \hline

                 $n$ &      729 &          961 &         1187 &         1412 &      1628  \\
                 $R^{2}$ &    0.172     &    0.199      &    0.178      &    0.125      &    0.148   \\
\hline
\end{tabular}

\vspace{0.4cm}

\begin{minipage}{1\textwidth}
\footnotesize
Note: Standard errors are in parentheses and are clustered at the specific network's components. Stars follow the key: * $p$ $<$ 0.10, ** $p$ $<$ 0.05, and *** $p$ $<$ 0.01, where $p$ stands for $p$-values. $R^2$ are calculated as the squared of the sample correlation coefficients between the observed outcomes and their fitted values. All specifications include indicator variables for Journal, Year and Alumni Network Components. 
\end{minipage}
\label{empt3}
\end{table}

\section{Discussion\label{conclusion}}
\noindent In this paper, we propose a novel and computationally simple way to identify and consistently estimate social parameters in a linear-in-means model with endogenous network formation. Identification can be achieved by the inclusion of the multiplex network data structure, where at least one of the layers can be assumed to be exogenous (potentially pre-determined), and the multiplex structure is such that the layers correlate with each other. Our research shows that peer and contextual effects can be uniquely recovered from an estimating sample. Unlike current alternatives that require smoothing techniques and/or Bayesian methods, the resulting estimator is simple to compute, and it is already implemented in a user-written \texttt{Stata} command, see \cite{netivreg_stata_journal}. The asymptotic normality of the proposed multi-step estimator is established, and a consistent estimator of its asymptotic variance-covariance matrix is proposed for performing inference.

The type of endogeneity allowed in our framework is general enough to encompass settings with measurement error, sample selection, and correlation between the unobservables driving network formation and outcomes in a linear-in-means model. It is argued that the full observability of a multiplex data structure, with at least one of the layers being exogenous, is not a limiting data requirement that can easily be constructed in some cases based purely on characteristics the researcher usually observes. With this in mind, an empirical application is presented where the tools of web scraping and text mining are used to construct a data set consisting of all peer-reviewed research articles published in 4 of economics' top general-interest journals between 2000 and 2006. Using publicly available information on where the authors of these publications obtained their Ph.D. degrees from, an \emph{Alumni} network is constructed and argued to be pre-determined yet correlated with the observed \emph{Co-authorship} ties among these scholars. The results show the existence of positive peer effects in terms of citations among peer-reviewed research articles connected through co-authorship connections of their authors, as well as significant positive effects of research teams that are gender diverse on the quality of a paper measured in terms of citation outcomes similar to \cite{RHJS:2024}. 

The results in our paper have the potential to be extended in different directions. From a theoretical perspective, this paper introduces the concept of multidimensional networks and uses, in particular, the structure of multiplex networks. Although the former is an active research topic in other fields, this paper proposes its usage for causal inference in the Social Sciences. An interesting extension would be to study the dynamics of peer effects using the setting proposed in this paper. Similarly, as with all IV-based estimation procedures, issues pertaining to weak or invalid instruments, many-instrument problems, or small-sample performance of the proposed estimator are valid concerns, but they are beyond the scope of this paper and therefore left for future research.

\section{Acknowledgements}
We are grateful to the Co-Editor, the Associate Editor, and two anonymous referees for their constructive comments and suggestions that significantly improved this manuscript. We also thank the discussants and participants at various conferences, workshops, and seminars where earlier versions of this research were presented for their valuable feedback. Portions of this paper are based on results from Juan Estrada's doctoral dissertation, \emph{Causal Inference in Multilayered Networks}, at Emory University. Kim P. Huynh dedicates this paper in the memory of Tony S. Wirjanto.

\appendix
\setcounter{section}{0}
\renewcommand\thesection{Appendix \Alph{section}}
\renewcommand{\theequation}{\Alph{section}.\arabic{equation}}

\section{Proofs of Main Results\label{Appendix_A}}
\begin{proof}[Proof of Theorem \ref{T1}]
First note that by Assumption \ref{A2}, the matrix of reduced form coefficients in \eqref{E3} is given by $\mathbf{\Pi}=E[\mathbf{S}^{\top}\mathbf{w}_{0;i}\mathbf{w}_{0;i}^{\top}\mathbf{S}]^{-1}E[\mathbf{S}^{\top}\mathbf{w}_{i;0}\mathbf{w}_{i}^{\top}\mathbf{S}]$. Similar arguments and Assumption \ref{A1} imply that the reduced--form coefficients in \eqref{E6} are identified. Therefore it follows that for each characteristic $k$ there is a set of structural parameters $(\beta,\gamma_{k},\delta_{k})$ associated with it. It then follows that two sets of structural parameters $(\alpha,\beta,\gamma_{k},\delta_{k})$ and $(\alpha^\prime,\beta^\prime,\gamma_{k}^\prime,\delta_{k}^\prime)$ lead to the same reduced form if

\begin{align}
\label{PE1}
&\alpha/(1-\pi_{1,1}\beta-\sum_{l=1}^{k}\pi_{1,l}\delta_{l})=\alpha^{\prime}/(1-\pi_{1,1}\beta^{\prime}-\sum_{l=1}^{k}\pi_{1,l}\delta_{l}^{\prime})\textbf{,} \\
\nonumber
&[\mathbf{I}-(\pi_{1,1}\beta+\sum_{l=1}^{k}\pi_{1,l}\delta_{l})\mathbf{W_{0}}]^{-1}\left[\mathbf{I}\gamma_{k}+\left(\pi_{k,1}\beta+\sum_{l=1}^{k}\pi_{k,l}\delta_{l}\right)\mathbf{W}_{0}\right]\\ \nonumber
&=[\mathbf{I}-(\pi_{1,1}\beta^{\prime}+\sum_{l=1}^{k}\pi_{1,l}\delta^{\prime}_{l})\mathbf{W_{0}}]^{-1}\left[\mathbf{I}\gamma^{\prime}_{k}+\left(\pi_{k,1}\beta^{\prime}+\sum_{l=1}^{k}\pi_{k,l}\delta^{\prime}_{l}\right)\mathbf{W}_{0}\right]\text{,} \quad \forall k\ge 1.
\end{align}

\noindent Multiplying both sides of the last equation by  $[\mathbf{I}-(\pi_{1,1}\beta+\sum_{l=1}^{k}\pi_{1,l}\delta_{l})\mathbf{W}_0][\mathbf{I}-(\pi_{1,1}\beta^{\prime}+\sum_{l=1}^{k}\pi_{1,l}\delta^{\prime}_{l})\mathbf{W}_0]$, recalling that for any scalar $a$, $(\mathbf{I}-a\mathbf{W}_{0})^{-1}\mathbf{W}_{0}=\mathbf{W}_{0}(\mathbf{I}-a\mathbf{W}_{0})^{-1}$, and rearranging some terms, one has

\begin{align*}
    &(\gamma_{k}-\gamma_{k}^{\prime})\mathbf{I}\\
    &+\left[\pi_{1,1}\gamma_{k}^{\prime}\beta-\pi_{1,1}\gamma_{k}\beta^{\prime}+\sum_{l=1}^{k}\pi_{k,l}\delta_{l}-\sum_{l=1}^{k}\pi_{k,l}\delta_{l}^{\prime}\right]\mathbf{W}_{0}\\
    &+\left[\pi_{k,1}\beta-\pi_{k,1}\beta^{\prime}+\gamma_{k}^{\prime}\sum_{l=1}^{k}\pi_{1,l}\delta_{l}-\gamma_{k}\sum_{l=1}^{k}\pi_{1,l}\delta^{\prime}_{l}\right]\mathbf{W}_{0}\\
    &+\left[(\pi_{1,1}\beta+\sum_{l=1}^{k}\pi_{1,l}\delta_{l})(\pi_{k,1}\beta^{\prime}+\sum_{i=l}^{k}\pi_{k,l}\delta_{l}')\right]\mathbf{W}_{0}^{2}\\
    &-\left[(\pi_{1,1}\beta^{\prime}+\sum_{l=1}^{k}\mathbf\pi_{1,l}\delta^{\prime}_{l})(\pi_{k,1}\beta+\sum_{l=1}^{k}\pi_{k,l}\delta_{l})\right]\mathbf{W}_{0}^{2}=0\text{.}
\end{align*}

\noindent If the matrices $\mathbf{I}$, $\mathbf{W}_{0}$ and $\mathbf{W}_{0}^{2}$ are linearly independent, it then follows

\begin{align}
\label{PE3}
    & \gamma_{k}=\gamma_{k}'\text{,} \\
\label{PE4}
    & (\pi_{1,1}\beta+\sum_{l=1}^{k}\pi_{1,l}\delta_{l})\gamma_{k}'+\pi_{k,1}\beta+\sum_{l=1}^{k}\pi_{k,l}\delta_{l}=(\pi_{1,1}\beta'+\sum_{l=1}^{k}\pi_{1,l}\delta'_{l})\gamma_{k}+\pi_{k,1}\beta'+\sum_{l=1}^{k}\pi_{k,l}\delta'_{l}\text{,}\\
\label{PE5}
    & (\pi_{1,1}\beta+\sum_{l=1}^{k}\pi_{1,l}\delta_{k})(\pi_{k,1}\beta'+\sum_{l=1}^{k}\pi_{k,l}\delta_{l}')=(\pi_{1,1}\beta'+\sum_{l=1}^{k}\mathbf\pi_{1,l}\delta'_{l})(\pi_{k,1}\beta+\sum_{l=1}^{k}\pi_{k,l}\delta_{l})\text{.}
\end{align}

\noindent In equation \eqref{PE5}, if $(\pi_{1,1}\beta+\sum_{l=1}^{k}\pi_{1,l}\delta_{l})(\pi_{k,1}\beta'+\sum_{l=1}^{k}\pi_{k,l}\delta_{l}^{\prime})=0$, it must be the case that either $\pi_{1,1}\beta+\sum_{l=1}^{k}\pi_{1,l}\delta_{l}=0$, or $\pi_{k,1}\beta^{\prime}+\sum_{i=1}^{k}\pi_{k,i}\delta_{l}^{\prime}=0$, or both. If $\pi_{1,1}\beta+\sum_{l=1}^{k}\pi_{1,l}\delta_{l}=0$, given that rank$(\mathbf{\Pi})=k+1$ in Assumption \ref{A2}, then $\beta=\delta_{1}=\dots=\delta_{k}=0$, which would contradict $\beta(\gamma_k\pi_{1,1}+\pi_{k,1})+\sum_{l=1}^{k}\delta_l(\gamma_{k}\pi_{1,l+1}+\pi_{k,l+1}) \neq 0$. The same argument applies to $\pi_{k,1}\beta^{\prime}+\sum_{l=1}^{k}\pi_{k,l}\delta_{l}^{\prime}=0$. Thus, it must be the case that $(\pi_{1,1}\beta+\sum_{l=1}^{k}\pi_{1,l}\delta_{l})(\pi_{k,1}\beta'+\sum_{l=1}^{k}\pi_{k,l}\delta_{l}')>0$.

\noindent Since $(\pi_{1,1}\beta+\sum_{l=1}^{k}\pi_{k,l}\delta_{l})(\pi_{k,1}\beta^{\prime}+\sum_{l=1}^{k}\pi_{k,l}\delta_{l}^{\prime})>0$, there exists $\lambda>0$ such that $\beta'=\lambda\beta$ and $\delta_{l}^{\prime}=\lambda\delta_{l}$ for all $l=1,\dots,k$. It then follows from equation \eqref{PE3} and \eqref{PE4} that $(\pi_{1,1}\beta^{\prime}+\sum_{l=1}^{k}\pi_{1,l}\delta^{\prime}_{l})\gamma_{k}+\pi_{k,1}\beta^{\prime}+\sum_{l=1}^{k}\pi_{k,l}\delta^{\prime}_{l}=\lambda[(\pi_{1,1}\beta+\sum_{l=1}^{k}\pi_{1,l}\delta_{l})\gamma_{k}+\pi_{k,1}\beta+\sum_{l=1}^{k}\pi_{k,l}\delta_{l}]=(\pi_{1,1}\beta+\sum_{l=1}^{k}\pi_{1,l}\delta_{l})\gamma_{k}+\pi_{k,1}\beta+\sum_{l=1}^{k}\pi_{k,l}\delta_{l}$. Therefore, given that $\beta(\gamma_{k}\pi_{1,1}+\pi_{k,1})+\sum_{l=1}^{k}\delta_l(\gamma_{k}\pi_{1,l+1}+\pi_{k,l+1}) \neq 0$, it must be that $\lambda=1$, $\beta^{\prime}=\beta$, and $\delta_{l}^{\prime}=\delta_{l}$ for all $l=1,\dots,k$. \eqref{PE1} will then also imply that $\alpha^{\prime}=\alpha$ and this concludes the proof.
\end{proof}

\begin{proof}[Proof of Proposition \ref{P1}]
Firstly, \eqref{E28} can be further written as

\begin{align}
\label{PE16}
(\mathbf{I}-\mathbf{W}_{0})\mathbf{y}&=(\mathbf{I}-\boldsymbol{\pi}^{\top}_{1}\boldsymbol{\theta}\mathbf{W}_{0})^{-1}(\gamma_{1}\mathbf{I}+\boldsymbol{\pi}^{\top}_{2}\boldsymbol{\theta}\mathbf{W}_{0})(\mathbf{I}-\mathbf{W}_{0})\mathbf{x}_{1} +\dots  \\ \nonumber
&+(\mathbf{I}-\boldsymbol{\pi}^{\top}_{1}\boldsymbol{\theta}\mathbf{W}_{0})^{-1}(\gamma_{k}\mathbf{I}+\boldsymbol{\pi}^{\top}_{k+1}\boldsymbol{\theta}\mathbf{W}_{0})(\mathbf{I}-\mathbf{W}_{0})\mathbf{x}_{k} \\ \nonumber
&+(\mathbf{I}-\boldsymbol{\pi}^{\top}_{1}\boldsymbol{\theta}\mathbf{W}_{0})^{-1}(\mathbf{I}-\mathbf{W}_{0})\mathbf{e}^{+}\textbf{,}
\end{align}

\noindent where the $1 \times (k+1)$ -- vector $\boldsymbol{\pi}^{\top}_{l}$ represents the $l$th row of the matrix of coefficients $\mathbf{\Pi}$, and $\mathbf{x}_{l}$ represents the $l$th column of the matrix $\mathbf{X}$. Recall that from Assumption \ref{A2}, the matrix of parameters $\mathbf{\Pi}$ is identified. Therefore by the same arguments in the proof of Theorem \ref{T1}, it follows that for each characteristic $k$ there is a set of structural parameters $(\beta,\gamma_{k},\delta_{k})$ associate with it. It then follows that two set of structural parameters $(\alpha,\beta,\gamma_{k},\delta_{k})$ and $(\alpha,\beta^\prime,\gamma_{k}^\prime,\delta_{k}^\prime)$ lead to the same reduced form if

\begin{align*}
&[\mathbf{I}-(\pi_{1,1}\beta+\sum_{l=1}^{k}\pi_{1,l}\delta_{l})\mathbf{W}_{0}]^{-1}\left[\mathbf{I}\gamma_{k}+\left(\pi_{k,1}\beta+\sum_{l=1}^{k}\pi_{k,l}\delta_{l}\right)\mathbf{W}_{0}\right](\mathbf{I}-\mathbf{W}_{0})= \\
&[\mathbf{I}-(\pi_{1,1}\beta'+\sum_{l=1}^{k}\pi_{1,l}\delta'_{l})\mathbf{W}_{0}]^{-1}\left[\mathbf{I}\gamma'_{k}+\left(\pi_{k,1}\beta'+\sum_{l=1}^{k}\pi_{k,l}\delta'_{l}\right)\mathbf{W}_{0}\right](\mathbf{I}-\mathbf{W}_{0})\text{, } \quad \forall k\ge 1.
\end{align*}

\noindent This equality is equivalent to
\begin{align*}
    &(\gamma_{k}-\gamma_{k}')\mathbf{I}\\
    &+\left[\pi_{1,1}\gamma_{k}'\beta-\pi_{1,1}\gamma_{k}\beta'+\sum_{l=1}^{k}\pi_{k,l}\delta_{l}-\sum_{l=1}^{k}\pi_{k,l}\delta_{l}'+\pi_{k,1}\beta-\pi_{k,1}\beta'+\gamma_{k}'\sum_{l=1}^{k}\pi_{1,l}\delta_{l}\right.\\
    &\hphantom{++}\left.-\gamma_{k}\sum_{l=1}^{k}\pi_{1,l}\delta'_{l}+\gamma_{k}-\gamma_{k}' \right]\mathbf{W}_{0}\\
    &+\left[(\pi_{1,1}\beta+\sum_{l=1}^{k}\pi_{1,l}\delta_{k})(\pi_{k,1}\beta'+\sum_{l=1}^{k}\pi_{k,l}\delta_{l}')\right]\mathbf{W}_{0}^{2} \\
    &-\left[(\pi_{1,1}\beta'+\sum_{l=1}^{k}\mathbf\pi_{1,l}\delta'_{l})(\pi_{k,1}\beta+\sum_{l=1}^{k}\pi_{k,l}\delta_{l})+\Xi\right]\mathbf{W}_{0}^{2} \\
    &+\left[(\pi_{1,1}\beta'+\sum_{l=1}^{k}\mathbf\pi_{1,l}\delta'_{l})(\pi_{k,1}\beta+\sum_{l=1}^{k}\pi_{k,l}\delta_{l})\right]\mathbf{W}_{0}^{3}\\
    &-\left[(\pi_{1,1}\beta+\sum_{l=1}^{k}\pi_{1,l}\delta_{l})(\pi_{k,1}\beta'+\sum_{l=1}^{k}\pi_{k,l}\delta_{l}')\right]\mathbf{W}_{0}^{3},
\end{align*}

\noindent where we used the notation $\Xi\equiv \pi_{1,1}\gamma_{k}\beta-\pi_{1,1}\gamma_{k}'\beta'+\sum_{l=1}^{k}\pi_{k,l}\delta_{l}'-\sum_{l=1}^{k}\pi_{k,l}\delta_{l}+\pi_{k,1}\beta'-\pi_{k,1}\beta+\gamma_{k}\sum_{l=1}^{k}\pi_{1,l}\delta'_{l}-\gamma_{k}'\sum_{l=1}^{k}\pi_{1,l}\delta_{l}$. Since $\mathbf{I}$, $\mathbf{W}_{0}$, $\mathbf{W}_{0}^{2}$ and $\mathbf{W}_{0}^{3}$ are linearly independent, it then follows that

\begin{align*}
    & \gamma_{k}=\gamma_{k}' \\
    & (\pi_{1,1}\beta+\sum_{l=1}^{k}\pi_{1,l}\delta_{l})\gamma_{k}'+\pi_{k,1}\beta+\sum_{l=1}^{k}\pi_{k,l}\delta_{l}=(\pi_{1,1}\beta'+\sum_{l=1}^{k}\pi_{1,l}\delta'_{l})\gamma_{k}+\pi_{k,1}\beta'+\sum_{l=1}^{k}\pi_{k,l}\delta'_{l}\\
    & (\pi_{1,1}\beta+\sum_{l=1}^{k}\pi_{1,l}\delta_{l})(\pi_{k,1}\beta'+\sum_{l=1}^{k}\pi_{k,l}\delta_{l}')=(\pi_{1,1}\beta'+\sum_{l=1}^{k}\mathbf\pi_{1,l}\delta'_{l})(\pi_{k,1}\beta+\sum_{l=1}^{k}\pi_{k,l}\delta_{l}),
\end{align*}

\noindent and similar arguments as in the proof of Theorem \ref{T1} implies that $\beta'=\beta$, $\delta_{l}'=\delta_{l}$, and $\gamma_{l}'=\gamma_{l}$, for all $l=1,\dots,k$.
\end{proof}

\begin{proof}[Proof of Theorem \ref{dist}]
First, write the estimator proposed in the main text as $\widehat{\boldsymbol{\psi}}_{\text{G3SLS}}=(\widehat{\widetilde{\mathbf{Z}}}{}^{\ast\top}\widehat{\mathbf{D}})^{-1}\widehat{\widetilde{\mathbf{Z}}}{}^{\ast\top}\mathbf{y}=(\widehat{\widetilde{\mathbf{Z}}}{}^{\ast\top}\widehat{\mathbf{D}})^{-1}\widehat{\widetilde{\mathbf{Z}}}{}^{\ast\top}(\mathbf{D}\boldsymbol{\psi}+\mathbf{v})$. Then, note that $\mathbf{D}\boldsymbol{\psi}=\widehat{\mathbf{D}}\boldsymbol{\psi}+(\mathbf{D}-\widehat{\mathbf{D}})\boldsymbol{\psi}$ and rewrite

\begin{equation}
\label{CP11}
\widehat{\boldsymbol{\psi}}_{\text{G3SLS}}-\boldsymbol{\psi}=(\widehat{\widetilde{\mathbf{Z}}}{}^{\ast\top}\widehat{\mathbf{D}})^{-1}[\widehat{\widetilde{\mathbf{Z}}}{}^{\ast\top}(\mathbf{D}-\widehat{\mathbf{D}})\boldsymbol{\psi}+\mathbf{v}].
\end{equation}

\noindent Recall $\mathbf{D} = [\boldsymbol{\iota},\mathbf{X},\mathbf{W}\mathbf{S}]$ and $\widehat{\mathbf{D}}=[\boldsymbol{\iota},\mathbf{X},\widehat{\mathbf{W}\mathbf{S}}]$ so that $\mathbf{D}-\widehat{\mathbf{D}}=[\mathbf{O},\mathbf{W}\mathbf{S}-\mathbf{W}_{0}\mathbf{S}\widehat{\mathbf{\Pi}}]$, where $\mathbf{O}$ is a $n\times(k+1)$ matrix of zeros, $\mathbf{W}\mathbf{S}-\mathbf{W}_{0}\mathbf{S}\widehat{\mathbf{\Pi}}=\mathbf{M}_{\mathbf{W}_{0}}\mathbf{U}$, with $\mathbf{M}_{\mathbf{W}_0}=\mathbf{I}_{n}-\mathbf{W}_{0}\mathbf{S}(\mathbf{S}^{\top}\mathbf{W}_{0}^{2}\mathbf{S})^{-1}\mathbf{S}^{\top}\mathbf{W_{0}}$ and $\mathbf{U}$ is the error in \eqref{E3}. Let $\mathbf{e}^{\ast}=\mathbf{M}_{\mathbf{W}_{0}}\mathbf{U}\boldsymbol{\theta}+\mathbf{v}$, then \eqref{CP11} can be expressed in terms of the error $\mathbf{e}^{\ast}$ as

\begin{equation*}
\widehat{\boldsymbol{\psi}}_{\text{G3SLS}}-\boldsymbol{\psi}=(\widehat{\widetilde{\mathbf{Z}}}{}^{\ast\top}\widehat{\mathbf{D}})^{-1}\widehat{\widetilde{\mathbf{Z}}}{}^{\ast\top}\mathbf{e}^{\ast}=(\widehat{\mathbf{\Gamma}}^{\top}\widehat{\mathbf{Z}}{}^{\ast\top}\mathbf{D}_{0}\widehat{\mathbf{\Gamma}})^{-1}\widehat{\widetilde{\mathbf{Z}}}{}^{\ast\top}\mathbf{e}^{\ast}\text{.}
\end{equation*}

\noindent The result follows from analysing each term in the last equality above. First, from Lemma \ref{l3}, one has $\widehat{\mathbf{\Gamma}}=\mathbf{\Gamma}+o_{p}(1)$. Given Lemmas \ref{l1} and \ref{l2} and the continuous mapping theorem it follows $E_{\mathbf{X},\mathbf{W}_0}[\mathbf{W}_0\mathbf{y}](\widehat{\boldsymbol{\psi}},\widehat{\mathbf{\Pi}})=E_{\mathbf{X},\mathbf{W}_0}[\mathbf{W}_0\mathbf{y}](\boldsymbol{\psi},\mathbf{\Pi})+o_{p}(1)$. Therefore, $\widehat{\mathbf{Z}}^{\ast}=\mathbf{Z}^{\ast}+o_{p}(1)$ and $\widehat{\mathbf{Z}}^{\ast\top}\mathbf{D}_{0}=\mathbf{Z}^{\ast\top}\mathbf{D}_{0}+o_{p}(1)$. Consequently, from Assumption \ref{D8}, one has $n^{-1} \widehat{\mathbf{\Gamma}}^{\top}\widehat{\mathbf{Z}}^{*\top}\mathbf{D}_{0}\widehat{\mathbf{\Gamma}}\overset{p}{\longrightarrow}\mathbf{\Gamma}^{\top}\mathbf{Q}_{\mathbf{Z}^{\ast}\mathbf{D}_{0}}\mathbf{\Gamma}$.

\noindent From Lemma \ref{l4}, $n^{-1/2}\widehat{\widetilde{\mathbf{Z}}}{}^{\ast\top}\mathbf{e}^{*}=n^{-1/2}\widetilde{\mathbf{Z}}^{*\top}\mathbf{e}^{*}+o_{p}(1)$. Based on the previous results, to show consistency, it is enough to show that $n^{-1}\widetilde{\mathbf{Z}}^{*\top}\mathbf{e}^{*}\overset{p}{\longrightarrow}0$. Applying the same arguments as in the proof of Lemma \ref{l2}, it follows that  $n^{-1}\widetilde{\mathbf{Z}}^{\top}\mathbf{e}\overset{p}{\longrightarrow}0$ and $\widehat{\boldsymbol{\psi}}_{\text{G3SLS}}\overset{p}{\longrightarrow}\boldsymbol{\psi}$.

Finally, applying the same arguments as in the proof of Lemma \ref{l4} one has

\begin{equation}
    n^{-1/2}\widetilde{\mathbf{Z}}^{\ast\top}\mathbf{e}^{\ast}\ \overset{d}{\longrightarrow} N(\mathbf{0},\boldsymbol{\Omega}),
\end{equation}

\noindent where $\boldsymbol{\Omega}=\lim_{n\to\infty}n^{-1}\sum_{i=1}^{n}\widetilde{\mathbf{z}}_{i}^{\ast}\widetilde{\mathbf{z}}_{i}^{\ast\top}E_{\widetilde{\mathbf{Z}}^{\ast}}[e_{i}^{\ast 2}]=\lim_{n\to\infty}n^{-1}\sum_{i=1}^{n}\widetilde{\mathbf{z}}_{i}^{\ast}\widetilde{\mathbf{z}}_{i}^{\ast\top}E_{\mathbf{X},\mathbf{W_0}}[e_{i}^{\ast 2}]$. 

The above implies that

\begin{align*}
n^{1/2}(\widehat{\boldsymbol{\psi}}_{\text{G3SLS}}-\boldsymbol{\psi})\overset{d}{\longrightarrow}&(\mathbf{\Gamma}^{\top}\mathbf{Q}_{\mathbf{Z}^{\ast}\mathbf{D}_{0}}\mathbf{\Gamma})^{-1}\times N(\mathbf{0},\boldsymbol{\Omega})\\
&N(\mathbf{0},(\mathbf{\Gamma}^{\top}\mathbf{Q}_{\mathbf{Z}^{\ast}\mathbf{D}_{0}}\mathbf{\Gamma})^{-1}\boldsymbol{\Omega}(\mathbf{\Gamma}\mathbf{Q}^{\top}_{\mathbf{Z}^{\ast}\mathbf{D}_{0}}\mathbf{\Gamma}^{\top})^{-1}).
\end{align*}
\hfill
\end{proof}

\section{Asymptotic Theory--Conditions\label{Appendix_B}}
\setcounter{assumption}{0}%

This section provides a set of sufficient conditions for the asymptotic distribution of the proposed estimator. The analysis is made conditional on the realization of the exogenous variables $\mathbf{X}$ and the exogenous adjacency matrix $\mathbf{W}_{0}$. To this end, the  notation $E_{\mathbf{X},\mathbf{W}_0}$, $\text{var}_{\mathbf{X},\mathbf{W}_0}$, and ${\Pr}_{\mathbf{X},\mathbf{W}_0}$ is used here to represent the expectations, variances, and probabilities conditional on $\mathbf{X}$ and $\mathbf{W}_0$, respectively.

\begin{assumption}
\label{D1} The row and column sums of the matrix $\mathbf{W}_{0}$ are bounded uniformly $\forall n$.
\end{assumption}

\begin{assumption}
\label{D2} All elements in the matrix $[\mathbf{X},\mathbf{v}]$ are uniformly bounded in absolute value $\forall n$.
\end{assumption}

\begin{assumption}
\label{D3} The innovations, $\{u_{j,i}\}$, for $j=1,\ldots,k+1$, $i=1,\ldots,n$ are independently jointly distributed $\forall n$, with $E_{\mathbf{X},\mathbf{W}_0}[u_{j,i}]=0$, and $E_{\mathbf{X},\mathbf{W}_0}[u_{j,i}^{2}]=\sigma^{2}_{u_{j,i}}<\infty$, $\forall i=1,\ldots,n$. Additionally, the $\{u_{j,i}\}$ elements are uniformly bounded in absolute value satisfying $\lim_{n\to\infty}n^{-1}\sum_{i=1}^{n}\allowbreak\sigma^{2}_{u_{j,i}} <\infty$, $\liminf_{n\to\infty}n^{-1}\sum_{i=1}^{n}\sigma^{2}_{u_{j,i}} >0$, and for some $\xi>0$, $\sup_{n,i}E|u_{j,i}|^{2+\xi}<\infty$ for $j=1,\dots,k+1$. Furthermore, the covariance between two errors from equations $j$ and $s$, i.e., $E[u_{j,i}u_{s,i}]=\sigma_{u_{(j,s);i}}$ satisfies $\lim_{n\to\infty} n^{-1}\sum_{i=1}^{\infty}\sigma_{u_{(j,s);i}} <\infty$, $\forall j,s=1,\dots,k+1$, $i=1,\ldots,n$.
\end{assumption}

\noindent Assumptions \ref{D1}, \ref{D2}, and \ref{D3} facilitate the derivations. Assumption \ref{D2}, for example, implies that $\mathbf{y}$, and therefore $\mathbf{S}=(\mathbf{y},\mathbf{X})$, are uniformly bounded in absolute value for all $n$. Assumption \ref{D3} leaves the dependence structure between the columns in $\mathbf{W}\mathbf{S}$ unrestricted, i.e., it allows for the existence of correlation between the errors in the system of equations $\mathbf{WS}=\mathbf{W}_{0}\mathbf{S}+\mathbf{U}$, which means that some unobserved variables can potentially jointly explain $\{\mathbf{Wy},\mathbf{W}\mathbf{x}_{1},\dots,\mathbf{W}\mathbf{x}_{k}\}$ after controlling for $\mathbf{W}_{0}\mathbf{S}$. Also note that \ref{D3} implies that the covariance of the errors does not explode asymptotically, i.e. the absolute sum of the covariances does not grow faster than $n$.

\begin{assumption}
\label{D4} The matrix $\mathbf{S}=[\mathbf{y},\mathbf{X}]$ is such that $\mathbf{Q}_{\mathbf{S}\mathbf{W}_{0}\mathbf{S}}$ is finite and non-singular, where $\mathbf{Q}_{\mathbf{S}\mathbf{W}_{0}\mathbf{S}}= \lim_{n\to\infty}n^{-1}\sum_{i=1}^{n}\bar{\mathbf{s}}_{0;(i)}\bar{\mathbf{s}}_{0;(i)}^{\top}$, and $\bar{\mathbf{s}}_{0;(i)}=\left[\bar{y}_{0;(i)},\bar{x}_{0;1(i)}, \ldots, \bar{x}_{0;k(i)}\right]^{\top}$, and $\bar{y}_{0;(i)}$ represents the average of outcomes of individual $i$'s connections, i.e., $\bar{y}_{0;(i)}\equiv\mathbf{w}_{0;i}\mathbf{y}$ where $\mathbf{w}_{0;i}$ is the first row of the adjacency matrix $\mathbf{W}_0$. The notation is analogous for $k$ regressors in the matrix $\mathbf{X}$.
\end{assumption}

\begin{assumption}
\label{D5} The matrix $\mathbf{Z}$ is such that $\mathbf{Q}_{\mathbf{ZZ}}= \lim_{n\to\infty}n^{-1}\sum_{i=1}^n\mathbf{z}_{i}\mathbf{z}_{i}^{\top}$ is finite and non-singular. Additionally, $\mathbf{Q}_{\mathbf{ZD}_{0}}= \lim_{n\to\infty}n^{-1}\sum_{i=1}^{n}\mathbf{z}_{i}\mathbf{d}_{0;i}^{\top}$ is finite and non-singular.
\end{assumption}

\begin{assumption}
\label{D6} The structural innovations $\{v_{i}\}$ are, for each $n$, independently jointly distributed with $E_{\mathbf{X},\mathbf{W}_0}[v_{i}]=0$ and $E_{\mathbf{X},\mathbf{W}_0}[v_{i}^{2}]=\sigma^{2}_{v_{i}}<\infty$, $\forall i=1,\ldots,n$. Additionally, they are uniformly bounded in absolute value, and satisfy $\lim_{n\to\infty} n^{-1}\sum_{i=1}^{n}\sigma^{2}_{v_{i}} <\infty$, $\liminf_{n\rightarrow\infty} n^{-1}\sum_{i=1}^{n}\sigma^{2}_{v_{i}}\allowbreak >0$, and for some $\xi>0$, $\sup_{n,i}E|v_{i}|^{2+\xi}<\infty$.
\end{assumption}

\begin{assumption}
\label{D7} The innovations $\{u_{j;i}\}$ and $\{v_{i}\}$ are, for each $n$, independently jointly distributed for all $j=1,\dots,k+1$ and $i=1,\ldots,n$.
\end{assumption}

\begin{assumption}
\label{D8} Both $\mathbf{Z}^{\ast}$ and $\mathbf{D}_{0}$ guarantee $\mathbf{Q}_{\mathbf{Z}^{\ast}\mathbf{D}_{0}}= \lim_{n\to\infty}n^{-1}\sum_{i=1}^{n}\mathbf{z}_{i}^{\ast}\mathbf{d}_{0;i}^{\top}$ is finite and non-singular.
\end{assumption}

\noindent Assumptions \ref{D4}, \ref{D5}, and \ref{D8} are standard in the sense that they guaranty the application of the continuous mapping theorem in various parts of the proofs and the existence of asymptotic variance--covariance matrices. Assumption \ref{D6} is similar to Assumption \ref{D3}, while \ref{D7} is reminiscent of the usually required statistical independence between reduced--form and structural errors in classical instrumental variable estimation.

\section{Auxiliary Results\label{Appendix_C}}
\renewcommand\thesection{\Alph{section}}
\begin{lemma}[Consistency of $\widehat{\mathbf{\Pi}}$]\label{l1}
Let Assumptions \ref{A2}, \ref{D1}--\ref{D4} hold, then $\widehat{\mathbf{\Pi}}\overset{p}{\longrightarrow}\mathbf{\Pi}$.
\end{lemma}

\begin{proof}[Proof of Lemma \ref{l1}]
From equation \eqref{E17} one has that the matrix of parameters $\widehat{\mathbf{\Pi}}$ is given by $\widehat{\mathbf{\Pi}}=(\widehat{\boldsymbol{\pi}}_1,\widehat{\boldsymbol{\pi}}_2,\ldots,\widehat{\boldsymbol{\pi}}_{k+1})^{\top}$. Without loss of generality, $\hat{\boldsymbol{\pi}}_1$ is analyzed here. The estimated coefficient is given by

\begin{equation}
\label{CP1}
\widehat{\boldsymbol{\pi}}_1^{\top}=(\mathbf{S}^{\top}\mathbf{W}_{0}^{2}\mathbf{S})^{-1}(\mathbf{S}^{\top}\mathbf{W}_{0}\mathbf{W}\mathbf{y}).
\end{equation}

\noindent It follows from equation \eqref{E4} that $\mathbf{W}\mathbf{y}=\mathbf{W}_{0}^{2}\mathbf{S}\boldsymbol{\pi}_{1}^{\top}+\mathbf{u}_{1}$, where $\mathbf{U}=\left[\mathbf{u}_{1},\dots,\mathbf{u}_{k+1}\right]$ and each $\mathbf{u}_{j}$ is a $(n \times 1)$--vector. Therefore, equation \eqref{CP1} can be written as:

\begin{equation}
\label{CP2}
\widehat{\boldsymbol{\pi}}_{1}^{\top}=\boldsymbol{\pi}_{1}^{\top}+(n^{-1}\mathbf{S}^{\top}\mathbf{W}_{0}^{2}\mathbf{S})^{-1}(n^{-1}\mathbf{S}^{\top}\mathbf{W}_{0}\mathbf{u}_{1}).
\end{equation}

\noindent Note that $\mathbf{w}_{0;i}\mathbf{x}_{j}$ represents the average of characteristic $j$ for individual $i$'s connections. To facilitate notation, let $\bar{x}_{0;j(i)}\equiv\mathbf{w}_{0;i}\mathbf{x}_{j}$. Then, for all $k$ characteristics and outcome variables, for individual $i$, we have $\bar{\mathbf{s}}_{0;(i)}=\left[\bar{y}_{0;(i)},\bar{x}_{0;1(i)}, \ldots, \bar{x}_{0;k(i)}\right]^{\top}$. With this new notation, note that $\sum_{i=1}^{n}\bar{\mathbf{s}}_{0;(i)}\bar{\mathbf{s}}_{0;(i)}^{\top}=\mathbf{S}^{\top}\mathbf{W}_{0}^{2}\mathbf{S}$. Therefore, by \ref{D4} and the continuous mapping theorem, $\left(n^{-1}\mathbf{S}^{\top}\mathbf{W}_{0}^{2}\mathbf{S}\right)^{-1}$ converges to $\mathbf{Q}_{\mathbf{S}\mathbf{W}_{0}\mathbf{S}}^{-1}$ which is finite. To analyze the second part of equation \eqref{CP2}, note that $\mathbf{S}^{\top}\mathbf{W}_{0}\mathbf{u}_{1}$ can be written as a vector given by the odd products between the individual averages of each variable in $\mathbf{S}$ and the vector of errors $\mathbf{u}_{1}$ -- recall that this is true for all vectors of errors $\mathbf{u}_{j}$, $j=1,\dots,k+1$ -- That is:

\begin{equation}
\label{CP3}
    \mathbf{S}^{\top}\mathbf{W}_{0}\mathbf{u}_{1}=\left[\bar{\mathbf{y}}^{\top}_{0}\mathbf{u}_{1},\bar{\mathbf{x}}_{0;1}^{\top}\mathbf{u}_{1},\dots,\bar{\mathbf{x}}_{0;k}^{\top}\mathbf{u}_{1}\right]^{\top},
\end{equation}

\noindent where $\bar{\mathbf{y}}_{0}^{\top}=\left[\bar{y}_{0;(1)},\ldots,\bar{y}_{0;(n)}\right]$ and $\bar{\mathbf{x}}_{0;j}^{\top}=\left[\bar{x}_{0;j(1)},\dots,\bar{x}_{0;j(n)}\right]$ for $j=1,\ldots,k$. Now we show that the second summand on the right-hand side of \eqref{CP2} converges in probability to zero. From Assumption \ref{A2} it follows that $E_{\mathbf{X},\mathbf{W}_0}[\mathbf{U}]=\mathbf{0}$. Then, by the law of iterated expectations, $E[\mathbf{S}^{\top}\mathbf{W}_{0}\mathbf{U}]=\mathbf{0}$, and the expectation of each element in the right-hand side of \eqref{CP3} equals zero. Therefore, it is sufficient to show that $\forall\varepsilon>0$, as $n\rightarrow\infty$, $\Pr_{\mathbf{X},\mathbf{W}_0}\{|n^{-1}\bar{\mathbf{y}}^{\top}_{0}\mathbf{u}_{1}|>\varepsilon\}\rightarrow0$ and  $\Pr_{\mathbf{X},\mathbf{W}_0}\{|n^{-1}\bar{\mathbf{x}}_{0;j}^{\top}\mathbf{u}_{1}|>\varepsilon\}\rightarrow0$ for $j=1,\dots,k$. Note that given Assumptions \ref{D1} and \ref{D2}, i.e., $\text{max}_{i:1,\dots,n} \hspace{1mm} \bar{y}_{0;(i)}$, and $\text{max}_{i:1,\dots,n} \hspace{1mm} \bar{x}_{0;j(i)}$ are bounded for all $j=1,\dots,k$, and the independence Assumption in \ref{D3}:

\begin{align*}
    & \text{var}_{\mathbf{X},\mathbf{W}_0}\left(n^{-1}\bar{\mathbf{y}}_{0}^{\top}\mathbf{u}_{1}\right)\leq \text{const.}\times\frac{\sum_{i=1}^{n}\sigma^{2}_{u_{1;i}}}{n^2}\text{,}\\
    & \text{var}_{\mathbf{X},\mathbf{W}_0}\left(n^{-1}\bar{\mathbf{x}}_{0;j}^{\top}\mathbf{u}_{j+1}\right)\leq\text{const.}\times\frac{\sum_{i=1}^{n}\sigma^{2}_{u_{j+1;i}}}{n^2}\text{, for }j=1,\ldots,k,
\end{align*}

\noindent where const. defines any generic constant that may differ from one another. Therefore Assumption \ref{D3} implies that $\text{var}_{\mathbf{X},\mathbf{W}_0}\left(n^{-1}\bar{\mathbf{y}}_{0}^{\top}\mathbf{u}_{1}\right)=o(1)$ and $\text{var}_{\mathbf{X},\mathbf{W}_0}(n^{-1}\bar{\mathbf{x}}_{0;j}^{\top}\mathbf{u}_{1})=o(1)$. Then, by the Chebychev's inequality, $\forall\varepsilon>0$:

\begin{align}
    \label{CP4}
    &{\Pr}_{\mathbf{X},\mathbf{W}_0} \{|n^{-1}\bar{\mathbf{y}}_{0}^{\top}\mathbf{u}_{1}|>\varepsilon\}\leq \frac{\text{var}_{\mathbf{X},\mathbf{W}_0}\left(n^{-1}\bar{\mathbf{y}}_{0}^{\top}\mathbf{u}_{1}\right)}{\varepsilon^{2}}\rightarrow 0\text{, as }n \rightarrow \infty\text{,}\\
    \label{CP5}
    &{\Pr}_{\mathbf{X},\mathbf{W}_0}\{|n^{-1}\bar{\mathbf{x}}_{0;j}^{\top}\mathbf{u}_{j+1}|>\varepsilon\}\leq \frac{\text{var}_{\mathbf{X},\mathbf{W}_0}\left(n^{-1}\bar{\mathbf{x}}_{0;j}^{\top}\mathbf{u}_{j+1}\right)}{\varepsilon^{2}}\rightarrow 0\text{, as }n \rightarrow \infty\text{, for } j=1,\ldots,k.
\end{align}

\noindent Equations \eqref{CP4} and \eqref{CP5} imply that $n^{-1}\mathbf{S}^{\top}\mathbf{W}_{0}\mathbf{u}_{j}\overset{p}{\longrightarrow}0$ for all $j=1,\dots, k+1$. Therefore, $\widehat{\boldsymbol{\pi}}_{j}\overset{p}{\longrightarrow}\boldsymbol{\pi}_{j}$ for $j=1,\dots, k+1$, i.e.,  $\widehat{\mathbf{\Pi}}\overset{p}{\longrightarrow}\mathbf{\Pi}$.\hfill
\end{proof}

\begin{lemma}[Consistency of $\widehat{\boldsymbol{\psi}}^{*}_{\text{2SLS}}$]
\label{l2}
Let Assumptions \ref{A1}, \ref{D3}, \ref{D5}, \ref{D6}, and \ref{D7} hold, then $\widehat{\boldsymbol{\psi}}^{*}_{\text{\emph{2SLS}}}\overset{p}{\longrightarrow}\boldsymbol{\psi}^{*}$.
\end{lemma}

\begin{proof}[Proof of Lemma \ref{l2}]
From equation \eqref{E19} note that the 2SLS estimator can be written as:

\begin{equation}
\label{CP6}
\widehat{\boldsymbol{\psi}}^{*}_{\text{2SLS}}=\boldsymbol{\psi}^{*}+\left[n^{-1}\mathbf{D}_{0}^{\top}\mathbf{Z}\left(n^{-1}\mathbf{Z}^{\top}\mathbf{Z}\right)^{-1}n^{-1}\mathbf{Z}^{\top}\mathbf{D}_{0}\right]^{-1} n^{-1}\mathbf{D}_{0}^{\top}\mathbf{Z}\left(n^{-1}\mathbf{Z}^{\top}\mathbf{Z}\right)^{-1}n^{-1}\mathbf{Z}^{\top}\mathbf{e}.
\end{equation}

\noindent Recall that $\mathbf{D}_{0}=[\mathbf{\iota},\mathbf{X},\mathbf{W}_{0}\mathbf{y},\mathbf{W}_{0}\mathbf{X}]$, $\mathbf{Z}=[\mathbf{\iota},\mathbf{X},\mathbf{W}^{2}_{0}\mathbf{X},\mathbf{W}_{0}\mathbf{X}]$, $\boldsymbol{\psi}^{*}=[\alpha,\boldsymbol{\gamma},\boldsymbol{\theta}^{*}]^{\top}$, and $\mathbf{e}=\mathbf{U}\boldsymbol{\theta}+\mathbf{v}$. Note that by Assumption \ref{D5} and the continuous mapping theorem,

\begin{align*}
&\left[n^{-1}\mathbf{D}_{0}^{\top}\mathbf{Z}\left(n^{-1}\mathbf{Z}^{\top}\mathbf{Z}\right)^{-1}n^{-1}\mathbf{Z}^{\top}\mathbf{D}_{0}\right]^{-1} n^{-1}\mathbf{D}_{0}^{\top}\mathbf{Z}\left(n^{-1}\mathbf{Z}^{\top}n^{-1}\mathbf{Z}\right)^{-1}\xrightarrow{p}\\
& \quad\quad\quad \left(\mathbf{Q}^{\top}_{\mathbf{ZD_{0}}}\mathbf{Q}^{-1}_{\mathbf{ZZ}}\mathbf{Q}_{\mathbf{ZD_{0}}}\right)^{-1}\mathbf{Q}^{\top}_{\mathbf{ZD_{0}}}\mathbf{Q}^{-1}_{\mathbf{ZZ}},
\end{align*}

\noindent which is a finite by Assumption \ref{D5}. The following notation is now introduced, let $\mathbf{U}\boldsymbol{\theta}\equiv\Tilde{\mathbf{u}}=[\Tilde{u}_{1},\ldots,\Tilde{u}_{n}]^{\top}$ where each element $\Tilde{u}_{i}=\beta u_{1,i}+\sum_{j=1}^{k}u_{j+1,i}\delta_{j}$ and $u_{j,i}$ is the $(j,i)$th element of the matrix $\mathbf{U}$. Therefore, note that $\mathbf{e}=\Tilde{\mathbf{u}}+\mathbf{v}$, and consequently:

\begin{align*}
    \mathbf{Z}^{\top}\mathbf{e} &= \begin{bmatrix}
           \sum_{i=1}^{n}\Tilde{u}_{i}+v_{i} \\
            \mathbf{X}^{\top}\Tilde{\mathbf{u}}+\mathbf{X}^{\top}\mathbf{v}\\
            (\mathbf{W}_{0}^{2}\mathbf{X})^{\top}\Tilde{\mathbf{u}}+(\mathbf{W}_{0}^{2}\mathbf{X})^{\top}\mathbf{v}\\
           (\mathbf{W}_{0}\mathbf{X})^{\top}\Tilde{\mathbf{u}}+(\mathbf{W}_{0}\mathbf{X})^{\top}\mathbf{v}
         \end{bmatrix}.
  \end{align*}

\noindent Assumptions \ref{A1}, \ref{D3}, and \ref{D6} along with the law of iterated expectations imply that $E[\mathbf{Z}^{\top}\mathbf{e}]=\mathbf{0}$. Again, to show converges it suffices to show that all the components of the vector $n^{-1}\mathbf{Z}^{\top}\mathbf{e}$ converge in probability to their expectation. It is useful to find the variance of each of the components of the vector:

\begin{multline}
    \text{var}_{\mathbf{X},\mathbf{W}_0}\left(n^{-1}\sum_{i=1}^{n}\Tilde{u}_{i}+v_{i}\right)=n^{-2}E_{\mathbf{X},\mathbf{W}_0}\left[\left(\sum_{i=1}^{n}\Tilde{u}_{i}+v_{i}\right)^{2}\right]
    \\=n^{-2}\left\{E_{\mathbf{X},\mathbf{W}_0}\left[\left(\sum_{i=1}^{n}\Tilde{u}_{i}\right)^{2}+\left(\sum_{i=1}^{n}v_{i}\right)^{2}\right]\right\},\nonumber
\end{multline}

\noindent where the second equality above follows from the independence between $\Tilde{u}_{i}$ and $v_{i}$ in Assumption \ref{D7}. Note that by Assumption \ref{D6}, $E_{\mathbf{X},\mathbf{W_0}}\left[(\sum_{i=1}^{n}v_{i})^{2}\right]=\sum_{i=1}^{n}\sigma^{2}_{v_{i}}$. However, because of the existence of potential correlation between the errors $u_{ji}$ for different outcomes $j$, the variance of $\Tilde{u}_{i}$ should reflect the presence of those covariance terms. Note that

\begin{align*}
&E_{\mathbf{X},\mathbf{W}_0}\left[\left(\sum_{i=1}^{n}\Tilde{u}_{i}\right)^{2}\right]=\sum_{i=1}^{n}E_{\mathbf{X},\mathbf{W}_0}[\Tilde{u}_{i}^{2}]+E_{\mathbf{X},\mathbf{W}_0}[\Tilde{u}_{1}\Tilde{u}_{2}]+\dots+E_{\mathbf{X},\mathbf{W}_0}[\Tilde{u}_{n-1}\Tilde{u}_{n}] =\\
& \quad\quad\quad \sum_{i=1}^{n}E_{\mathbf{X},\mathbf{W}_0}[(\Tilde{u}_{i})^{2}],
\end{align*}

\noindent where the second equality follows from the independence in Assumption \ref{D3}. Note that the covariance between errors happens inside each $\Tilde{u}_{i}$, i.e.,

\begin{align*}
&E_{\mathbf{X},\mathbf{W}_0}[\Tilde{u}_{i}^{2}]=E_{\mathbf{X},\mathbf{W}_0}\left[\left(\beta u_{1,i}+\sum_{j=1}^{k}u_{j+1,i}\delta_{j}\right)^{2}\right]\text{,}\\
&=E\Bigg[\beta^{2}E[u_{1,i}]+\sum_{j=1}^{k}E[u^{2}_{j+1,i}]\delta^{2}_{j}+\beta\sum_{j=1}^{k}E[u_{1,i}u_{j+1,i}]\delta_{j}+E[u_{2,i}u_{3,i}]\delta_{1}\delta_{2}+\dots \\
&+E[u_{k,i}u_{k+1,i}]\delta_{k-1}\delta_{k}\Bigg]\text{,}\\
&=\beta\sigma^{2}_{u_{1,i}}+\beta\sum_{j=1}^{k}\delta_{j}\sigma_{u_{(1,j+1);i}}+\sum_{s=1}^{k}\sum_{j=1}^{k}\delta_{s}\delta_{j}\sigma_{u_{(s,j);i}}\text{,}
\end{align*}

\noindent where $\sigma_{u_{(s,j);i}}$ represents the covariance of the errors from equations $s$ and $j$ for individual $i$. Therefore, the total variance can be written as:

\begin{align}
\text{var}_{\mathbf{X},\mathbf{W}_0}\left(n^{-1}\sum_{i=1}^{n}\Tilde{u}_{i}+v_{i}\right)=&\frac{1}{n^{2}}\sum_{i=1}^{n}\sigma^{2}_{v_{i}}\nonumber\\
&+\frac{1}{n^{2}}\sum_{i=1}^{n}\left(\beta\sigma^{2}_{u_{1,i}}+\beta\sum_{j=1}^{k}\delta_{j}\sigma_{u_{(1,j+1);i}}\right)\nonumber \\ 
&+\frac{1}{n^{2}}\sum_{i=1}^{n}\sum_{s=1}^{k}\sum_{j=1}^{k}\delta_{s}\delta_{j}\sigma_{u_{(s,j);i}}.
\label{CP7}
\end{align}

\noindent Firstly, $n^{-1}\sum_{i=1}^{n}\sigma^{2}_{v_{i}}=O(1)$ by Assumption \ref{D6} which implies that $n^{-2}\sum_{i=1}^{n}\sigma^{2}_{v_{i}}=o(1)$. Similarly, given that $k$ is finite and the average of the absolute value of the covariances is uniformly bounded by Assumption \ref{D3}, it follows that the two summand terms in \eqref{CP7} are $o(1)$. Using similar arguments and conditioning on the realization of $\mathbf{X}$ and $\mathbf{W}_{0}$ one also has

\begin{align}
\text{var}_{\mathbf{X},\mathbf{W}_0}\left(n^{-1}\mathbf{X}^{\top}\Tilde{\mathbf{u}}+\mathbf{X}^{\top}\mathbf{v}\right)=&\frac{1}{n^{2}}\sum_{i=1}^{n}x_{j;i}^{2}\sigma^{2}_{v_{i}}+\frac{1}{n^{2}}\sum_{i=1}^{n}x_{j;i}^{2}\left(\beta\sigma^{2}_{u_{1,i}}+\beta\sum_{j=1}^{k}\delta_{j}\sigma_{u_{(1,j+1);i}}\right)\nonumber\\
&+\frac{1}{n^{2}}\sum_{i=1}^{n}\sum_{s=1}^{k}\sum_{j=1}^{k}x_{j;i}^{2}\delta_{s}\delta_{j}\sigma_{u_{(s,j);i}},\label{CP8}
\end{align}

\begin{align}
\text{var}_{\mathbf{X},\mathbf{W}_0}&\left(n^{-1}(\mathbf{W}_{0}^{2}\mathbf{X})^{\top}\Tilde{\mathbf{u}}+(\mathbf{W}_{0}^{2}\mathbf{X})^{\top}\mathbf{v}\right)=\frac{1}{n^{2}}\sum_{i=1}^{n}\bar{x}_{02;j(i)}^{2}\sigma^{2}_{v_{i}}\nonumber\\
&+\frac{1}{n^{2}}\sum_{i=1}^{n}\bar{x}_{02;j(i)}^{2}\left(\beta\sigma^{2}_{u_{1,i}}+\beta\sum_{j=1}^{k}\delta_{j}\sigma_{u_{(1,j+1);i}}\right)\nonumber \\
&+\frac{1}{n^{2}}\sum_{i=1}^{n}\sum_{s=1}^{k}\sum_{j=1}^{k}\bar{x}_{02,j(i)}^{2}\delta_{s}\delta_{j}\sigma_{u_{(s,j);i}},\label{CP9}
\end{align}

\begin{align}
\text{var}_{\mathbf{X},\mathbf{W}_0}\left(n^{-1}(\mathbf{W}_{0}\mathbf{X})^{\top}\Tilde{\mathbf{u}}+(\mathbf{W}_{0}\mathbf{X})^{\top}\mathbf{v}\right)=&\frac{1}{n^{2}}\sum_{i=1}^{n}\bar{x}_{0;j(i)}^{2}\sigma^{2}_{v_{i}}\nonumber \\
&+\frac{1}{n^{2}}\sum_{i=1}^{n}\bar{x}_{0;j(i)}^{2}\left(\beta\sigma^{2}_{u_{1,i}}+\beta\sum_{j=1}^{k}\delta_{j}\sigma_{u_{(1,j+1);i}}\right)\nonumber\\ &+\frac{1}{n^{2}}\sum_{i=1}^{n}\sum_{s=1}^{k}\sum_{j=1}^{k}\bar{x}_{0,j(i)}^{2}\delta_{s}\delta_{s}\sigma_{u_{(s,j);i}},\label{CP10}
\end{align}

\noindent where as mentioned before $\bar{x}_{0;j(i)}=\mathbf{w}_{0;i}\mathbf{x}_{j}$ is the average of individual's $i$ connections in the $j$th characteristic. Similarly, $\bar{x}_{02;j(i)}=\mathbf{w}_{0;i}^{2}\mathbf{x}_{j}$ represents the average in the $j$th characteristic of individual $i$'s indirect connections. The regressors and the predetermined adjacency matrix are both bounded given Assumptions \ref{D1} and \ref{D2}. Therefore, the same analysis performed for the variance in equation $\eqref{CP7}$ applies to $\eqref{CP8}$, $\eqref{CP9}$, and $\eqref{CP10}$. Thus, $\text{var}_{\mathbf{X},\mathbf{W}_0}\left(n^{-1}\mathbf{X}^{\top}\Tilde{\mathbf{u}}+\mathbf{X}^{\top}\mathbf{v}\right)=o(1)$, $\text{var}_{\mathbf{X},\mathbf{W}_0}\left(n^{-1}(\mathbf{W}_{0}^{2}\mathbf{X})^{\top}\Tilde{\mathbf{u}}+(\mathbf{W}_{0}^{2}\mathbf{X})^{\top}\mathbf{v}\right)=o(1)$, and $\text{var}_{\mathbf{X},\mathbf{W}_0}\left(n^{-1}(\mathbf{W}_{0}\mathbf{X})^{\top}\Tilde{\mathbf{u}}+(\mathbf{W}_{0}\mathbf{X})^{\top}\mathbf{v}\right)=o(1)$. This information along with an application of the Chebychev's inequality as in (\ref{CP4}) and (\ref{CP5}) imply that $n^{-1}\mathbf{Z}^{\top}\mathbf{e}\overset{p}{\longrightarrow}0$ and $\widehat{\boldsymbol{\psi}}^{*}_{\text{2SLS}}\overset{p}{\longrightarrow}\boldsymbol{\psi}^{*}$.\hfill
\end{proof}

\noindent Recall from \eqref{En2} that $\boldsymbol{\theta}^{*}=\mathbf{\Pi}\boldsymbol{\theta}$. Let $\widehat{\boldsymbol{\theta}}^{*}_{\text{2SLS}}$ be the sub-vector of $\widehat{\boldsymbol{\psi}}^{*}_{\text{2SLS}}$ containing the estimators of $\beta$ and $\boldsymbol{\delta}$. Therefore, by the plug-in principle, an estimator for $\boldsymbol{\theta}$ is given by $\widehat{\boldsymbol{\theta}}=\widehat{\Pi}^{-1}\widehat{\boldsymbol{\theta}}^{*}_{\text{2SLS}}$. The following Lemma establishes its consistency.

\begin{lemma}[Consistency of $\widehat{\boldsymbol{\theta}}$]
\label{l3}
Let Assumptions \ref{A1}, \ref{A2}, \ref{D1}--\ref{D7} hold, then $\widehat{\boldsymbol{\theta}}\overset{p}{\longrightarrow}\boldsymbol{\theta}$.
\end{lemma}

\begin{proof}[Proof of Lemma \ref{l3}]
The result follows from Lemmas \ref{l1}, \ref{l2}, and the continuous mapping theorem.\hfill
\end{proof}

\begin{lemma}
\label{l4}
Define $\widehat{\widetilde{\mathbf{Z}}}{}^{\ast}$ as in equation \eqref{EW0y} and $\mathbf{e}^{*}=\mathbf{M}_{\mathbf{W}_{0}}\mathbf{U}\boldsymbol{\theta}+\mathbf{v}$. Let Assumptions \ref{D1}--\ref{D6} hold; then $n^{-1/2}\widehat{\widetilde{\mathbf{Z}}}{}^{\ast\top}\mathbf{e}^{\ast}=n^{-1/2}\widetilde{\mathbf{Z}}^{\ast\top}\mathbf{e}^{\ast}+o_{p}(1)$.
\end{lemma}

\begin{proof}
Starting with the definition of $\widehat{\widetilde{\mathbf{Z}}}{}^{\ast}$ from equation \eqref{EW0y}, note that $(\widehat{\widetilde{\mathbf{Z}}}{}^{\ast}-\widetilde{\mathbf{Z}}^{\ast})^{\top}\mathbf{e}^{\ast}=\mathbf{\Gamma}^{\top}(\widehat{\mathbf{Z}}^{\ast}-\mathbf{Z}^{\ast})^{\top}\mathbf{e}^{\ast}+(\widehat{\mathbf{\Gamma}}-\mathbf{\Gamma})^{\top}\widehat{\mathbf{Z}}^{\ast\top}\mathbf{e}^{\ast}$. Therefore, it is sufficient to show that:

\begin{align}
&n^{-1/2}(\widehat{\mathbf{Z}}^{\ast}-\mathbf{Z}^{\ast})^{\top}\mathbf{e}^{*}=o_{p}(1)\text{, and}\label{C2E1}\\
& n^{-1/2}\widehat{\mathbf{Z}}^{*\top}\mathbf{e}^{*}=O_{p}(1).\label{C2E2}
\end{align}

\noindent First, to show equation \eqref{C2E1}, note that $\widehat{\mathbf{Z}}^{\ast}-\mathbf{Z}^{\ast}=[\mathbf{0},\mathbf{O},\Delta E_{\mathbf{X},\mathbf{W_0}}[\mathbf{W}_0\mathbf{y}],\mathbf{O}]$, where $\mathbf{0}$ is a $n\times 1$ vector of zeros, $\mathbf{O}$ is a $n \times k$ matrix of zeros, and $\Delta E_{\mathbf{X},\mathbf{W}_0}[\mathbf{W}_0\mathbf{y}]=E_{\mathbf{X},\mathbf{W}_0}[\mathbf{W}_0\mathbf{y}](\widehat{\boldsymbol{\psi}},\widehat{\mathbf{\Pi}})-E_{\mathbf{X},\mathbf{W}_0}[\mathbf{W}_0\mathbf{y}](\boldsymbol{\psi},\mathbf{\Pi})$. From equation $\eqref{En}$ in the main text, one has $\mathbf{y}(\mathbf{I}-\beta^{*}\mathbf{W}_{0})=\alpha\boldsymbol{\iota}+\mathbf{W}_{0}\mathbf{X}\boldsymbol{\delta}^{*}+\mathbf{X}\boldsymbol{\gamma}+\mathbf{e}\equiv\widetilde{\mathbf{X}}\boldsymbol{\eta}+\mathbf{e}$, where $\widetilde{\mathbf{X}}=[\boldsymbol{\iota},\mathbf{W}_{0}\mathbf{X},\mathbf{X}]$ and $\boldsymbol{\eta}=[\alpha,\boldsymbol{\delta}^{\ast\top},\boldsymbol{\gamma}^{\top}]^{\top}$. Under the assumption that $|\beta^{\ast}|<1$, Assumptions \ref{D3}, and \ref{D6}, the expectation of $\mathbf{W}_{0}\mathbf{y}$ can be written as: $E_{\mathbf{X},\mathbf{W}_0}[\mathbf{W}_0\mathbf{y}](\boldsymbol{\psi},\mathbf{\Pi})= \mathbf{W}_0(\mathbf{I}-\beta^{\ast}\mathbf{W}_{0})^{-1}\widetilde{\mathbf{X}}\boldsymbol{\eta}=\mathbf{W}_0(\sum_{r=0}^{\infty}\beta^{\ast r}\mathbf{W}_{0}^{r})\widetilde{\mathbf{X}}\boldsymbol{\eta}=\mathbf{W}_0\widetilde{\mathbf{X}}\boldsymbol{\eta}+\mathbf{W}^{2}_0\widetilde{\mathbf{X}}\boldsymbol{\eta}\beta^{*}+\mathbf{W}^{3}_0\widetilde{\mathbf{X}}\boldsymbol{\eta}\beta^{\ast 2}+\ldots$ This last equation implies that $\Delta E_{\mathbf{X},\mathbf{W}_0}[\mathbf{W}_0\mathbf{y}]^{\top}=(\widehat{\boldsymbol{\eta}}_{\text{2SLS}}-\boldsymbol{\eta})^{\top}(\mathbf{W}_0\widetilde{\mathbf{X}})^{\top}+(\widehat{\boldsymbol{\eta}}_{\text{2SLS}}\widehat{\beta}_{\text{2SLS}}^{*}-\boldsymbol{\eta}\beta^{\ast})^{\top}(\mathbf{W}^{2}_0\widetilde{\mathbf{X}})^{\top}+\ldots$, where $\widehat{\boldsymbol{\eta}}_{\text{2SLS}}$ and $\widehat{\beta}_{\text{2SLS}}^{\ast}$ are elements of $\widehat{\boldsymbol{\psi}}^{\ast}_{\text{\text{2SLS}}}$. Note that except for $\Delta E_{\mathbf{X},\mathbf{W}_0}[\mathbf{W}_0\mathbf{y}]^{\top}\mathbf{e}^{*}$, the other elements in $(\widehat{\mathbf{Z}}^{\ast}-\mathbf{Z}^{\ast})^{\top}\mathbf{e}^{*}=0$, $\forall n$ and therefore $o_{p}(1)$. It is only left to show that $\Delta E_{\mathbf{X},\mathbf{W}_0}[\mathbf{W}_0\mathbf{y}]^{\top}\mathbf{e}^{*}=o_{p}(1)$. From the definition before,

\begin{equation}
\Delta E[\mathbf{W}_0\mathbf{y}]^{\top}\mathbf{e}^{\ast}=(\widehat{\boldsymbol{\eta}}_{\text{2SLS}}-\boldsymbol{\eta})^{\top}(\mathbf{W}_0\widetilde{\mathbf{X}})^{\top}\mathbf{e}^{\ast}+(\widehat{\boldsymbol{\eta}}_{\text{2SLS}}\widehat{\beta}_{\text{2SLS}}^{\ast}-\boldsymbol{\eta}\beta^{\ast})^{\top}(\mathbf{W}^{2}_0\widetilde{\mathbf{X}})^{\top}\mathbf{e}^{\ast}+\ldots
\end{equation}

\noindent Lemma \ref{l2} along with the continuous mapping theorem imply  $(\widehat{\boldsymbol{\eta}}_{\text{2SLS}}\widehat{\beta}_{\text{2SLS}}^{\ast r}-\boldsymbol{\eta}\beta^{\ast r})=o_{p}(1)$, $\forall r\ge 1$. Now, it is required to bound  the elements of $(\mathbf{W}_{0}^{r}\widetilde{\mathbf{X}})^{\top}\mathbf{e}^{*}$, $\forall r\ge 1$. Note that each $\mathbf{W}_{0}^{r}$ is equivalent to calculating the average of the indirect connections of individual $i$ at distance $r$. Moreover, Assumption \ref{A1} guaranties that $E_{\mathbf{X},\mathbf{W}_0}[\mathbf{v}]=\mathbf{0}$, and by Assumption \ref{A2}, one also has $E_{\mathbf{X},\mathbf{W}_0}[\mathbf{U}]=\mathbf{O}$, implying $E_{\mathbf{X},\mathbf{W}_0}[\mathbf{e}^{*}]=\mathbf{0}$. Therefore, $E[(\mathbf{W}_{0}^{r}\widetilde{\mathbf{X}})^{\top}\mathbf{e}^{*}]=\mathbf{0}$, $r\ge 1$.

\noindent Assumptions \ref{D3} and \ref{D6} imply that the innovations $e_{i}^{\ast}$ are independent across individuals $i$ with $E_{\mathbf{X},\mathbf{W_0}}[{e}_{i}^{\ast}]=0$, but are not identically distributed. Now, for some $\xi>0$ we have that $E|\beta u_{1,i}+\dots+\delta_k u_{k+1,i}|^{2+\xi}\le \beta E|u_{1,i}|^{2+\xi}+\dots+\delta_k E|u_{k+1,i}|^{2+\xi}$, $\forall i=1,\ldots,n$. Similarly, $E|u_{1,i}+\dots+\delta_{k}u_{k+1,i}+v_{i}|^{2+\xi}\leq E|u_{1,i}|^{2+\xi}+\dots+\delta_{k}E|u_{k+1,i}|^{2+\xi}+E|v_{i}|^{2+\xi} \leq \sup_{n,i} (E|u_{1,i}|^{2+\xi}+\dots+\delta_{k}E|u_{k+1,i}|^{2+\xi}+E|v_{i}|^{2+\xi})<\infty$, $\forall i=1,\ldots,n$, where the last inequality comes from Assumptions \ref{D3} and \ref{D6}. Thus, $\sup_{n,i} E|e^{*}_{i}|^{2+\xi}<\infty$. From the same set of Assumptions, $\liminf_{n\to\infty} n^{-1}\sum_{i=1}^{n}\sigma^{2}_{v_{i}} >0$ and $\liminf_{n\to\infty} n^{-1}\sum_{i=1}^{n}\sigma^{2}_{u_{j,i}} >0$, for $j=1,\dots,k+1$, which implies $\liminf_{n\to\infty} n^{-1}\sum_{i=1}^{n}E_{\mathbf{X},\mathbf{W}_0}[e_{i}^{\ast 2}] >0$. Therefore, Lyapunov's condition holds,, and by the Lindeberg-Feller central limittheorem,m, one has $n^{-1/2}(\mathbf{W}_{0}^{r}\widetilde{\mathbf{X}})^{\top}\mathbf{e}^{*}=O_{p}(1)$, $\forall r\ge 1$, and $n^{-1/2}(\widehat{\mathbf{Z}}^{\ast}-\mathbf{Z}^{\ast})^{\top}\mathbf{e}^{\ast}=o_{p}(1)$.

\noindent Furthermore, Assumptions $\ref{A1}$ and $\ref{A2}$ imply $E_{\mathbf{Z}^{\ast},\mathbf{W}_{0}}[\mathbf{e}^{\ast}]=\mathbf{0}$ and therefore $E[\mathbf{Z}^{\ast\top}\mathbf{e}^{\ast}]=0$. As before, a central limit theorem implies $n^{-1/2}\mathbf{Z}^{\ast\top}\mathbf{e}^{\ast}=O_{p}(1)$. All together, this implies that $n^{-1/2}\widehat{\mathbf{Z}}^{\ast\top}\mathbf{e}^{*}=\mathbf{Z}^{\ast\top}\mathbf{e}^{\ast}+o_{p}(1)$, so $(\widehat{\mathbf{\Gamma}}-\mathbf{\Gamma})n^{-1/2}\widehat{\mathbf{Z}}^{\ast\top}\mathbf{e}^{\ast}=(\widehat{\mathbf{\Gamma}}-\mathbf{\Gamma})[n^{-1/2}\mathbf{Z}^{*\top}\mathbf{e}^{\ast}+o_{p}(1)]$ and therefore by Lemma \ref{l1} one has $n^{-1/2}(\widehat{\widetilde{\mathbf{Z}}}^{\ast}-\widetilde{\mathbf{Z}}^{\ast})^{\top}\mathbf{e}^{\ast}=o_{p}(1)$ completing the proof.\hfill
\end{proof}

\putbib[typ]
\end{bibunit}

\emptythanks
\clearpage
\setcounter{page}{1}

\if1\blind
{
  \title{Regression with Observational Multilayered Network Data\\ -- Supplementary  Material --}
\author{Juan Estrada\thanks{Analysis Group Economic Consulting, Washington, DC, USA. \faEnvelopeO: \href{mailto:juan.estrada@analysisgroup.com}{juan.estrada@analysisgroup.com}.}  \and Kim P. Huynh\thanks{Department of Economics, Indiana University, 100 S Woodlawn, Bloomington, IN 47405, USA. \faEnvelopeO: \href{mailto:kim@huynh.tv}{kim@huynh.tv}. Laboratoire d’\'Economie d’Orl\'eans, Universit\'e d'Orl\'eans, Orl\'eans, France.}\and David T. Jacho-Ch\'{a}vez\thanks{Corresponding Author: Department of Economics, Emory University, Rich Building 306, 1602 Fishburne Dr., Atlanta, GA 30322-2240, USA. \faEnvelopeO: \href{mailto:djachocha@emory.edu}{djachocha@emory.edu}.} \and Leonardo S\'{a}nchez-Arag\'{o}n\thanks{Facultad de Ciencias Sociales y Human\'{i}sticas, Escuela Superior Polit\'{e}cnica del Litoral, ESPOL, Campus Gustavo Galindo Km. 30.5 V\'{i}a Perimetral, P.O. Box 09-01-5863, Guayaquil, Ecuador. \faEnvelopeO: \href{mailto:
lfsanche@espol.edu.ec}{lfsanche@espol.edu.ec}.}}
\maketitle
} \fi

\if0\blind
{
  \bigskip
  \bigskip
  \bigskip
  \title{\bf Regression with Observational Multilayered Network Data\\ -- Supplementary  Material --}
  \maketitle
  \medskip
} \fi

\setcounter{section}{0}

\appendix\renewcommand\thesection{Appendix \Alph{section}}
\renewcommand{\theequation}{\Alph{section}-\arabic{equation}}
\renewcommand{\theenumi}{\alph{enumi}}\renewcommand{\labelenumi}%
{\emph{(\theenumi)}}
\renewcommand{\theenumii}{(\roman{enumii})}\renewcommand{\labelenumii
}{\theenumii}
\renewcommand{\thetheorem}{\Alph{section}.\arabic{theorem}.}\renewcommand
{\theassumption}{\Alph{section}.\arabic{assumption}}

\begin{bibunit}[jpe]

\renewcommand\thesection{Appendix \Alph{section}}
\setcounter{section}{3}
\section{Additional Monte Carlo Results\label{Appendix_D}}
\renewcommand\thesection{\Alph{section}}
\subsection{Alternative Parameterizations}

The following tables report the results of further Monte Carlo exercises conducted to evaluate the finite-sample performance of the three estimators mentioned in Section \ref{mc} of the main manuscript, i.e., OLS, G2SLS, and G3SLS. The performance assessment is implemented systematically across all three designs presented there. Within each design, results are organized around specific variations in key design parameters. Unless otherwise explicitly stated, all remaining parameters of the data-generating process are held fixed at their baseline values. This approach ensures that any observed differences in estimator performance can be cleanly attributed to the parameter variations under consideration.

Table \ref{tab:design2_misclassification_app1} reports results for an alternative parameterization in which the Bernoulli variable $e_{1;i,j}$ in Design 2 is drawn with a success probability $0.5$ instead of $0.75$ (in the main text); that is, $\mathbb{P}(e_{1;i,j}=1)=0.5$. All other parameters of the data-generating process remain fixed at their baseline values.  By lowering the probability that true links in the latent adjacency matrix $\mathbf{W}_{0}^{\ast}$ are correctly retained in the observed network $\mathbf{W}$, this specification increases the degree of link misclassification relative to the baseline design. Table \ref{tab:design2_misclassification_app2} exhibits results for $\mathbb{P}(e_{1;i,j}=1)=0.8$, which decreases the degree of link misclassified relative to the baseline design instead.
    
Table \ref{tab:design3_homophily_app1} reports results for an alternative parameterization in Design 3 in which the homophily cutoff is reduced from the 99\% to the 95\% empirical quantile, that is, $\widehat{F}_{\varepsilon_3^\ast}^{-1}(0.95)$ replaces $\widehat{F}_{\varepsilon_3^\ast}^{-1}(0.99)$ in the link formation rule. All other parameters of the data-generating process remain fixed at their baseline values. By lowering the quantile threshold, a larger set of individuals with extreme realizations of $\varepsilon_{3}^{\ast}$ influences network formation. This modification strengthens the endogenous sorting mechanism embedded in the design and increases the correlation between the network structure and the outcome disturbance. Table \ref{tab:design3_homophily_app2} exhibits results when the homophily cutoff is reduced from 99\% to 97.5\%. 

\begin{landscape}   
    \begin{table}[!htbp]
\centering
\small
\begin{threeparttable}
\caption{Estimator Robustness to Network Link Misclassification (Design 2)}
\label{tab:design2_misclassification_app1}
\begin{tabular}{ccc c cccc c cccc c cccc}
\toprule
 &  &  &  &\multicolumn{4}{c}{Peer effects} &  &\multicolumn{4}{c}{Contextual effects} &  &\multicolumn{4}{c}{Direct effects} \\
\cmidrule(lr){5-8} \cmidrule(lr){10-13} \cmidrule(lr){15-18}
$\tau$ & Estimator & $n$ &  & Bias & SD & RMSE & IQR &  & Bias & SD & RMSE & IQR &  & Bias & SD & RMSE & IQR \\
\midrule
\multirow{9}{*}{0.01} & \multirow{3}{*}{OLS} & 50 &  & -0.245 & 0.129 & 0.277 & 0.202 &  & -0.367 & 0.276 & 0.459 & 0.439 &  & 0.264 & 0.189 & 0.325 & 0.296 \\
 &  & 100 &  & -0.320 & 0.110 & 0.339 & 0.177 &  & -0.565 & 0.181 & 0.593 & 0.282 &  & 0.121 & 0.087 & 0.149 & 0.137 \\
 &  & 200 &  & -0.250 & 0.106 & 0.271 & 0.171 &  & -0.729 & 0.144 & 0.743 & 0.228 &  & 0.039 & 0.044 & 0.059 & 0.070 \\
\noalign{\vskip 0.5em}
 & \multirow{3}{*}{G2SLS} & 50 &  & -0.445 & 0.196 & 0.486 & 0.290 &  & -0.012 & 0.377 & 0.377 & 0.584 &  & 0.382 & 0.224 & 0.443 & 0.340 \\
 &  & 100 &  & -0.548 & 0.170 & 0.574 & 0.260 &  & -0.255 & 0.245 & 0.354 & 0.372 &  & 0.203 & 0.104 & 0.228 & 0.161 \\
 &  & 200 &  & -0.522 & 0.205 & 0.561 & 0.312 &  & -0.412 & 0.239 & 0.477 & 0.370 &  & 0.086 & 0.055 & 0.102 & 0.086 \\
\noalign{\vskip 0.5em}
 & \multirow{3}{*}{G3SLS} & 50 &  & 0.192 & 0.177 & 0.261 & 0.265 &  & 0.790 & 0.695 & 1.052 & 1.047 &  & 0.052 & 0.092 & 0.105 & 0.139 \\
 &  & 100 &  & 0.028 & 0.126 & 0.128 & 0.198 &  & 0.307 & 0.327 & 0.448 & 0.498 &  & 0.023 & 0.059 & 0.063 & 0.092 \\
 &  & 200 &  & -0.015 & 0.159 & 0.160 & 0.248 &  & 0.061 & 0.243 & 0.250 & 0.394 &  & 0.007 & 0.038 & 0.039 & 0.058 \\
\midrule
\multirow{9}{*}{0.05} & \multirow{3}{*}{OLS} & 50 &  & -0.245 & 0.129 & 0.277 & 0.202 &  & -0.367 & 0.276 & 0.459 & 0.439 &  & 0.264 & 0.189 & 0.325 & 0.296 \\
 &  & 100 &  & -0.320 & 0.110 & 0.339 & 0.177 &  & -0.565 & 0.181 & 0.593 & 0.282 &  & 0.121 & 0.087 & 0.149 & 0.137 \\
 &  & 200 &  & -0.250 & 0.106 & 0.271 & 0.171 &  & -0.729 & 0.144 & 0.743 & 0.228 &  & 0.039 & 0.044 & 0.059 & 0.070 \\
\noalign{\vskip 0.5em}
 & \multirow{3}{*}{G2SLS} & 50 &  & -0.445 & 0.196 & 0.486 & 0.290 &  & -0.012 & 0.377 & 0.377 & 0.584 &  & 0.382 & 0.224 & 0.443 & 0.340 \\
 &  & 100 &  & -0.548 & 0.170 & 0.574 & 0.260 &  & -0.255 & 0.245 & 0.354 & 0.372 &  & 0.203 & 0.104 & 0.228 & 0.161 \\
 &  & 200 &  & -0.522 & 0.205 & 0.561 & 0.312 &  & -0.412 & 0.239 & 0.477 & 0.370 &  & 0.086 & 0.055 & 0.102 & 0.086 \\
\noalign{\vskip 0.5em}
 & \multirow{3}{*}{G3SLS} & 50 &  & 0.137 & 0.199 & 0.241 & 0.294 &  & 0.882 & 0.724 & 1.141 & 1.059 &  & 0.084 & 0.114 & 0.142 & 0.175 \\
 &  & 100 &  & -0.018 & 0.142 & 0.143 & 0.217 &  & 0.371 & 0.358 & 0.515 & 0.559 &  & 0.038 & 0.067 & 0.077 & 0.106 \\
 &  & 200 &  & -0.048 & 0.168 & 0.174 & 0.259 &  & 0.101 & 0.254 & 0.273 & 0.401 &  & 0.013 & 0.040 & 0.042 & 0.062 \\
\bottomrule
\end{tabular}
\begin{tablenotes}[para]
\footnotesize
\setstretch{1}
\item \textit{Notes:} This table reports Monte Carlo results under network link misclassification. The parameter $\tau$ controls the intensity of misclassification and is common across individuals. Reported statistics include bias, standard deviation (SD), root mean squared error (RMSE), and inter-quantile range (IQR) for the peer effect ($\beta$=0.7), contextual ($\delta=1$), and direct ($\gamma=1$) Effects.
\end{tablenotes}
\end{threeparttable}
\normalsize
\end{table} 
\end{landscape}

\begin{landscape}   
    \begin{table}[!htbp]
\centering
\small
\begin{threeparttable}
\caption{Estimator Robustness to Network Link Misclassification (Design 2)}
\label{tab:design2_misclassification_app2}
\begin{tabular}{ccc c cccc c cccc c cccc}
\toprule
 &  &  &  &\multicolumn{4}{c}{Peer effects} &  &\multicolumn{4}{c}{Contextual effects} &  &\multicolumn{4}{c}{Direct effects} \\
\cmidrule(lr){5-8} \cmidrule(lr){10-13} \cmidrule(lr){15-18}
$\tau$ & Estimator & $n$ &  & Bias & SD & RMSE & IQR &  & Bias & SD & RMSE & IQR &  & Bias & SD & RMSE & IQR \\
\midrule
\multirow{9}{*}{0.01} & \multirow{3}{*}{OLS} & 50 &  & -0.413 & 0.133 & 0.434 & 0.215 &  & -0.432 & 0.328 & 0.542 & 0.529 &  & 0.436 & 0.225 & 0.490 & 0.343 \\
 &  & 100 &  & -0.552 & 0.081 & 0.558 & 0.125 &  & -0.647 & 0.151 & 0.664 & 0.242 &  & 0.215 & 0.100 & 0.237 & 0.165 \\
 &  & 200 &  & -0.597 & 0.051 & 0.599 & 0.079 &  & -0.785 & 0.084 & 0.790 & 0.137 &  & 0.100 & 0.047 & 0.110 & 0.074 \\
\noalign{\vskip 0.5em}
 & \multirow{3}{*}{G2SLS} & 50 &  & -0.544 & 0.230 & 0.591 & 0.313 &  & -0.191 & 0.498 & 0.533 & 0.721 &  & 0.480 & 0.255 & 0.543 & 0.395 \\
 &  & 100 &  & -0.661 & 0.127 & 0.673 & 0.177 &  & -0.494 & 0.216 & 0.539 & 0.321 &  & 0.243 & 0.114 & 0.268 & 0.183 \\
 &  & 200 &  & -0.689 & 0.079 & 0.694 & 0.116 &  & -0.675 & 0.110 & 0.684 & 0.172 &  & 0.112 & 0.051 & 0.124 & 0.078 \\
\noalign{\vskip 0.5em}
 & \multirow{3}{*}{G3SLS} & 50 &  & 1.043 & 0.633 & 1.220 & 0.914 &  & 3.312 & 2.721 & 4.286 & 3.574 &  & 0.050 & 0.091 & 0.104 & 0.138 \\
 &  & 100 &  & 0.373 & 0.224 & 0.435 & 0.331 &  & 1.347 & 0.848 & 1.591 & 1.303 &  & 0.023 & 0.059 & 0.063 & 0.092 \\
 &  & 200 &  & 0.090 & 0.188 & 0.208 & 0.295 &  & 0.345 & 0.397 & 0.526 & 0.610 &  & 0.007 & 0.038 & 0.039 & 0.059 \\
\midrule
\multirow{9}{*}{0.05} & \multirow{3}{*}{OLS} & 50 &  & -0.413 & 0.133 & 0.434 & 0.215 &  & -0.432 & 0.328 & 0.542 & 0.529 &  & 0.436 & 0.225 & 0.490 & 0.343 \\
 &  & 100 &  & -0.552 & 0.081 & 0.558 & 0.125 &  & -0.647 & 0.151 & 0.664 & 0.242 &  & 0.215 & 0.100 & 0.237 & 0.165 \\
 &  & 200 &  & -0.597 & 0.051 & 0.599 & 0.079 &  & -0.785 & 0.084 & 0.790 & 0.137 &  & 0.100 & 0.047 & 0.110 & 0.074 \\
\noalign{\vskip 0.5em}
 & \multirow{3}{*}{G2SLS} & 50 &  & -0.544 & 0.230 & 0.591 & 0.313 &  & -0.191 & 0.498 & 0.533 & 0.721 &  & 0.480 & 0.255 & 0.543 & 0.395 \\
 &  & 100 &  & -0.661 & 0.127 & 0.673 & 0.177 &  & -0.494 & 0.216 & 0.539 & 0.321 &  & 0.243 & 0.114 & 0.268 & 0.183 \\
 &  & 200 &  & -0.689 & 0.079 & 0.694 & 0.116 &  & -0.675 & 0.110 & 0.684 & 0.172 &  & 0.112 & 0.051 & 0.124 & 0.078 \\
\noalign{\vskip 0.5em}
 & \multirow{3}{*}{G3SLS} & 50 &  & 0.930 & 0.659 & 1.140 & 0.959 &  & 3.534 & 2.791 & 4.503 & 3.641 &  & 0.083 & 0.114 & 0.140 & 0.174 \\
 &  & 100 &  & 0.304 & 0.246 & 0.391 & 0.362 &  & 1.437 & 0.886 & 1.688 & 1.357 &  & 0.038 & 0.067 & 0.077 & 0.105 \\
 &  & 200 &  & 0.051 & 0.197 & 0.203 & 0.304 &  & 0.399 & 0.409 & 0.572 & 0.649 &  & 0.013 & 0.040 & 0.042 & 0.062 \\
\bottomrule
\end{tabular}
\begin{tablenotes}[para]
\footnotesize
\setstretch{1}
\item \textit{Notes:} This table reports Monte Carlo results under network link misclassification. The parameter $\tau$ controls the intensity of misclassification and is common across individuals. Reported statistics include bias, standard deviation (SD), root mean squared error (RMSE), and inter-quantile range (IQR) for the peer effect ($\beta$=0.7), contextual ($\delta=1$), and direct ($\gamma=1$) Effects.
\end{tablenotes}
\end{threeparttable}
\normalsize
\end{table} 
\end{landscape}

\begin{landscape}
   \begin{table}[!htbp]
\centering
\small
\begin{threeparttable}
\caption{Estimator Performance under Unobserved Homophily (Design 3)}
\label{tab:design3_homophily_app1}
\begin{tabular}{ccc c cccc c cccc c cccc}
\toprule
 &  &  &  &\multicolumn{4}{c}{Peer effects} &  &\multicolumn{4}{c}{Contextual effects} &  &\multicolumn{4}{c}{Direct effects} \\
\cmidrule(lr){5-8} \cmidrule(lr){10-13} \cmidrule(lr){15-18}
$m$ & Estimator & $n$ &  & Bias & SD & RMSE & IQR &  & Bias & SD & RMSE & IQR &  & Bias & SD & RMSE & IQR \\
\midrule
\multirow{9}{*}{1} & \multirow{3}{*}{OLS} & 50 &  & 0.096 & 0.034 & 0.101 & 0.053 &  & -0.235 & 0.122 & 0.265 & 0.192 &  & -0.123 & 0.085 & 0.150 & 0.133 \\
 &  & 100 &  & 0.093 & 0.022 & 0.095 & 0.036 &  & -0.231 & 0.080 & 0.245 & 0.127 &  & -0.124 & 0.060 & 0.138 & 0.094 \\
 &  & 200 &  & 0.117 & 0.020 & 0.118 & 0.031 &  & -0.251 & 0.062 & 0.258 & 0.099 &  & -0.113 & 0.040 & 0.120 & 0.064 \\
\noalign{\vskip 0.5em}
 & \multirow{3}{*}{G2SLS} & 50 &  & 0.047 & 0.038 & 0.061 & 0.056 &  & -0.119 & 0.130 & 0.176 & 0.207 &  & -0.060 & 0.090 & 0.108 & 0.139 \\
 &  & 100 &  & 0.041 & 0.025 & 0.048 & 0.040 &  & -0.100 & 0.086 & 0.132 & 0.135 &  & -0.053 & 0.064 & 0.083 & 0.104 \\
 &  & 200 &  & 0.060 & 0.021 & 0.064 & 0.033 &  & -0.129 & 0.064 & 0.143 & 0.099 &  & -0.059 & 0.040 & 0.072 & 0.064 \\
\noalign{\vskip 0.5em}
 & \multirow{3}{*}{G3SLS} & 50 &  & -0.037 & 0.075 & 0.084 & 0.116 &  & 0.088 & 0.197 & 0.216 & 0.305 &  & 0.020 & 0.132 & 0.134 & 0.202 \\
 &  & 100 &  & -0.031 & 0.049 & 0.058 & 0.077 &  & 0.080 & 0.141 & 0.162 & 0.220 &  & 0.011 & 0.096 & 0.097 & 0.149 \\
 &  & 200 &  & -0.035 & 0.043 & 0.056 & 0.067 &  & 0.075 & 0.105 & 0.129 & 0.166 &  & 0.009 & 0.061 & 0.062 & 0.094 \\
\midrule
\multirow{9}{*}{3} & \multirow{3}{*}{OLS} & 50 &  & 0.165 & 0.042 & 0.171 & 0.065 &  & -0.405 & 0.177 & 0.442 & 0.276 &  & -0.212 & 0.150 & 0.259 & 0.233 \\
 &  & 100 &  & 0.163 & 0.029 & 0.165 & 0.045 &  & -0.402 & 0.113 & 0.418 & 0.177 &  & -0.223 & 0.102 & 0.245 & 0.162 \\
 &  & 200 &  & 0.209 & 0.024 & 0.210 & 0.038 &  & -0.448 & 0.085 & 0.456 & 0.137 &  & -0.204 & 0.068 & 0.215 & 0.110 \\
\noalign{\vskip 0.5em}
 & \multirow{3}{*}{G2SLS} & 50 &  & 0.128 & 0.053 & 0.138 & 0.078 &  & -0.307 & 0.210 & 0.372 & 0.324 &  & -0.161 & 0.162 & 0.229 & 0.258 \\
 &  & 100 &  & 0.114 & 0.035 & 0.119 & 0.055 &  & -0.275 & 0.131 & 0.305 & 0.208 &  & -0.156 & 0.113 & 0.193 & 0.175 \\
 &  & 200 &  & 0.169 & 0.030 & 0.171 & 0.044 &  & -0.359 & 0.095 & 0.371 & 0.143 &  & -0.165 & 0.071 & 0.179 & 0.114 \\
\noalign{\vskip 0.5em}
 & \multirow{3}{*}{G3SLS} & 50 &  & -0.049 & 0.129 & 0.138 & 0.191 &  & 0.112 & 0.328 & 0.346 & 0.515 &  & 0.035 & 0.214 & 0.217 & 0.335 \\
 &  & 100 &  & -0.041 & 0.079 & 0.089 & 0.120 &  & 0.108 & 0.228 & 0.252 & 0.350 &  & 0.021 & 0.149 & 0.150 & 0.230 \\
 &  & 200 &  & -0.050 & 0.071 & 0.087 & 0.107 &  & 0.106 & 0.176 & 0.206 & 0.271 &  & 0.015 & 0.105 & 0.106 & 0.168 \\
\bottomrule
\end{tabular}
\begin{tablenotes}[para]
\footnotesize
\setstretch{1}
\item \textit{Notes:} This table reports Monte Carlo results under unobserved homophily. The parameter $m$ controls the intensity of homophilous link formation and is common across individuals. Reported statistics include bias, standard deviation (SD), root mean squared error (RMSE), and inter-quantile range (IQR) for the peer effect ($\beta$=0.7), contextual ($\delta=1$), and direct ($\gamma=1$) Effects.
\end{tablenotes}
\end{threeparttable}
\normalsize
\end{table} 
\end{landscape}

\begin{landscape}
   \begin{table}[!htbp]
\centering
\small
\begin{threeparttable}
\caption{Estimator Performance under Unobserved Homophily (Design 3)}
\label{tab:design3_homophily_app2}
\begin{tabular}{ccc c cccc c cccc c cccc}
\toprule
 &  &  &  &\multicolumn{4}{c}{Peer effects} &  &\multicolumn{4}{c}{Contextual effects} &  &\multicolumn{4}{c}{Direct effects} \\
\cmidrule(lr){5-8} \cmidrule(lr){10-13} \cmidrule(lr){15-18}
$m$ & Estimator & $n$ &  & Bias & SD & RMSE & IQR &  & Bias & SD & RMSE & IQR &  & Bias & SD & RMSE & IQR \\
\midrule
\multirow{9}{*}{1} & \multirow{3}{*}{OLS} & 50 &  & 0.087 & 0.033 & 0.093 & 0.052 &  & -0.221 & 0.122 & 0.253 & 0.199 &  & -0.119 & 0.085 & 0.146 & 0.134 \\
 &  & 100 &  & 0.082 & 0.022 & 0.085 & 0.035 &  & -0.212 & 0.079 & 0.226 & 0.131 &  & -0.118 & 0.060 & 0.132 & 0.092 \\
 &  & 200 &  & 0.094 & 0.019 & 0.096 & 0.030 &  & -0.209 & 0.062 & 0.218 & 0.096 &  & -0.100 & 0.039 & 0.107 & 0.063 \\
\noalign{\vskip 0.5em}
 & \multirow{3}{*}{G2SLS} & 50 &  & 0.035 & 0.037 & 0.051 & 0.061 &  & -0.090 & 0.128 & 0.157 & 0.194 &  & -0.046 & 0.090 & 0.101 & 0.142 \\
 &  & 100 &  & 0.027 & 0.025 & 0.037 & 0.042 &  & -0.069 & 0.086 & 0.110 & 0.133 &  & -0.037 & 0.063 & 0.073 & 0.102 \\
 &  & 200 &  & 0.035 & 0.022 & 0.041 & 0.033 &  & -0.078 & 0.065 & 0.101 & 0.105 &  & -0.038 & 0.040 & 0.055 & 0.062 \\
\noalign{\vskip 0.5em}
 & \multirow{3}{*}{G3SLS} & 50 &  & -0.023 & 0.061 & 0.065 & 0.091 &  & 0.058 & 0.175 & 0.184 & 0.272 &  & 0.018 & 0.118 & 0.119 & 0.182 \\
 &  & 100 &  & -0.016 & 0.039 & 0.042 & 0.061 &  & 0.046 & 0.118 & 0.126 & 0.186 &  & 0.008 & 0.079 & 0.080 & 0.122 \\
 &  & 200 &  & -0.014 & 0.032 & 0.035 & 0.049 &  & 0.031 & 0.085 & 0.090 & 0.129 &  & 0.003 & 0.051 & 0.051 & 0.081 \\
\midrule
\multirow{9}{*}{3} & \multirow{3}{*}{OLS} & 50 &  & 0.144 & 0.042 & 0.150 & 0.063 &  & -0.363 & 0.170 & 0.401 & 0.273 &  & -0.197 & 0.142 & 0.242 & 0.227 \\
 &  & 100 &  & 0.134 & 0.027 & 0.137 & 0.044 &  & -0.344 & 0.107 & 0.360 & 0.167 &  & -0.196 & 0.095 & 0.217 & 0.148 \\
 &  & 200 &  & 0.155 & 0.024 & 0.157 & 0.037 &  & -0.344 & 0.083 & 0.354 & 0.132 &  & -0.164 & 0.062 & 0.175 & 0.097 \\
\noalign{\vskip 0.5em}
 & \multirow{3}{*}{G2SLS} & 50 &  & 0.096 & 0.052 & 0.110 & 0.077 &  & -0.239 & 0.198 & 0.311 & 0.301 &  & -0.130 & 0.154 & 0.202 & 0.236 \\
 &  & 100 &  & 0.079 & 0.035 & 0.087 & 0.057 &  & -0.197 & 0.123 & 0.233 & 0.192 &  & -0.114 & 0.106 & 0.156 & 0.163 \\
 &  & 200 &  & 0.103 & 0.028 & 0.107 & 0.044 &  & -0.228 & 0.087 & 0.244 & 0.134 &  & -0.109 & 0.065 & 0.127 & 0.103 \\
\noalign{\vskip 0.5em}
 & \multirow{3}{*}{G3SLS} & 50 &  & -0.033 & 0.099 & 0.105 & 0.142 &  & 0.078 & 0.271 & 0.282 & 0.418 &  & 0.033 & 0.179 & 0.182 & 0.276 \\
 &  & 100 &  & -0.022 & 0.058 & 0.062 & 0.089 &  & 0.063 & 0.179 & 0.190 & 0.282 &  & 0.015 & 0.117 & 0.118 & 0.182 \\
 &  & 200 &  & -0.019 & 0.048 & 0.052 & 0.074 &  & 0.042 & 0.131 & 0.138 & 0.202 &  & 0.005 & 0.077 & 0.077 & 0.121 \\
\bottomrule
\end{tabular}
\begin{tablenotes}[para]
\footnotesize
\setstretch{1}
\item \textit{Notes:} This table reports Monte Carlo results under unobserved homophily. The parameter $m$ controls the intensity of homophilous link formation and is common across individuals. Reported statistics include bias, standard deviation (SD), root mean squared error (RMSE), and inter-quantile range (IQR) for the peer effect ($\beta$=0.7), contextual ($\delta=1$), and direct ($\gamma=1$) Effects.
\end{tablenotes}
\end{threeparttable}
\normalsize
\end{table} 
\end{landscape}

\subsection{Alternative Network Densities}

This section provides Monte Carlo evidence of the finite-sample performance of OLS, G2SLS, and G3SLS mentioned in the main manuscript across the three Monte Carlo designs under alternative network densities. In particular, we vary the density parameter from relatively sparse to increasingly dense network structures, while maintaining the remaining elements of the data-generating process at their baseline values. This exercise allows us to examine how the estimators behave as the intensity of network connections increases across the different structural environments considered.

Table \ref{tab:design1_heterogeneity_test1} reports results for an alternative parameterization of the network formation process in which the logistic disturbance $u_{ij}$ in Design 1 is drawn from a Logistic distribution with a location parameter $\mu_u=-0.05$ instead of the baseline value $\mu_u=0$. All other parameters of the data-generating process remain fixed at their baseline values. A negative shift in the location parameter reduces the typical realization of $u_{ij}$, thereby increasing the probability that $u_{ij} \le z_i z_j + \psi(a_i+a_j)$, and hence increasing the expected density of both the endogenous network $\mathbf{W}$ and the exogenous network $\mathbf{W}_0$ relative to the baseline design. For the same design, Table \ref{tab:design1_heterogeneity_test2} exhibits results for $\mu_u=+0.05$, which shifts the logistic distribution to the right and decreases the probability of link formation. This specification therefore generates sparser networks relative to the baseline design.

Table \ref{tab:design2_misclassification_test1} reports results for an alternative parameterization of the latent network formation process in Design 2, in which the Erd\"{o}s–R\'{e}nyi probability governing $\mathbf{W}_0^{\ast}$ is set equal to 0.04 instead of the baseline value of 0.05. All other parameters of the data-generating process remain fixed at their baseline values. Reducing the link probability lowers the expected density of the true adjacency matrix $\mathbf{W}_0^{\ast}$ and, consequently, of both the observed endogenous network $\mathbf{W}$ and the exogenous network $\mathbf{W}_0$. This specification, therefore, evaluates estimator performance in a sparser-network environment while preserving the same misclassification mechanism and outcome equation. Similarly, Table \ref{tab:design2_misclassification_test2} exhibits results for a higher link probability of 0.06 for the same design, which increases the expected density of the latent adjacency matrix $\mathbf{W}_0^{\ast}$ relative to the baseline design. As a result, both $\mathbf{W}$ and $\mathbf{W}_0$ become denser in expectation.

Finally, Table \ref{tab:design3_homophily_test1} reports the results for an alternative parameterization in Design 3, in which the Erd\"{o}s–R\'{e}nyi link probability governing the exogenous adjacency matrix $\mathbf{W}_0$ is set equal to 0.005 instead of the baseline value 0.01. All other parameters of the data-generating process remain fixed at their baseline values. Reducing the link probability lowers the expected density of $\mathbf{W}_0$, and given the construction of $\mathbf{W}$, also reduces the overall density of the endogenous network. Lastly, Table \ref{tab:design3_homophily_test2} exhibits results for a higher link probability of 0.015, which increases the expected density of the exogenous adjacency matrix $\mathbf{W}_0$ relative to the baseline design. Because $\mathbf{W}$ is constructed by modifying $\mathbf{W}_0$ according to the homophily rule, this change also leads to a denser endogenous network, in expectation.

\subsubsection*{Results}

The Monte Carlo results reported in Tables \ref{tab:design1_heterogeneity_test1} and \ref{tab:design1_heterogeneity_test2} allow for an assessment of the G3SLS estimator's sensitivity to variations in network density in the presence of unobserved degree heterogeneity (Design 1), relative to the baseline results presented in Table \ref{tab:design1_heterogeneity_main}. A comparison of these alternative parameterizations reveals that the proposed estimator remains robust to fluctuations in the expected density of the endogenous and exogenous networks. Specifically, Table \ref{tab:design1_heterogeneity_test1} indicates that in a denser network environment ($\mu_u = -0.05$), the G3SLS estimator exhibits a marginal increase in the Root Mean Squared Error (RMSE) for the peer effects parameter ($\beta=0.7$) compared to the baseline; yet, it continues to significantly outperform the OLS and G2SLS estimators in terms of bias reduction. Conversely, the results for the sparser network specification in Table \ref{tab:design1_heterogeneity_test2} ($\mu_u = +0.05$) demonstrate that the finite-sample performance of the G3SLS estimator is preserved, yielding RMSE and bias metrics that are quantitatively similar to those observed in the baseline design.

Comparing these results to the baseline specification in Table \ref{tab:design2_misclassification_main} (density 0.05) reveals how the intensity of network connections interacts with the link misclassification mechanism. In a sparser network environment (Table \ref{tab:design2_misclassification_test1}, density 0.04), the finite-sample bias of the G3SLS estimator for the peer and contextual effects is exacerbated relative to the baseline, reflecting the increased sensitivity of the estimation when fewer true links are available to offset the misclassification rate $\tau$. Conversely, increasing the expected network density (Table \ref{tab:design2_misclassification_test2}, density 0.06) mitigates the finite-sample bias for the G3SLS estimator across all structural parameters, though it introduces a marginal increase in the root mean squared error for the peer effects parameter driven by a higher standard deviation for larger samples. Crucially, across both sparser and denser specifications, the G3SLS estimator systematically corrects for the severe bias that persists in both the naive OLS and conventional G2SLS estimators, demonstrating its robustness to measurement error regardless of the underlying network density.

A comparison of these results against the baseline specification in Table \ref{tab:design3_homophily_main} reveals that the finite-sample performance of the proposed estimator is remarkably insensitive to variations in the density of the exogenous network layer. Across all specifications — ranging from the sparser environment in Table \ref{tab:design3_homophily_test1} (density 0.005) to the denser structure in \ref{tab:design3_homophily_test2} (density 0.015) — the G3SLS estimator consistently delivers negligible bias and stable root mean squared error values that are virtually indistinguishable from the baseline results. This invariance stands in sharp contrast to the persistent bias observed in the OLS and G2SLS estimators, providing strong evidence that the identification strategy's ability to correct for unobserved homophily remains robust regardless of the sparsity or density of the underlying network connections for this design.

\begin{landscape}  
    \begin{table}[!htbp]
\centering
\small
\begin{threeparttable}
\caption{Estimator Performance under Unobserved Degree Heterogeneity (Design 1)}
\label{tab:design1_heterogeneity_test1}
\begin{tabular}{ccc c cccc c cccc c cccc}
\toprule
 &  &  &  &\multicolumn{4}{c}{Peer effects} &  &\multicolumn{4}{c}{Contextual effects} &  &\multicolumn{4}{c}{Direct effects} \\
\cmidrule(lr){5-8} \cmidrule(lr){10-13} \cmidrule(lr){15-18}
$m$ & Estimator & $n$ &  & Bias & SD & RMSE & IQR &  & Bias & SD & RMSE & IQR &  & Bias & SD & RMSE & IQR \\
\midrule
\multirow{9}{*}{10} & \multirow{3}{*}{OLS} & 50 &  & 0.355 & 0.045 & 0.358 & 0.070 &  & -1.103 & 0.880 & 1.411 & 1.363 &  & -0.449 & 0.617 & 0.763 & 0.975 \\
 &  & 100 &  & 0.348 & 0.030 & 0.349 & 0.045 &  & -1.087 & 0.579 & 1.232 & 0.939 &  & -0.442 & 0.420 & 0.609 & 0.639 \\
 &  & 200 &  & 0.345 & 0.021 & 0.346 & 0.033 &  & -1.096 & 0.405 & 1.168 & 0.641 &  & -0.471 & 0.290 & 0.553 & 0.439 \\
\noalign{\vskip 0.5em}
 & \multirow{3}{*}{G2SLS} & 50 &  & 0.687 & 0.337 & 0.765 & 0.317 &  & -1.548 & 2.303 & 2.774 & 3.109 &  & -0.562 & 1.117 & 1.250 & 1.566 \\
 &  & 100 &  & 0.759 & 0.288 & 0.812 & 0.290 &  & -2.182 & 1.851 & 2.861 & 2.566 &  & -0.846 & 0.851 & 1.199 & 1.215 \\
 &  & 200 &  & 0.804 & 0.289 & 0.855 & 0.301 &  & -2.474 & 1.593 & 2.942 & 2.225 &  & -1.020 & 0.740 & 1.260 & 1.021 \\
\noalign{\vskip 0.5em}
 & \multirow{3}{*}{G3SLS} & 50 &  & 0.062 & 1.008 & 1.010 & 1.287 &  & 0.339 & 6.628 & 6.635 & 7.889 &  & 0.509 & 1.622 & 1.699 & 2.429 \\
 &  & 100 &  & 0.022 & 0.789 & 0.789 & 1.050 &  & 0.215 & 4.100 & 4.104 & 5.341 &  & 0.556 & 1.086 & 1.219 & 1.661 \\
 &  & 200 &  & -0.032 & 0.582 & 0.582 & 0.821 &  & 0.259 & 2.473 & 2.486 & 3.472 &  & 0.575 & 0.722 & 0.923 & 1.094 \\
\midrule
\multirow{9}{*}{12} & \multirow{3}{*}{OLS} & 50 &  & 0.357 & 0.047 & 0.360 & 0.072 &  & -1.114 & 1.049 & 1.530 & 1.642 &  & -0.453 & 0.738 & 0.866 & 1.153 \\
 &  & 100 &  & 0.349 & 0.031 & 0.351 & 0.046 &  & -1.092 & 0.691 & 1.292 & 1.121 &  & -0.441 & 0.504 & 0.669 & 0.772 \\
 &  & 200 &  & 0.346 & 0.022 & 0.347 & 0.034 &  & -1.103 & 0.483 & 1.204 & 0.754 &  & -0.475 & 0.347 & 0.589 & 0.524 \\
\noalign{\vskip 0.5em}
 & \multirow{3}{*}{G2SLS} & 50 &  & 0.730 & 0.325 & 0.799 & 0.318 &  & -1.675 & 2.697 & 3.174 & 3.449 &  & -0.601 & 1.300 & 1.431 & 1.797 \\
 &  & 100 &  & 0.806 & 0.251 & 0.844 & 0.282 &  & -2.221 & 2.009 & 2.994 & 2.706 &  & -0.864 & 0.952 & 1.285 & 1.418 \\
 &  & 200 &  & 0.849 & 0.227 & 0.879 & 0.269 &  & -2.617 & 1.631 & 3.083 & 2.334 &  & -1.088 & 0.768 & 1.332 & 1.083 \\
\noalign{\vskip 0.5em}
 & \multirow{3}{*}{G3SLS} & 50 &  & 0.093 & 1.014 & 1.018 & 1.292 &  & 0.290 & 7.743 & 7.745 & 9.361 &  & 0.467 & 1.923 & 1.978 & 2.918 \\
 &  & 100 &  & 0.107 & 0.838 & 0.845 & 1.105 &  & -0.133 & 4.884 & 4.884 & 6.503 &  & 0.502 & 1.280 & 1.374 & 1.922 \\
 &  & 200 &  & 0.031 & 0.621 & 0.621 & 0.876 &  & 0.166 & 2.917 & 2.921 & 3.901 &  & 0.557 & 0.862 & 1.026 & 1.302 \\
\bottomrule
\end{tabular}
\begin{tablenotes}[para]
\footnotesize
\setstretch{1}
\item \textit{Notes:} This table reports Monte Carlo results under unobserved degree heterogeneity. The parameter $m$ controls the intensity of heterogeneity in node degrees and is common across individuals. Reported statistics include bias, standard deviation (SD), root mean squared error (RMSE), and inter-quantile range (IQR) for the peer effect ($\beta$=0.7), contextual ($\delta=1$), and direct ($\gamma=1$) Effects.
\end{tablenotes}
\end{threeparttable}
\normalsize
\end{table}
\end{landscape}

\begin{landscape}   
    \begin{table}[!htbp]
\centering
\small
\begin{threeparttable}
\caption{Estimator Performance under Unobserved Degree Heterogeneity (Design 1)}
\label{tab:design1_heterogeneity_test2}
\begin{tabular}{ccc c cccc c cccc c cccc}
\toprule
 &  &  &  &\multicolumn{4}{c}{Peer effects} &  &\multicolumn{4}{c}{Contextual effects} &  &\multicolumn{4}{c}{Direct effects} \\
\cmidrule(lr){5-8} \cmidrule(lr){10-13} \cmidrule(lr){15-18}
$m$ & Estimator & $n$ &  & Bias & SD & RMSE & IQR &  & Bias & SD & RMSE & IQR &  & Bias & SD & RMSE & IQR \\
\midrule
\multirow{9}{*}{10} & \multirow{3}{*}{OLS} & 50 &  & 0.363 & 0.047 & 0.366 & 0.073 &  & -1.120 & 0.873 & 1.420 & 1.329 &  & -0.450 & 0.618 & 0.764 & 0.989 \\
 &  & 100 &  & 0.356 & 0.031 & 0.357 & 0.050 &  & -1.105 & 0.580 & 1.248 & 0.917 &  & -0.441 & 0.415 & 0.606 & 0.648 \\
 &  & 200 &  & 0.352 & 0.022 & 0.353 & 0.034 &  & -1.128 & 0.405 & 1.199 & 0.625 &  & -0.471 & 0.287 & 0.551 & 0.435 \\
\noalign{\vskip 0.5em}
 & \multirow{3}{*}{G2SLS} & 50 &  & 0.684 & 0.300 & 0.747 & 0.294 &  & -1.581 & 2.173 & 2.687 & 2.857 &  & -0.570 & 1.079 & 1.220 & 1.446 \\
 &  & 100 &  & 0.732 & 0.259 & 0.776 & 0.259 &  & -2.117 & 1.778 & 2.764 & 2.346 &  & -0.811 & 0.803 & 1.141 & 1.122 \\
 &  & 200 &  & 0.778 & 0.251 & 0.817 & 0.260 &  & -2.430 & 1.450 & 2.830 & 2.139 &  & -0.973 & 0.672 & 1.182 & 0.894 \\
\noalign{\vskip 0.5em}
 & \multirow{3}{*}{G3SLS} & 50 &  & 0.119 & 0.976 & 0.982 & 1.257 &  & 0.502 & 6.845 & 6.861 & 8.589 &  & 0.520 & 1.570 & 1.653 & 2.428 \\
 &  & 100 &  & 0.009 & 0.764 & 0.764 & 1.023 &  & 0.189 & 4.134 & 4.137 & 5.388 &  & 0.551 & 1.039 & 1.176 & 1.543 \\
 &  & 200 &  & -0.014 & 0.546 & 0.546 & 0.794 &  & 0.199 & 2.489 & 2.496 & 3.451 &  & 0.544 & 0.716 & 0.899 & 1.112 \\
\midrule
\multirow{9}{*}{12} & \multirow{3}{*}{OLS} & 50 &  & 0.365 & 0.049 & 0.368 & 0.075 &  & -1.133 & 1.040 & 1.537 & 1.606 &  & -0.456 & 0.739 & 0.868 & 1.170 \\
 &  & 100 &  & 0.358 & 0.032 & 0.360 & 0.052 &  & -1.108 & 0.691 & 1.306 & 1.088 &  & -0.441 & 0.498 & 0.665 & 0.775 \\
 &  & 200 &  & 0.354 & 0.023 & 0.355 & 0.035 &  & -1.138 & 0.482 & 1.235 & 0.740 &  & -0.476 & 0.344 & 0.587 & 0.532 \\
\noalign{\vskip 0.5em}
 & \multirow{3}{*}{G2SLS} & 50 &  & 0.733 & 0.286 & 0.786 & 0.288 &  & -1.732 & 2.544 & 3.077 & 3.250 &  & -0.628 & 1.261 & 1.408 & 1.726 \\
 &  & 100 &  & 0.784 & 0.230 & 0.817 & 0.246 &  & -2.211 & 1.942 & 2.942 & 2.603 &  & -0.852 & 0.897 & 1.237 & 1.311 \\
 &  & 200 &  & 0.799 & 0.192 & 0.822 & 0.217 &  & -2.453 & 1.449 & 2.849 & 2.125 &  & -0.988 & 0.701 & 1.211 & 0.972 \\
\noalign{\vskip 0.5em}
 & \multirow{3}{*}{G3SLS} & 50 &  & 0.098 & 0.992 & 0.997 & 1.288 &  & 0.341 & 8.011 & 8.016 & 9.823 &  & 0.504 & 1.858 & 1.924 & 2.815 \\
 &  & 100 &  & 0.093 & 0.823 & 0.828 & 1.083 &  & 0.193 & 4.958 & 4.960 & 6.639 &  & 0.564 & 1.257 & 1.377 & 1.926 \\
 &  & 200 &  & 0.045 & 0.592 & 0.593 & 0.812 &  & 0.103 & 2.888 & 2.889 & 3.988 &  & 0.539 & 0.849 & 1.006 & 1.294 \\
\bottomrule
\end{tabular}
\begin{tablenotes}[para]
\footnotesize
\setstretch{1}
\item \textit{Notes:} This table reports Monte Carlo results under unobserved degree heterogeneity. The parameter $m$ controls the intensity of heterogeneity in node degrees and is common across individuals. Reported statistics include bias, standard deviation (SD), root mean squared error (RMSE), and inter-quantile range (IQR) for the peer effect ($\beta$=0.7), contextual ($\delta=1$), and direct ($\gamma=1$) Effects.
\end{tablenotes}
\end{threeparttable}
\normalsize
\end{table}
\end{landscape}

\begin{landscape}   
    \begin{table}[!htbp]
\centering
\small
\begin{threeparttable}
\caption{Estimator Robustness to Network Link Misclassification (Design 2)}
\label{tab:design2_misclassification_test1}
\begin{tabular}{ccc c cccc c cccc c cccc}
\toprule
 &  &  &  &\multicolumn{4}{c}{Peer effects} &  &\multicolumn{4}{c}{Contextual effects} &  &\multicolumn{4}{c}{Direct effects} \\
\cmidrule(lr){5-8} \cmidrule(lr){10-13} \cmidrule(lr){15-18}
$\tau$ & Estimator & $n$ &  & Bias & SD & RMSE & IQR &  & Bias & SD & RMSE & IQR &  & Bias & SD & RMSE & IQR \\
\midrule
\multirow{9}{*}{0.01} & \multirow{3}{*}{OLS} & 50 &  & -0.323 & 0.132 & 0.349 & 0.213 &  & -0.346 & 0.349 & 0.491 & 0.553 &  & 0.490 & 0.242 & 0.546 & 0.379 \\
 &  & 100 &  & -0.483 & 0.090 & 0.491 & 0.140 &  & -0.574 & 0.165 & 0.597 & 0.264 &  & 0.269 & 0.116 & 0.292 & 0.181 \\
 &  & 200 &  & -0.553 & 0.062 & 0.556 & 0.100 &  & -0.735 & 0.093 & 0.741 & 0.147 &  & 0.124 & 0.052 & 0.134 & 0.083 \\
\noalign{\vskip 0.5em}
 & \multirow{3}{*}{G2SLS} & 50 &  & -0.483 & 0.230 & 0.535 & 0.338 &  & -0.050 & 0.512 & 0.514 & 0.757 &  & 0.560 & 0.275 & 0.624 & 0.430 \\
 &  & 100 &  & -0.630 & 0.137 & 0.645 & 0.195 &  & -0.354 & 0.232 & 0.423 & 0.354 &  & 0.317 & 0.131 & 0.343 & 0.204 \\
 &  & 200 &  & -0.687 & 0.084 & 0.692 & 0.131 &  & -0.573 & 0.118 & 0.585 & 0.180 &  & 0.151 & 0.056 & 0.161 & 0.088 \\
\noalign{\vskip 0.5em}
 & \multirow{3}{*}{G3SLS} & 50 &  & 0.950 & 0.567 & 1.106 & 0.841 &  & 2.916 & 2.351 & 3.745 & 3.388 &  & 0.063 & 0.104 & 0.121 & 0.156 \\
 &  & 100 &  & 0.353 & 0.202 & 0.407 & 0.319 &  & 1.403 & 0.816 & 1.623 & 1.222 &  & 0.037 & 0.064 & 0.074 & 0.098 \\
 &  & 200 &  & 0.086 & 0.150 & 0.173 & 0.235 &  & 0.379 & 0.352 & 0.517 & 0.554 &  & 0.012 & 0.040 & 0.041 & 0.063 \\
\midrule
\multirow{9}{*}{0.05} & \multirow{3}{*}{OLS} & 50 &  & -0.323 & 0.132 & 0.349 & 0.213 &  & -0.346 & 0.349 & 0.491 & 0.553 &  & 0.490 & 0.242 & 0.546 & 0.379 \\
 &  & 100 &  & -0.483 & 0.090 & 0.491 & 0.140 &  & -0.574 & 0.165 & 0.597 & 0.264 &  & 0.269 & 0.116 & 0.292 & 0.181 \\
 &  & 200 &  & -0.553 & 0.062 & 0.556 & 0.100 &  & -0.735 & 0.093 & 0.741 & 0.147 &  & 0.124 & 0.052 & 0.134 & 0.083 \\
\noalign{\vskip 0.5em}
 & \multirow{3}{*}{G2SLS} & 50 &  & -0.483 & 0.230 & 0.535 & 0.338 &  & -0.050 & 0.512 & 0.514 & 0.757 &  & 0.560 & 0.275 & 0.624 & 0.430 \\
 &  & 100 &  & -0.630 & 0.137 & 0.645 & 0.195 &  & -0.354 & 0.232 & 0.423 & 0.354 &  & 0.317 & 0.131 & 0.343 & 0.204 \\
 &  & 200 &  & -0.687 & 0.084 & 0.692 & 0.131 &  & -0.573 & 0.118 & 0.585 & 0.180 &  & 0.151 & 0.056 & 0.161 & 0.088 \\
\noalign{\vskip 0.5em}
 & \multirow{3}{*}{G3SLS} & 50 &  & 0.852 & 0.572 & 1.026 & 0.817 &  & 3.116 & 2.433 & 3.953 & 3.455 &  & 0.097 & 0.127 & 0.159 & 0.198 \\
 &  & 100 &  & 0.282 & 0.222 & 0.359 & 0.350 &  & 1.509 & 0.852 & 1.733 & 1.269 &  & 0.060 & 0.072 & 0.093 & 0.113 \\
 &  & 200 &  & 0.049 & 0.163 & 0.170 & 0.250 &  & 0.429 & 0.363 & 0.562 & 0.567 &  & 0.020 & 0.042 & 0.047 & 0.068 \\
\bottomrule
\end{tabular}
\begin{tablenotes}[para]
\footnotesize
\setstretch{1}
\item \textit{Notes:} This table reports Monte Carlo results under network link misclassification. The parameter $\tau$ controls the intensity of misclassification and is common across individuals. Reported statistics include bias, standard deviation (SD), root mean squared error (RMSE), and inter-quantile range (IQR) for the peer effect ($\beta$=0.7), contextual ($\delta=1$), and direct ($\gamma=1$) Effects.
\end{tablenotes}
\end{threeparttable}
\normalsize
\end{table}
\end{landscape}

\begin{landscape}   
    \begin{table}[!htbp]
\centering
\small
\begin{threeparttable}
\caption{Estimator Robustness to Network Link Misclassification (Design 2)}
\label{tab:design2_misclassification_test2}
\begin{tabular}{ccc c cccc c cccc c cccc}
\toprule
 &  &  &  &\multicolumn{4}{c}{Peer effects} &  &\multicolumn{4}{c}{Contextual effects} &  &\multicolumn{4}{c}{Direct effects} \\
\cmidrule(lr){5-8} \cmidrule(lr){10-13} \cmidrule(lr){15-18}
$\tau$ & Estimator & $n$ &  & Bias & SD & RMSE & IQR &  & Bias & SD & RMSE & IQR &  & Bias & SD & RMSE & IQR \\
\midrule
\multirow{9}{*}{0.01} & \multirow{3}{*}{OLS} & 50 &  & -0.429 & 0.128 & 0.448 & 0.207 &  & -0.491 & 0.276 & 0.563 & 0.441 &  & 0.341 & 0.194 & 0.393 & 0.305 \\
 &  & 100 &  & -0.550 & 0.081 & 0.556 & 0.131 &  & -0.676 & 0.135 & 0.689 & 0.209 &  & 0.162 & 0.088 & 0.184 & 0.135 \\
 &  & 200 &  & -0.550 & 0.065 & 0.554 & 0.104 &  & -0.807 & 0.098 & 0.812 & 0.157 &  & 0.073 & 0.043 & 0.084 & 0.067 \\
\noalign{\vskip 0.5em}
 & \multirow{3}{*}{G2SLS} & 50 &  & -0.574 & 0.209 & 0.611 & 0.296 &  & -0.246 & 0.415 & 0.482 & 0.628 &  & 0.397 & 0.224 & 0.456 & 0.350 \\
 &  & 100 &  & -0.664 & 0.123 & 0.675 & 0.179 &  & -0.523 & 0.192 & 0.557 & 0.292 &  & 0.191 & 0.099 & 0.215 & 0.158 \\
 &  & 200 &  & -0.686 & 0.108 & 0.695 & 0.160 &  & -0.651 & 0.137 & 0.665 & 0.212 &  & 0.091 & 0.047 & 0.103 & 0.072 \\
\noalign{\vskip 0.5em}
 & \multirow{3}{*}{G3SLS} & 50 &  & 0.581 & 0.399 & 0.705 & 0.560 &  & 2.079 & 1.644 & 2.650 & 2.397 &  & 0.039 & 0.086 & 0.095 & 0.135 \\
 &  & 100 &  & 0.170 & 0.190 & 0.255 & 0.285 &  & 0.714 & 0.574 & 0.916 & 0.870 &  & 0.017 & 0.057 & 0.060 & 0.088 \\
 &  & 200 &  & 0.008 & 0.204 & 0.204 & 0.316 &  & 0.181 & 0.346 & 0.390 & 0.537 &  & 0.006 & 0.038 & 0.038 & 0.060 \\
\midrule
\multirow{9}{*}{0.05} & \multirow{3}{*}{OLS} & 50 &  & -0.429 & 0.128 & 0.448 & 0.207 &  & -0.491 & 0.276 & 0.563 & 0.441 &  & 0.341 & 0.194 & 0.393 & 0.305 \\
 &  & 100 &  & -0.550 & 0.081 & 0.556 & 0.131 &  & -0.676 & 0.135 & 0.689 & 0.209 &  & 0.162 & 0.088 & 0.184 & 0.135 \\
 &  & 200 &  & -0.550 & 0.065 & 0.554 & 0.104 &  & -0.807 & 0.098 & 0.812 & 0.157 &  & 0.073 & 0.043 & 0.084 & 0.067 \\
\noalign{\vskip 0.5em}
 & \multirow{3}{*}{G2SLS} & 50 &  & -0.574 & 0.209 & 0.611 & 0.296 &  & -0.246 & 0.415 & 0.482 & 0.628 &  & 0.397 & 0.224 & 0.456 & 0.350 \\
 &  & 100 &  & -0.664 & 0.123 & 0.675 & 0.179 &  & -0.523 & 0.192 & 0.557 & 0.292 &  & 0.191 & 0.099 & 0.215 & 0.158 \\
 &  & 200 &  & -0.686 & 0.108 & 0.695 & 0.160 &  & -0.651 & 0.137 & 0.665 & 0.212 &  & 0.091 & 0.047 & 0.103 & 0.072 \\
\noalign{\vskip 0.5em}
 & \multirow{3}{*}{G3SLS} & 50 &  & 0.501 & 0.434 & 0.662 & 0.604 &  & 2.223 & 1.711 & 2.805 & 2.512 &  & 0.065 & 0.105 & 0.123 & 0.162 \\
 &  & 100 &  & 0.119 & 0.204 & 0.236 & 0.324 &  & 0.784 & 0.602 & 0.989 & 0.918 &  & 0.028 & 0.062 & 0.068 & 0.104 \\
 &  & 200 &  & -0.025 & 0.210 & 0.211 & 0.323 &  & 0.224 & 0.354 & 0.419 & 0.540 &  & 0.010 & 0.039 & 0.040 & 0.062 \\
\bottomrule
\end{tabular}
\begin{tablenotes}[para]
\footnotesize
\setstretch{1}
\item \textit{Notes:} This table reports Monte Carlo results under network link misclassification. The parameter $\tau$ controls the intensity of misclassification and is common across individuals. Reported statistics include bias, standard deviation (SD), root mean squared error (RMSE), and inter-quantile range (IQR) for the peer effect ($\beta$=0.7), contextual ($\delta=1$), and direct ($\gamma=1$) Effects.
\end{tablenotes}
\end{threeparttable}
\normalsize
\end{table}
\end{landscape}

\begin{landscape}   
    \begin{table}[!htbp]
\centering
\small
\begin{threeparttable}
\caption{Estimator Performance under Unobserved Homophily (Design 3)}
\label{tab:design3_homophily_test1}
\begin{tabular}{ccc c cccc c cccc c cccc}
\toprule
 &  &  &  &\multicolumn{4}{c}{Peer effects} &  &\multicolumn{4}{c}{Contextual effects} &  &\multicolumn{4}{c}{Direct effects} \\
\cmidrule(lr){5-8} \cmidrule(lr){10-13} \cmidrule(lr){15-18}
$m$ & Estimator & $n$ &  & Bias & SD & RMSE & IQR &  & Bias & SD & RMSE & IQR &  & Bias & SD & RMSE & IQR \\
\midrule
\multirow{9}{*}{1} & \multirow{3}{*}{OLS} & 50 &  & 0.079 & 0.032 & 0.085 & 0.051 &  & -0.192 & 0.115 & 0.224 & 0.182 &  & -0.099 & 0.082 & 0.129 & 0.132 \\
 &  & 100 &  & 0.072 & 0.022 & 0.075 & 0.035 &  & -0.191 & 0.080 & 0.207 & 0.128 &  & -0.109 & 0.057 & 0.122 & 0.093 \\
 &  & 200 &  & 0.073 & 0.015 & 0.074 & 0.022 &  & -0.194 & 0.055 & 0.202 & 0.090 &  & -0.109 & 0.039 & 0.116 & 0.062 \\
\noalign{\vskip 0.5em}
 & \multirow{3}{*}{G2SLS} & 50 &  & 0.021 & 0.040 & 0.045 & 0.063 &  & -0.048 & 0.124 & 0.133 & 0.194 &  & -0.022 & 0.088 & 0.090 & 0.138 \\
 &  & 100 &  & 0.010 & 0.027 & 0.029 & 0.042 &  & -0.027 & 0.089 & 0.093 & 0.141 &  & -0.014 & 0.062 & 0.063 & 0.096 \\
 &  & 200 &  & 0.010 & 0.018 & 0.020 & 0.027 &  & -0.027 & 0.059 & 0.065 & 0.092 &  & -0.014 & 0.041 & 0.043 & 0.069 \\
\noalign{\vskip 0.5em}
 & \multirow{3}{*}{G3SLS} & 50 &  & -0.012 & 0.050 & 0.051 & 0.076 &  & 0.036 & 0.148 & 0.152 & 0.232 &  & 0.011 & 0.100 & 0.101 & 0.160 \\
 &  & 100 &  & -0.007 & 0.031 & 0.032 & 0.049 &  & 0.018 & 0.099 & 0.100 & 0.155 &  & 0.005 & 0.067 & 0.067 & 0.105 \\
 &  & 200 &  & -0.005 & 0.021 & 0.021 & 0.033 &  & 0.013 & 0.068 & 0.070 & 0.106 &  & 0.003 & 0.045 & 0.046 & 0.071 \\
\midrule
\multirow{9}{*}{3} & \multirow{3}{*}{OLS} & 50 &  & 0.121 & 0.040 & 0.127 & 0.063 &  & -0.295 & 0.147 & 0.329 & 0.226 &  & -0.151 & 0.125 & 0.196 & 0.201 \\
 &  & 100 &  & 0.100 & 0.026 & 0.103 & 0.039 &  & -0.264 & 0.097 & 0.281 & 0.149 &  & -0.152 & 0.078 & 0.171 & 0.124 \\
 &  & 200 &  & 0.101 & 0.018 & 0.102 & 0.027 &  & -0.269 & 0.066 & 0.277 & 0.107 &  & -0.151 & 0.055 & 0.161 & 0.089 \\
\noalign{\vskip 0.5em}
 & \multirow{3}{*}{G2SLS} & 50 &  & 0.060 & 0.052 & 0.079 & 0.082 &  & -0.140 & 0.168 & 0.219 & 0.264 &  & -0.067 & 0.136 & 0.151 & 0.214 \\
 &  & 100 &  & 0.033 & 0.033 & 0.047 & 0.052 &  & -0.086 & 0.111 & 0.140 & 0.166 &  & -0.050 & 0.089 & 0.102 & 0.138 \\
 &  & 200 &  & 0.032 & 0.022 & 0.039 & 0.034 &  & -0.086 & 0.075 & 0.114 & 0.119 &  & -0.047 & 0.061 & 0.077 & 0.098 \\
\noalign{\vskip 0.5em}
 & \multirow{3}{*}{G3SLS} & 50 &  & -0.017 & 0.075 & 0.077 & 0.109 &  & 0.052 & 0.220 & 0.226 & 0.340 &  & 0.026 & 0.140 & 0.142 & 0.208 \\
 &  & 100 &  & -0.007 & 0.041 & 0.041 & 0.062 &  & 0.023 & 0.129 & 0.131 & 0.198 &  & 0.007 & 0.086 & 0.086 & 0.130 \\
 &  & 200 &  & -0.004 & 0.028 & 0.028 & 0.043 &  & 0.012 & 0.090 & 0.091 & 0.138 &  & 0.005 & 0.059 & 0.059 & 0.093 \\
\bottomrule
\end{tabular}
\begin{tablenotes}[para]
\footnotesize
\setstretch{1}
\item \textit{Notes:} This table reports Monte Carlo results under unobserved homophily. The parameter $m$ controls the intensity of homophilous link formation and is common across individuals. Reported statistics include bias, standard deviation (SD), root mean squared error (RMSE), and inter-quantile range (IQR) for the peer effect ($\beta$=0.7), contextual ($\delta=1$), and direct ($\gamma=1$) Effects.
\end{tablenotes}
\end{threeparttable}
\normalsize
\end{table}
\end{landscape}

\begin{landscape}   
    \begin{table}[!htbp]
\centering
\small
\begin{threeparttable}
\caption{Estimator Performance under Unobserved Homophily (Design 3)}
\label{tab:design3_homophily_test2}
\begin{tabular}{ccc c cccc c cccc c cccc}
\toprule
 &  &  &  &\multicolumn{4}{c}{Peer effects} &  &\multicolumn{4}{c}{Contextual effects} &  &\multicolumn{4}{c}{Direct effects} \\
\cmidrule(lr){5-8} \cmidrule(lr){10-13} \cmidrule(lr){15-18}
$m$ & Estimator & $n$ &  & Bias & SD & RMSE & IQR &  & Bias & SD & RMSE & IQR &  & Bias & SD & RMSE & IQR \\
\midrule
\multirow{9}{*}{1} & \multirow{3}{*}{OLS} & 50 &  & 0.076 & 0.032 & 0.083 & 0.051 &  & -0.199 & 0.114 & 0.229 & 0.181 &  & -0.115 & 0.083 & 0.142 & 0.134 \\
 &  & 100 &  & 0.072 & 0.024 & 0.076 & 0.038 &  & -0.176 & 0.083 & 0.195 & 0.131 &  & -0.094 & 0.054 & 0.108 & 0.084 \\
 &  & 200 &  & 0.108 & 0.026 & 0.111 & 0.041 &  & -0.206 & 0.067 & 0.216 & 0.102 &  & -0.076 & 0.035 & 0.084 & 0.055 \\
\noalign{\vskip 0.5em}
 & \multirow{3}{*}{G2SLS} & 50 &  & 0.017 & 0.038 & 0.042 & 0.057 &  & -0.045 & 0.123 & 0.131 & 0.197 &  & -0.026 & 0.089 & 0.093 & 0.138 \\
 &  & 100 &  & 0.010 & 0.028 & 0.030 & 0.043 &  & -0.023 & 0.089 & 0.092 & 0.142 &  & -0.011 & 0.058 & 0.059 & 0.093 \\
 &  & 200 &  & 0.033 & 0.030 & 0.044 & 0.046 &  & -0.064 & 0.074 & 0.097 & 0.119 &  & -0.021 & 0.038 & 0.043 & 0.060 \\
\noalign{\vskip 0.5em}
 & \multirow{3}{*}{G3SLS} & 50 &  & -0.012 & 0.048 & 0.050 & 0.074 &  & 0.030 & 0.147 & 0.150 & 0.238 &  & 0.006 & 0.100 & 0.101 & 0.149 \\
 &  & 100 &  & -0.008 & 0.033 & 0.034 & 0.052 &  & 0.021 & 0.100 & 0.102 & 0.156 &  & 0.008 & 0.063 & 0.064 & 0.097 \\
 &  & 200 &  & -0.004 & 0.035 & 0.035 & 0.055 &  & 0.008 & 0.084 & 0.085 & 0.138 &  & 0.005 & 0.042 & 0.042 & 0.067 \\
\midrule
\multirow{9}{*}{3} & \multirow{3}{*}{OLS} & 50 &  & 0.118 & 0.038 & 0.124 & 0.059 &  & -0.302 & 0.144 & 0.335 & 0.228 &  & -0.178 & 0.126 & 0.218 & 0.199 \\
 &  & 100 &  & 0.102 & 0.028 & 0.105 & 0.045 &  & -0.247 & 0.097 & 0.265 & 0.151 &  & -0.132 & 0.075 & 0.152 & 0.119 \\
 &  & 200 &  & 0.155 & 0.031 & 0.159 & 0.048 &  & -0.298 & 0.081 & 0.309 & 0.126 &  & -0.109 & 0.049 & 0.119 & 0.077 \\
\noalign{\vskip 0.5em}
 & \multirow{3}{*}{G2SLS} & 50 &  & 0.054 & 0.050 & 0.074 & 0.077 &  & -0.134 & 0.172 & 0.218 & 0.257 &  & -0.082 & 0.144 & 0.166 & 0.227 \\
 &  & 100 &  & 0.036 & 0.035 & 0.051 & 0.054 &  & -0.085 & 0.110 & 0.139 & 0.173 &  & -0.045 & 0.083 & 0.094 & 0.130 \\
 &  & 200 &  & 0.097 & 0.037 & 0.103 & 0.057 &  & -0.184 & 0.092 & 0.206 & 0.142 &  & -0.065 & 0.054 & 0.085 & 0.082 \\
\noalign{\vskip 0.5em}
 & \multirow{3}{*}{G3SLS} & 50 &  & -0.017 & 0.069 & 0.071 & 0.105 &  & 0.048 & 0.219 & 0.224 & 0.332 &  & 0.016 & 0.147 & 0.148 & 0.232 \\
 &  & 100 &  & -0.009 & 0.044 & 0.045 & 0.067 &  & 0.027 & 0.132 & 0.134 & 0.198 &  & 0.012 & 0.082 & 0.082 & 0.124 \\
 &  & 200 &  & -0.007 & 0.048 & 0.048 & 0.075 &  & 0.015 & 0.116 & 0.117 & 0.187 &  & 0.007 & 0.056 & 0.057 & 0.089 \\
\bottomrule
\end{tabular}
\begin{tablenotes}[para]
\footnotesize
\setstretch{1}
\item \textit{Notes:} This table reports Monte Carlo results under unobserved homophily. The parameter $m$ controls the intensity of homophilous link formation and is common across individuals. Reported statistics include bias, standard deviation (SD), root mean squared error (RMSE), and inter-quantile range (IQR) for the peer effect ($\beta$=0.7), contextual ($\delta=1$), and direct ($\gamma=1$) Effects.
\end{tablenotes}
\end{threeparttable}
\normalsize
\end{table}
\end{landscape}

\subsection{Monte Carlo Performance - Step 2 (2SLS)}\label{2SLS_MC}

In this section, we present Monte Carlo evidence of the performance of the by-product estimator in Step 2 in Section \ref{est} for all three designs described in Section \ref{mc} in the main manuscript. In particular, Tables \ref{tab:design1_heterogeneity_2sls}, \ref{tab:design2_misclassification_2sls}, and \ref{tab:design3_homophily_2sls} reproduce Tables \ref{tab:design1_heterogeneity_main}, \ref{tab:design2_misclassification_main}, and \ref{tab:design3_homophily_main} in Section \ref{est} in the main manuscript, but with the efficient G3SLS estimator replaced by the 2SLS estimator in Step 2 for designs 1, 2, and 3 respectively.

As expected, the 2SLS seems to show numerical evidence of its consistency across designs, but it displays a larger Monte Carlo standard deviation than the proposed G3SLS for all parameters of interest in Designs 1 and 3 across all scenarios. As for Design 2, Table \ref{tab:design2_misclassification_2sls} below displays a lower Monte Carlo standard deviation as well as RMSE across all sample sizes for the contextual effect 2SLS estimator, in comparison to the corresponding G3SLS reported in Table \ref{tab:design2_misclassification_main}. 

\begin{landscape}   
    \begin{table}[!htbp]
\centering
\small
\begin{threeparttable}
\caption{Estimator Performance under Unobserved Degree Heterogeneity (Design 1)}
\label{tab:design1_heterogeneity_2sls}
\begin{tabular}{ccc c cccc c cccc c cccc}
\toprule
 &  &  &  &\multicolumn{4}{c}{Peer effects} &  &\multicolumn{4}{c}{Contextual effects} &  &\multicolumn{4}{c}{Direct effects} \\
\cmidrule(lr){5-8} \cmidrule(lr){10-13} \cmidrule(lr){15-18}
$m$ & Estimator & $n$ &  & Bias & SD & RMSE & IQR &  & Bias & SD & RMSE & IQR &  & Bias & SD & RMSE & IQR \\
\midrule
\multirow{9}{*}{10} & \multirow{3}{*}{OLS} & 50 &  & 0.359 & 0.046 & 0.361 & 0.072 &  & -1.119 & 0.875 & 1.420 & 1.357 &  & -0.449 & 0.616 & 0.762 & 0.983 \\
 &  & 100 &  & 0.352 & 0.031 & 0.353 & 0.049 &  & -1.104 & 0.576 & 1.245 & 0.897 &  & -0.441 & 0.417 & 0.607 & 0.643 \\
 &  & 200 &  & 0.348 & 0.022 & 0.349 & 0.033 &  & -1.110 & 0.404 & 1.181 & 0.633 &  & -0.472 & 0.287 & 0.552 & 0.442 \\
\noalign{\vskip 0.5em}
 & \multirow{3}{*}{G2SLS} & 50 &  & 0.691 & 0.317 & 0.760 & 0.306 &  & -1.612 & 2.233 & 2.754 & 2.917 &  & -0.588 & 1.087 & 1.235 & 1.502 \\
 &  & 100 &  & 0.742 & 0.268 & 0.789 & 0.278 &  & -2.143 & 1.786 & 2.789 & 2.500 &  & -0.816 & 0.811 & 1.150 & 1.158 \\
 &  & 200 &  & 0.810 & 0.272 & 0.854 & 0.289 &  & -2.532 & 1.551 & 2.969 & 2.203 &  & -1.034 & 0.714 & 1.256 & 0.957 \\
\noalign{\vskip 0.5em}
 & \multirow{3}{*}{2SLS} & 50 &  & 0.192 & 1.524 & 1.535 & 1.788 &  & 0.488 & 8.053 & 8.065 & 9.990 &  & 0.624 & 1.831 & 1.934 & 2.711 \\
 &  & 100 &  & 0.176 & 1.362 & 1.373 & 1.664 &  & 0.610 & 5.756 & 5.786 & 7.580 &  & 0.652 & 1.271 & 1.428 & 1.904 \\
 &  & 200 &  & -0.051 & 1.166 & 1.167 & 1.348 &  & 0.559 & 3.749 & 3.789 & 4.954 &  & 0.625 & 0.894 & 1.090 & 1.350 \\
\midrule
\multirow{9}{*}{12} & \multirow{3}{*}{OLS} & 50 &  & 0.361 & 0.047 & 0.364 & 0.074 &  & -1.132 & 1.043 & 1.539 & 1.628 &  & -0.454 & 0.736 & 0.865 & 1.163 \\
 &  & 100 &  & 0.354 & 0.031 & 0.355 & 0.050 &  & -1.110 & 0.687 & 1.305 & 1.078 &  & -0.441 & 0.501 & 0.667 & 0.776 \\
 &  & 200 &  & 0.350 & 0.022 & 0.351 & 0.034 &  & -1.118 & 0.481 & 1.217 & 0.749 &  & -0.476 & 0.344 & 0.588 & 0.533 \\
\noalign{\vskip 0.5em}
 & \multirow{3}{*}{G2SLS} & 50 &  & 0.732 & 0.304 & 0.793 & 0.305 &  & -1.770 & 2.619 & 3.160 & 3.341 &  & -0.627 & 1.266 & 1.413 & 1.758 \\
 &  & 100 &  & 0.795 & 0.236 & 0.829 & 0.266 &  & -2.199 & 1.952 & 2.940 & 2.632 &  & -0.840 & 0.914 & 1.241 & 1.371 \\
 &  & 200 &  & 0.827 & 0.204 & 0.852 & 0.240 &  & -2.543 & 1.541 & 2.974 & 2.249 &  & -1.044 & 0.732 & 1.275 & 0.997 \\
\noalign{\vskip 0.5em}
 & \multirow{3}{*}{2SLS} & 50 &  & 0.219 & 1.551 & 1.565 & 1.814 &  & 0.240 & 9.574 & 9.573 & 11.903 &  & 0.632 & 2.184 & 2.273 & 3.178 \\
 &  & 100 &  & 0.193 & 1.380 & 1.393 & 1.711 &  & 0.222 & 6.583 & 6.584 & 8.660 &  & 0.639 & 1.474 & 1.606 & 2.204 \\
 &  & 200 &  & 0.069 & 1.244 & 1.245 & 1.422 &  & 0.263 & 4.342 & 4.348 & 5.820 &  & 0.608 & 1.036 & 1.201 & 1.522 \\
\bottomrule
\end{tabular}
\begin{tablenotes}[para]
\footnotesize
\setstretch{1}
\item \textit{Notes:} This table reports Monte Carlo results under unobserved degree heterogeneity. The parameter $m$ controls the intensity of heterogeneity in node degrees and is common across individuals. Reported statistics include bias, standard deviation (SD), root mean squared error (RMSE), and inter-quantile range (IQR) for the peer effect ($\beta$=0.7), contextual ($\delta=1$), and direct ($\gamma=1$) Effects. 2SLS corresponds to the estimator of the second stage discussed in section \ref{est}.
\end{tablenotes}
\end{threeparttable}
\normalsize
\end{table}
\end{landscape}

\begin{landscape}   
    \begin{table}[!htbp]
\centering
\small
\begin{threeparttable}
\caption{Estimator Robustness to Network Link Misclassification (Design 2)}
\label{tab:design2_misclassification_2sls}
\begin{tabular}{ccc c cccc c cccc c cccc}
\toprule
 &  &  &  &\multicolumn{4}{c}{Peer effects} &  &\multicolumn{4}{c}{Contextual effects} &  &\multicolumn{4}{c}{Direct effects} \\
\cmidrule(lr){5-8} \cmidrule(lr){10-13} \cmidrule(lr){15-18}
$\tau$ & Estimator & $n$ &  & Bias & SD & RMSE & IQR &  & Bias & SD & RMSE & IQR &  & Bias & SD & RMSE & IQR \\
\midrule
\multirow{9}{*}{0.01} & \multirow{3}{*}{OLS} & 50 &  & -0.387 & 0.129 & 0.408 & 0.201 &  & -0.425 & 0.314 & 0.528 & 0.497 &  & 0.411 & 0.220 & 0.466 & 0.346 \\
 &  & 100 &  & -0.525 & 0.087 & 0.533 & 0.137 &  & -0.628 & 0.153 & 0.646 & 0.244 &  & 0.203 & 0.098 & 0.225 & 0.154 \\
 &  & 200 &  & -0.555 & 0.062 & 0.559 & 0.098 &  & -0.778 & 0.091 & 0.783 & 0.149 &  & 0.092 & 0.047 & 0.103 & 0.073 \\
\noalign{\vskip 0.5em}
 & \multirow{3}{*}{G2SLS} & 50 &  & -0.538 & 0.222 & 0.582 & 0.305 &  & -0.152 & 0.469 & 0.493 & 0.718 &  & 0.475 & 0.251 & 0.538 & 0.387 \\
 &  & 100 &  & -0.649 & 0.133 & 0.663 & 0.184 &  & -0.457 & 0.215 & 0.505 & 0.326 &  & 0.237 & 0.112 & 0.262 & 0.179 \\
 &  & 200 &  & -0.688 & 0.092 & 0.695 & 0.140 &  & -0.623 & 0.122 & 0.635 & 0.191 &  & 0.114 & 0.051 & 0.125 & 0.080 \\
\noalign{\vskip 0.5em}
 & \multirow{3}{*}{2SLS} & 50 &  & 0.730 & 0.496 & 0.882 & 0.736 &  & 2.422 & 1.989 & 3.134 & 2.780 &  & 0.054 & 0.104 & 0.117 & 0.155 \\
 &  & 100 &  & 0.262 & 0.205 & 0.333 & 0.329 &  & 0.974 & 0.656 & 1.174 & 0.969 &  & 0.019 & 0.061 & 0.064 & 0.096 \\
 &  & 200 &  & 0.055 & 0.184 & 0.192 & 0.287 &  & 0.217 & 0.346 & 0.408 & 0.544 &  & 0.005 & 0.039 & 0.040 & 0.061 \\
\midrule
\multirow{9}{*}{0.05} & \multirow{3}{*}{OLS} & 50 &  & -0.387 & 0.129 & 0.408 & 0.201 &  & -0.425 & 0.314 & 0.528 & 0.497 &  & 0.411 & 0.220 & 0.466 & 0.346 \\
 &  & 100 &  & -0.525 & 0.087 & 0.533 & 0.137 &  & -0.628 & 0.153 & 0.646 & 0.244 &  & 0.203 & 0.098 & 0.225 & 0.154 \\
 &  & 200 &  & -0.555 & 0.062 & 0.559 & 0.098 &  & -0.778 & 0.091 & 0.783 & 0.149 &  & 0.092 & 0.047 & 0.103 & 0.073 \\
\noalign{\vskip 0.5em}
 & \multirow{3}{*}{G2SLS} & 50 &  & -0.538 & 0.222 & 0.582 & 0.305 &  & -0.152 & 0.469 & 0.493 & 0.718 &  & 0.475 & 0.251 & 0.538 & 0.387 \\
 &  & 100 &  & -0.649 & 0.133 & 0.663 & 0.184 &  & -0.457 & 0.215 & 0.505 & 0.326 &  & 0.237 & 0.112 & 0.262 & 0.179 \\
 &  & 200 &  & -0.688 & 0.092 & 0.695 & 0.140 &  & -0.623 & 0.122 & 0.635 & 0.191 &  & 0.114 & 0.051 & 0.125 & 0.080 \\
\noalign{\vskip 0.5em}
 & \multirow{3}{*}{2SLS} & 50 &  & 0.665 & 0.516 & 0.842 & 0.753 &  & 2.524 & 2.004 & 3.222 & 2.862 &  & 0.075 & 0.125 & 0.145 & 0.188 \\
 &  & 100 &  & 0.210 & 0.219 & 0.303 & 0.345 &  & 1.039 & 0.678 & 1.241 & 1.068 &  & 0.032 & 0.068 & 0.075 & 0.108 \\
 &  & 200 &  & 0.019 & 0.189 & 0.190 & 0.303 &  & 0.265 & 0.357 & 0.444 & 0.559 &  & 0.011 & 0.041 & 0.042 & 0.063 \\
\bottomrule
\end{tabular}
\begin{tablenotes}[para]
\footnotesize
\setstretch{1}
\item \textit{Notes:} This table reports Monte Carlo results under network link misclassification. The parameter $\tau$ controls the intensity of misclassification and is common across individuals. Reported statistics include bias, standard deviation (SD), root mean squared error (RMSE), and inter-quantile range (IQR) for the peer effect ($\beta$=0.7), contextual ($\delta=1$), and direct ($\gamma=1$) Effects. 2SLS corresponds to the estimator of the second stage discussed in section \ref{est}.
\end{tablenotes}
\end{threeparttable}
\normalsize
\end{table}
\end{landscape}

\begin{landscape}   
    \begin{table}[!htbp]
\centering
\small
\begin{threeparttable}
\caption{Estimator Performance under Unobserved Homophily (Design 3)}
\label{tab:design3_homophily_2sls}
\begin{tabular}{ccc c cccc c cccc c cccc}
\toprule
 &  &  &  &\multicolumn{4}{c}{Peer effects} &  &\multicolumn{4}{c}{Contextual effects} &  &\multicolumn{4}{c}{Direct effects} \\
\cmidrule(lr){5-8} \cmidrule(lr){10-13} \cmidrule(lr){15-18}
$m$ & Estimator & $n$ &  & Bias & SD & RMSE & IQR &  & Bias & SD & RMSE & IQR &  & Bias & SD & RMSE & IQR \\
\midrule
\multirow{9}{*}{1} & \multirow{3}{*}{OLS} & 50 &  & 0.077 & 0.033 & 0.084 & 0.051 &  & -0.201 & 0.119 & 0.234 & 0.189 &  & -0.109 & 0.082 & 0.137 & 0.129 \\
 &  & 100 &  & 0.071 & 0.022 & 0.074 & 0.035 &  & -0.188 & 0.079 & 0.204 & 0.127 &  & -0.107 & 0.056 & 0.121 & 0.091 \\
 &  & 200 &  & 0.080 & 0.019 & 0.082 & 0.032 &  & -0.181 & 0.062 & 0.191 & 0.099 &  & -0.088 & 0.037 & 0.096 & 0.060 \\
\noalign{\vskip 0.5em}
 & \multirow{3}{*}{G2SLS} & 50 &  & 0.019 & 0.038 & 0.043 & 0.060 &  & -0.051 & 0.127 & 0.137 & 0.195 &  & -0.024 & 0.088 & 0.091 & 0.141 \\
 &  & 100 &  & 0.010 & 0.026 & 0.028 & 0.040 &  & -0.026 & 0.084 & 0.088 & 0.139 &  & -0.013 & 0.060 & 0.061 & 0.094 \\
 &  & 200 &  & 0.015 & 0.022 & 0.027 & 0.036 &  & -0.034 & 0.065 & 0.073 & 0.104 &  & -0.018 & 0.039 & 0.043 & 0.062 \\
\noalign{\vskip 0.5em}
 & \multirow{3}{*}{2SLS} & 50 &  & -0.011 & 0.053 & 0.054 & 0.082 &  & 0.028 & 0.159 & 0.161 & 0.247 &  & 0.009 & 0.105 & 0.105 & 0.169 \\
 &  & 100 &  & -0.004 & 0.032 & 0.032 & 0.050 &  & 0.013 & 0.101 & 0.101 & 0.153 &  & 0.003 & 0.069 & 0.069 & 0.109 \\
 &  & 200 &  & -0.005 & 0.029 & 0.029 & 0.045 &  & 0.010 & 0.077 & 0.078 & 0.124 &  & -0.001 & 0.045 & 0.045 & 0.073 \\
\midrule
\multirow{9}{*}{3} & \multirow{3}{*}{OLS} & 50 &  & 0.118 & 0.040 & 0.124 & 0.064 &  & -0.304 & 0.152 & 0.340 & 0.236 &  & -0.168 & 0.124 & 0.209 & 0.198 \\
 &  & 100 &  & 0.098 & 0.026 & 0.101 & 0.040 &  & -0.258 & 0.095 & 0.275 & 0.149 &  & -0.149 & 0.077 & 0.168 & 0.118 \\
 &  & 200 &  & 0.113 & 0.023 & 0.115 & 0.037 &  & -0.256 & 0.075 & 0.267 & 0.122 &  & -0.124 & 0.052 & 0.135 & 0.081 \\
\noalign{\vskip 0.5em}
 & \multirow{3}{*}{G2SLS} & 50 &  & 0.056 & 0.050 & 0.075 & 0.078 &  & -0.143 & 0.174 & 0.225 & 0.265 &  & -0.078 & 0.137 & 0.157 & 0.209 \\
 &  & 100 &  & 0.032 & 0.034 & 0.046 & 0.055 &  & -0.082 & 0.110 & 0.137 & 0.171 &  & -0.047 & 0.086 & 0.098 & 0.135 \\
 &  & 200 &  & 0.048 & 0.027 & 0.055 & 0.041 &  & -0.109 & 0.080 & 0.135 & 0.123 &  & -0.053 & 0.056 & 0.077 & 0.091 \\
\noalign{\vskip 0.5em}
 & \multirow{3}{*}{2SLS} & 50 &  & -0.011 & 0.076 & 0.077 & 0.116 &  & 0.028 & 0.232 & 0.234 & 0.347 &  & 0.012 & 0.145 & 0.146 & 0.232 \\
 &  & 100 &  & -0.003 & 0.042 & 0.042 & 0.064 &  & 0.016 & 0.133 & 0.134 & 0.205 &  & 0.002 & 0.086 & 0.086 & 0.136 \\
 &  & 200 &  & -0.005 & 0.038 & 0.038 & 0.058 &  & 0.011 & 0.101 & 0.102 & 0.157 &  & 0.000 & 0.059 & 0.059 & 0.093 \\
\bottomrule
\end{tabular}
\begin{tablenotes}[para]
\footnotesize
\setstretch{1}
\item \textit{Notes:} This table reports Monte Carlo results under unobserved homophily. The parameter $m$ controls the intensity of homophilous link formation and is common across individuals. Reported statistics include bias, standard deviation (SD), root mean squared error (RMSE), and inter-quantile range (IQR) for the peer effect ($\beta$=0.7), contextual ($\delta=1$), and direct ($\gamma=1$) Effects. 2SLS corresponds to the estimator of the second stage discussed in section \ref{est}.
\end{tablenotes}
\end{threeparttable}
\normalsize
\end{table}
\end{landscape}

\newpage
\renewcommand\thesection{Appendix \Alph{section}}
\setcounter{section}{4}
\section{Real Data Application Details\label{Appendix_E}}
\renewcommand\thesection{\Alph{section}}
\subsection{Multiplex Network Data}

\noindent The multiplex network composed of two different types of connections among authors and editors (scholars hereafter) can be constructed using their co-authorship information, research interests, education, and employment history. In particular, the \emph{Co-author} layer is made up of connections (edges) between scholars $l$ and $k$ if they co-authored a paper together. The \emph{Alumni} layer is made of edges between scholars $l$ and $k$ if both obtained their Ph.D. from the same institution during the same time window. Details of how these two types of academic ties among scholars are constructed can be found in \ref{scholar_network_constructions} below.

Given that the identities of article authors are known, these two categories of professional connections can be merged and analyzed at the article level. This means articles $i$ and $j$ are connected in networks $\mathbf{W}$ or $\mathbf{W}_{0}$ if at least one of the authors of article $i$ shares a co-authorship or alumni connection with at least one of the authors of article $j$. Table \ref{empt2} presents relevant network statistics for the implied articles' \emph{Co-authors} ($\mathbf{W}=[\mathbf{w}_{i,j}]$) and \emph{Alumni} ($\mathbf{W}_{0}=[\mathbf{w}_{0;i,j}]$) networks. These networks are formed cumulatively; e.g., the 459 articles (nodes) in 2001 include the 221 articles published in 2000 and so on until all 1,628 articles are accounted for in 2006.

The two layers display low densities, which is a common feature of empirical social networks \citep{Paula2017}, with a somewhat small number of connected components for the \emph{Alumni} network and a large number for the \emph{Co-authors} network. Although these networks are formed at the article level, and direct comparisons with classic collaboration network analysis, as in \cite{Goyal2006}, are not possible, they still display \textit{small-world} properties; i.e., the levels of clustering (transitivity) are high, while the average distance (average shortest path) is short \citep[see i.e., ][ for a formal definition of a \textit{small-world}]{Humphries2008}.

\begin{table}
\centering
\small
\caption{Network Statistics}
\begin{tabular}{lrrrrrrr}
\hline
               Statistic &  2000 &   2001 &   2002 &   2003 &    2004 &    2005 &    2006 \\
\hline
 Nodes &   221 &    459 &    729 &    961 &    1,187 &    1,412 &    1,628 \\ \hline
       \textbf{Co-authors Network} &       &        &        &        &         &         &         \\
                   Edges &    28 &    189 &    674 &   1,217 &    2,120 &    3,302 &    4,732 \\
          Average Degree &  0.25 &   0.82 &   1.85 &   2.53 &    3.57 &    4.68 &    5.81 \\
                  Density &  0.00 &   0.00 &   0.00 &   0.00 &     0.00 &    0.00 &   0.00 \\
            Transitivity &  0.82 &   0.83 &   0.71 &   0.66 &    0.65 &    0.63 &     0.6 \\
        Average Distance &  1.04 &    1.10 &   1.25 &   1.57 &    1.92 &    2.56 &    3.32 \\
              Components &   199 &    351 &    478 &    549 &     611 &     651 &     682 \\ \hline
      \textbf{Alumni Network} &       &        &        &        &         &         &         \\
                   Edges &   795 &   3,392 &   8,838 &  15,424 &   22,923 &   33,491 &   45,276 \\
          Average Degree &  7.19 &  14.78 &  24.25 &   32.1 &   38.62 &   47.44 &   55.62 \\
                  Density &  0.03 &   0.03 &   0.03 &   0.03 &    0.03 &    0.03 &    0.03 \\
            Transitivity &  0.57 &   0.57 &   0.55 &   0.55 &    0.55 &    0.54 &    0.53 \\
        Average Distance &  3.14 &   3.25 &   2.95 &    2.80 &    2.81 &    2.73 &     2.70 \\
              Components &    49 &     44 &     48 &     54 &      56 &      62 &      58 \\
\hline
\end{tabular}

\vspace{0.4cm}

\begin{minipage}{0.9\textwidth}
Note: `Nodes' refers to the total number of articles including all previous years since 2000. `Edges' refers to the total number of pair-wise connections among nodes. `Average Degree' represents the average number of edges connected to each node, while `Transitivity' presents the fraction of all possible triangles present in each network. `Density' is defined as the ratio of the total number of observed edges to the total number of all possible edges in these networks. `Components' displays the total number of connected components (subgraphs), while `Average Distance' refers to the average number of steps along the shortest paths for all possible pairs of network nodes.
\end{minipage}
\label{empt2}
\end{table}

\subsubsection{Scholars Network Constructions}\label{scholar_network_constructions}\smallskip

\noindent \textbf{\emph{Co-authors}} -- A co-authorship connection is established between a scholar and all individuals with whom they have co-authored at least one article during the seven-year observation window. For instance, suppose a researcher co-authors a paper in the \emph{American Economic Review} with four collaborators in 2003 and later publishes another paper in the \emph{Quarterly Journal of Economics} in 2004 together with a different researcher who did not participate in the previous publication.  These connections remain active for the years following each publication within the observation window, meaning that earlier collaborations continue to generate links as long as they fall within the seven-year timeframe considered.  

\noindent \textbf{\emph{Alumni}} -- This type of link connects scholars who obtained their doctoral degrees from the same institution within a three-year period. For instance, two economists who both earned their Ph.D. at Yale University around 1992 would be considered to share an alumni connection according to this criterion.

Further details on how connections are formed and additional examples can be found in \cite{netivreg_g3sls} and \cite{EstradaDingRosales2025}.

\begin{table}[H]
\centering
\small
\caption{Estimation Results for Social and Direct Effects using the OLS Estimator}

\begin{tabular}{lccccc}
\hline
& \multicolumn{5}{c}{\emph{Co-author Network}}  \\

\cline{2-6}

& 2002 & 2003 & 2004 & 2005 & 2006  \\
\hline
     Peer Effects ($\widehat{\beta}$) &    0.338*** &    0.396*** &    0.462*** &    0.458*** &    0.462***    \\
                   &  (0.069)      &  (0.054)      &  (0.049)      &  (0.046)       &   (0.040)    \\ \hline
 Contextual Effects ($\boldsymbol{\widehat{\delta}}$)& & & &  \\
 
 \hspace{0.2cm} \texttt{Editor} &    0.178      &    0.083       &       -0.001      &   -0.069       &    -0.030    \\
                   &   (0.280)       &  (0.176)       &  (0.154)     &  (0.148)      &  (0.134) \\
     \hspace{0.2cm}  \texttt{Different Gender} &   -0.125      &    0.014       &   -0.068      &   -0.015       &   -0.035 \\
                   &  (0.207)       &  (0.139)       &  (0.125)       &  (0.118)      &  (0.108)    \\ \hline
                   
Direct Effects ($\boldsymbol{\widehat{\gamma}}$) & & & & \\

   \hspace{0.2cm} \texttt{Editor} &    0.023 &   -0.036       &   -0.032      &    0.039      &     0.070 \\
                   &  (0.141)      &  (0.123)       &  (0.129)       &  (0.128)      &  (0.122)   \\
       \hspace{0.2cm}  \texttt{Different Gender} &    0.232* &     0.190* &     0.190** &    0.143* &    0.132  \\
                   &   (0.130)       &   (0.110)      &  (0.091)      &  (0.081)       &   (0.080)    \\
           \hspace{0.2cm} \texttt{Number of Pages}  &    0.028*** &    0.025*** &    0.021*** &    0.018*** &    0.017***  \\
                   &  (0.004)       &  (0.004)      &  (0.003)      &  (0.003)       &  (0.003)    \\
         \hspace{0.2cm}   \texttt{Number of Authors} &    0.076       &    0.091* &    0.077* &    0.086** &    0.067** \\
                   &  (0.057)       &  (0.047)     &  (0.041)       &  (0.036)      &   (0.030)  \\
      \hspace{0.2cm}   \texttt{Number of References} &     0.010*** &     0.010*** &    0.009*** &     0.010*** &    0.011***  \\
                   &  (0.002)      &  (0.002)       &  (0.002)      &  (0.002)       &  (0.001)   \\
 \hspace{0.2cm}  \texttt{Isolated} &    0.975*** &    1.192*** &    1.319*** &    1.276*** &    1.284***  \\
                   &  (0.129)&       (0.108)     &  (0.099)      &  (0.088)     &  (0.084)   \\
                   
\hline

                 $n$ &      729       &      961      &     1187       &     1412      &     1628     \\
                $R^{2}$ &    0.304      &    0.328       &    0.337       &    0.319       &    0.314 \\
\hline
\end{tabular}

\vspace{0.4cm}

\begin{minipage}{0.9\textwidth}
Note: Estimation results using the OLS estimator. Standard errors are in parenthesis and are clustered at the specific network's components. Stars follow the key: * $p$ $<$ 0.10, ** $p$ $<$ 0.05, and *** $p$ $<$ 0.01, where $p$ stands for $p$-values. $R^2$ are calculated as the squared of the sample correlation coefficients between the observed outcomes and their fitted values.  All specifications include indicator variables for Journal, Year and Alumni Network Components.
\end{minipage}
\label{empt4}
\end{table}

\begin{table}[H]
\centering
\small
\caption{Estimation Results for Social and Direct Effects using the G2SLS Estimator}

\begin{tabular}{lccccc}
\hline
& \multicolumn{5}{c}{\emph{Co-author Network}}  \\

\cline{2-6}

                   & 2002 & 2003 & 2004 & 2005 & 2006 \\
\hline
     Peer Effects ($\widehat{\beta}$) &    0.665*** &    0.604*** &    0.456*** &    0.505*** &      0.600***  \\ 
     &  (0.163)       &  (0.113)       &  (0.079)       &  (0.085)       &  (0.112) \\
     \hline
 Contextual Effects ($\boldsymbol{\widehat{\delta}}$) & & & & &  \\
 \hspace{0.2cm} \texttt{Editor} &    0.061      &     0.020      &    0.001      &   -0.077      &   -0.056    \\
                   &  (0.295)      &  (0.189)       &  (0.153)      &   (0.150)      &  (0.137)  \\
     \hspace{0.2cm}  \texttt{Different Gender} &   -0.198      &   -0.032      &   -0.067      &   -0.028       &   -0.075  \\
                   &  (0.217)      &   (0.150)       &  (0.131)     &  (0.123)       &  (0.121)   \\ \hline
Direct Effects ($\boldsymbol{\widehat{\gamma}}$)  & & & & &  \\

   \hspace{0.2cm} \texttt{Editor} &   -0.032      &   -0.041      &   -0.032      &    0.038       &     0.070 \\
                   &  (0.156)       &  (0.126)      &  (0.129)      &  (0.129)     &  (0.123)     \\
   \hspace{0.2cm}  \texttt{Different Gender} &    0.244* &    0.192* &     0.190** &    0.143* &    0.134* \\
                   &  (0.131)       &  (0.111)      &  (0.091)       &  (0.081)      &  (0.081)   \\
 \hspace{0.2cm} \texttt{Number of Pages} &    0.027*** &    0.024*** &    0.021*** &    0.018*** &    0.016*** \\
                   &  (0.005)       &  (0.004)      &  (0.003)       &  (0.003)      &  (0.003)    \\
\hspace{0.2cm}   \texttt{Number of Authors} &    0.082      &    0.092** &    0.077* &    0.085** &    0.065** \\
                   &  (0.056)       &  (0.046)       &  (0.041)      &  (0.036)       &   (0.030)    \\
    \hspace{0.2cm}   \texttt{Number of References} &    0.008*** &    0.009*** &    0.009*** &     0.010*** &     0.010***  \\
                   &  (0.003)       &  (0.002)       &  (0.002)       &  (0.002)      &  (0.001)   \\
  \hspace{0.2cm}  \texttt{Isolated} &    2.139*** &    1.936*** &    1.299*** &    1.447*** &    1.788***  \\
                   &  (0.131)      &  (0.109)      &  (0.099)      &  (0.089)      &  (0.087)    \\

\hline
                 $n$ &      729      &      961     &     1187      &     1412      &     1628     \\
                $R^{2}$ &    0.286      &     0.320      &    0.337      &    0.319       &    0.311  \\
\hline
\end{tabular}

\vspace{0.4cm}

\begin{minipage}{0.9\textwidth}
Note: Standard errors are in parenthesis and are clustered at the specific network's components. Stars follow the key: * $p$ $<$ 0.10, ** $p$ $<$ 0.05, and *** $p$ $<$ 0.01, where $p$ stands for $p$-values. $R^2$ are calculated as the squared of the sample correlation coefficients between the observed outcomes and their fitted values. All specifications include indicator variables for Journal, Year and Alumni Network Components.
\end{minipage}
\label{empt5}
\end{table}

\subsection{Assessing Assumptions}\label{validate_test}

Since equation \eqref{E1}, Assumption \ref{A1}, and the multiplex data structure together imply an exclusion restriction on the adjacency matrix $\mathbf{W}_{0}$ in our empirical application, we perform an empirical assessment of this assumption. Because the exclusion restriction is fundamentally untestable, we provide heuristic evidence by estimating a linear projection of the outcome onto the set of regressors, including the endogenous peer-outcome term $\mathbf{W}\mathbf{y}$, augmented with the proposed instruments $\mathbf{W}_{0}\mathbf{X}$ and $\mathbf{W}_{0}^{2}\mathbf{X}$ using Ordinary Least Squares (OLS). Namely, we estimate the following linear projection:

\begin{eqnarray}\label{eqn_exrest}
   \texttt{log(citation8)} & = & \beta\textbf{W}\text{ }\texttt{log(citation8)}  +  \delta_1\textbf{W}\text{ }\texttt{Editor}   \\
   \nonumber
    & &{}+{} \delta_2\textbf{W}\text{ }\texttt{Different\hspace{0.8mm}Gender} +  \gamma_1\texttt{Editor} \\
    \nonumber
    & &{}+{} \gamma_2\text{ }\texttt{Different\hspace{0.8mm}Gender} +\gamma_3\text{ }\texttt{Number\hspace{0.8mm}Pages} +\gamma_4\text{ }\texttt{Number\hspace{0.8mm}Authors}   \\
    \nonumber
    & &{}+{} \gamma_5\text{ }\texttt{Number\hspace{0.8mm}References}\text{ }+
   \gamma_6\text{ }\texttt{Isolated}   \\
   \nonumber
    & &{}+{} \theta_1\textbf{W}_{0}\texttt{Editor} + \theta_2\textbf{W}_{0 }\texttt{Different\hspace{0.8mm}Gender}  \\ \nonumber
    & &{}+{} \theta_3\textbf{W}^2_{0}\text{ }\texttt{Editor} +\theta_4\textbf{W}^2_{0}\text{ }\texttt{Different\hspace{0.8mm}Gender} \\ \nonumber
    & &{}+{} \lambda_{r} + \lambda_{t} + \lambda_{0}+\texttt{error}\text{,}
   \nonumber
\end{eqnarray}

\noindent where $\mathbf{W}$ denotes the row-normalized \emph{Co-authorship} adjacency matrix and $\mathbf{W}_{0}$ denotes the row-normalized \emph{Alumni} adjacency matrix. 

It is important to emphasize that because $\mathbf{W}\mathbf{y}$ is endogenous, OLS yields inconsistent estimates of the true structural parameters ($\beta, \boldsymbol{\delta}, \boldsymbol{\gamma}$). Consequently, the estimated coefficients in equation \eqref{eqn_exrest} do not correspond to causal effects and must be interpreted strictly as linear projection pseudo-parameters. The sole purpose of this exercise is to assess empirically whether the instruments $\mathbf{W}_{0}\mathbf{X}$ and $\mathbf{W}_{0}^{2}\mathbf{X}$ possess significant residual explanatory power for the outcome once the other variables are linearly partialled out.

The estimation results for the linear projection in equation \eqref{eqn_exrest} are shown in Table \ref{tab:ols_coauthor_2002_2006}. The pseudo-parameter associated with the peer effect, $\mathbf{W}\times \texttt{log(citation8)}$, is positive and highly statistically significant across all years, with magnitudes that remain stable as the cumulative sample expands.

Among the direct effects, the pseudo-parameters for the number of pages, references, and the isolation indicator display consistently positive and statistically significant coefficients. The projection coefficient for \texttt{Different\hspace{0.8mm}Gender} is positive and significant in earlier samples, declining slightly as the sample size increases, while the coefficient on \texttt{Editor} is not statistically different from zero. Turning to the contextual variables constructed with the co-author network, the projection coefficients on $\mathbf{W}\times \texttt{Editor}$ and $\mathbf{W}\times \texttt{Different\hspace{0.8mm}Gender}$ are generally small and statistically insignificant. 

Regarding the proposed instruments, the pseudo-parameters associated with $\mathbf{W}_{0}\times\texttt{Editor}$, $\mathbf{W}_{0}\times\texttt{Different\hspace{0.8mm}Gender}$, $\mathbf{W}_{0}^2\times\texttt{Editor}$, and $\mathbf{W}_{0}^2\times\texttt{Different\hspace{0.8mm}Gender}$ are not statistically different from zero in most years. Although $\mathbf{W}_{0}^2\times\texttt{Different\hspace{0.8mm}Gender}$ is marginally significant in the earliest cumulative sample (2002), this significance vanishes as the sample expands from $n=729$ to $n=1628$. This pattern suggests that the initial rejection is not robust and is likely driven by higher sampling variability in the smaller sample, rather than by a stable correlation between $\texttt{log(citation8)}$ and the instruments. Empirically, the lack of significant residual explanatory power from these instruments in the linear projection provides heuristic support for the exclusion restriction of the alumni network in our data.

\subsection*{First-Stage Results}

To validate Assumption~\ref{A2}, equation~\eqref{E3} 
is estimated via OLS applied to the empirical setting 
described in Section~\ref{emp}. Again, $\mathbf{W}$ denotes the row-normalized \emph{Co-authorship} adjacency matrix, $\mathbf{W}_{0}$ denotes the 
row-normalized \emph{Alumni} adjacency matrix, and $\mathbf{S}$ contains only the outcome variables and those regressors for which contextual effects are calculated, i.e., $\mathbf{S}= 
[\texttt{log(citation8)},\ \texttt{Editor},\ \texttt{Different\hspace{0.8mm}Gender}]$ 
is the $n \times 3$ matrix collecting the outcome variable and the 
covariates. Equation~\eqref{E3} then specializes to

\begin{equation}\label{eqn:firststage}
    \mathbf{W} \mathbf{S} 
    = \mathbf{W}_{0}\mathbf{S}\,\widehat{\boldsymbol{\Pi}} 
    + \widehat{\mathbf{U}},
\end{equation}

\noindent where $\widehat{\boldsymbol{\Pi}}$ is the $3 \times 3$ matrix 
of estimated projection coefficients. Note that $\mathrm{rank}(\boldsymbol{\Pi}) = 3$ requires the two 
network layers to be sufficiently correlated so that the \emph{Alumni} network 
$\mathbf{W}_{0}$ provides relevant variation for the endogenous 
\emph{Co-authorship} network $\mathbf{W}$ across all columns of $\mathbf{S}$.

Since $\widehat{\boldsymbol{\Pi}}$ is obtained column by column via 
OLS, \eqref{eqn:firststage} amounts to three separate OLS regressions, one for each of $\mathbf{W}\times\texttt{log(citation8)}$, 
$\mathbf{W}\times\texttt{Editor}$, and 
$\mathbf{W}\times\texttt{DifferentGender}$. In all three cases, the 
regressors are identical, namely $\mathbf{W}_{0}\times\texttt{log(citation8)}$, 
$\mathbf{W}_{0}\times\texttt{Editor}$, and 
$\mathbf{W}_{0}\times\texttt{DifferentGender}$, corresponding to the 
columns of $\mathbf{W}_{0}\mathbf{S}$. Each regression is estimated 
over an expanding window for $t = 2002, 2003, 2004, 2005,$ and $2006$ 
using the cumulative sample described in Section~\ref{emp}, and standard errors clustered at the co-author network component level are 
reported throughout.

Table~\ref{tab:first_stage_G3SLS_coauthor} presents the results. The coefficients on $\mathbf{W}_{0}\times\texttt{log(citation8)}$ are positive and statistically significant across all years and all three panels. Furthermore, the $F$-statistics consistently reject the null hypothesis that all slope coefficients are jointly equal to zero. These findings provide empirical support for Assumption \ref{A2} holding in our data, i.e., the predetermined \emph{Alumni} network $\mathbf{W}_{0}$ is sufficiently correlated with the endogenous \emph{Co-authorship} network $\mathbf{W}$, so that $\mathrm{rank}(\boldsymbol{\Pi}) = 3$ holds in the sample, empirically satisfying the relevance condition required for the consistency of the proposed G3SLS estimator.

\newpage

\begin{landscape}   
    \begin{table}[!htbp]\centering
{
\renewcommand{\arraystretch}{1.3}
\caption{OLS Estimation results of equation (\ref{est_eq}) augmented with $\mathbf{W}_0\textbf{X}$ and $\mathbf{W}_0^2\textbf{X}$ as regressors, ($\mathbf{W}:$ \emph{Co-authorship}, $\mathbf{W}_0:$ \emph{Alumni})}
\label{tab:ols_coauthor_2002_2006}
\begin{tabular}{lcccccccccc}
\toprule
\textbf{Variables} & \multicolumn{2}{c}{\textbf{2002}} & \multicolumn{2}{c}{\textbf{2003}} & \multicolumn{2}{c}{\textbf{2004}} & \multicolumn{2}{c}{\textbf{2005}} & \multicolumn{2}{c}{\textbf{2006}} \\
\cmidrule(lr){1-1}\cmidrule(lr){2-3}\cmidrule(lr){4-5}\cmidrule(lr){6-7}\cmidrule(lr){8-9}\cmidrule(lr){10-11}
 & Coef. & SE & Coef. & SE & Coef. & SE & Coef. & SE & Coef. & SE \\
\midrule
$\mathbf{W}\times \texttt{log(citation8)}$ & 0.339*** & 0.071 & 0.395*** & 0.055 & 0.460*** & 0.049 & 0.457*** & 0.047 & 0.460*** & 0.040 \\
$\texttt{Editor}$ & -0.037 & 0.144 & -0.026 & 0.128 & -0.039 & 0.140 & 0.033 & 0.137 & 0.064 & 0.125 \\
$\texttt{Different Gender}$ & 0.362*** & 0.134 & 0.262** & 0.112 & 0.209** & 0.095 & 0.154* & 0.085 & 0.135* & 0.081 \\
$\texttt{Pages}$ & 0.027*** & 0.004 & 0.025*** & 0.004 & 0.021*** & 0.003 & 0.018*** & 0.003 & 0.017*** & 0.003 \\
$\texttt{Authors}$ & 0.076 & 0.056 & 0.094** & 0.047 & 0.077* & 0.041 & 0.085** & 0.036 & 0.067** & 0.030 \\
$\texttt{References}$ & 0.011*** & 0.003 & 0.010*** & 0.002 & 0.009*** & 0.002 & 0.010*** & 0.002 & 0.011*** & 0.001 \\
$\texttt{Isolated}\_{\texttt{coauthor}}$ & 0.996*** & 0.281 & 1.204*** & 0.221 & 1.320*** & 0.194 & 1.277*** & 0.186 & 1.285*** & 0.163 \\
$\mathbf{W}\times \texttt{Editor}$ & 0.119 & 0.294 & 0.090 & 0.185 & -0.009 & 0.167 & -0.077 & 0.149 & -0.039 & 0.142 \\
$\mathbf{W}\times \texttt{Different Gender}$ & -0.025 & 0.211 & 0.084 & 0.143 & -0.049 & 0.128 & -0.004 & 0.121 & -0.030 & 0.110 \\
$\mathbf{W}_0\times \texttt{Editor}$ & 0.468 & 0.462 & 0.525 & 0.392 & 0.295 & 0.414 & 0.143 & 0.440 & 0.404 & 0.417 \\
$\mathbf{W}_0\times \texttt{Different Gender}$ & 0.141 & 0.349 & 0.062 & 0.302 & 0.002 & 0.243 & -0.097 & 0.218 & -0.235 & 0.244 \\
$\mathbf{W}_0^2\times \texttt{Editor}$ & 0.655 & 0.676 & -0.168 & 0.575 & -0.007 & 0.596 & 0.056 & 0.601 & -0.005 & 0.635 \\
$\mathbf{W}_0^2\times \texttt{Different Gender}$ & -1.536** & 0.689 & -1.029* & 0.584 & -0.334 & 0.518 & -0.173 & 0.425 & 0.001 & 0.528 \\
$\texttt{Constant}$ & 1.599*** & 0.311 & 1.512*** & 0.254 & 1.494*** & 0.228 & 1.569*** & 0.217 & 1.553*** & 0.191 \\
\midrule
Journal dummies & Yes &  & Yes &  & Yes &  & Yes &  & Yes &  \\
Year fixed effects & Yes &  & Yes &  & Yes &  & Yes &  & Yes &  \\
Alumni fixed effects & Yes &  & Yes &  & Yes &  & Yes &  & Yes &  \\
\midrule
$n$ & 729 &  & 961 &  & 1187 &  & 1412 &  & 1628 &  \\
$R^2$ & 0.3134 &  & 0.3327 &  & 0.3381 &  & 0.3195 &  & 0.3153 &  \\
\bottomrule
\end{tabular}
}
\begin{flushleft}\footnotesize \hspace{0.5cm} Note: Standard errors (SE) are clustered at the \emph{Co-author} network.  $^{*}\,p<0.10$, $^{**}\,p<0.05$, $^{***}\,p<0.01$.\end{flushleft}
\end{table}
\end{landscape}

\newpage
\begin{landscape} 
\begin{table}[!htbp]\centering
\small
{
\renewcommand{\arraystretch}{1.1}
\caption{First-Stage OLS: $\mathbf{W} \mathbf{S}$ on $\mathbf{W}_{0}\mathbf{S}$, ($\mathbf{W}:$ \emph{Co-authorship}, $\mathbf{W}_0:$ \emph{Alumni})}
\label{tab:first_stage_G3SLS_coauthor}
\begin{tabular}{lcccccccccc}
\toprule
\textbf{Variables} & \multicolumn{2}{c}{\textbf{2002}} & \multicolumn{2}{c}{\textbf{2003}} & \multicolumn{2}{c}{\textbf{2004}} & \multicolumn{2}{c}{\textbf{2005}} & \multicolumn{2}{c}{\textbf{2006}} \\
\cmidrule(lr){1-1}\cmidrule(lr){2-3}\cmidrule(lr){4-5}\cmidrule(lr){6-7}\cmidrule(lr){8-9}\cmidrule(lr){10-11}
 & Coef. & SE & Coef. & SE & Coef. & SE & Coef. & SE & Coef. & SE \\
\midrule
\multicolumn{11}{l}{Panel A: Dependent variable: $\mathbf{W}\times \texttt{log(citation8)}$} \\
\midrule
$\mathbf{W}_{0}\times \texttt{log(citation8)}$ & 0.496*** & 0.038 & 0.582*** & 0.036 & 0.624*** & 0.036 & 0.678*** & 0.037 & 0.717*** & 0.042 \\
$\mathbf{W}_{0}\times \texttt{Editor}$ & 0.425 & 0.937 & 1.391 & 1.134 & 2.162* & 1.115 & 1.94* & 1.01 & 2.277** & 0.961 \\
$\mathbf{W}_{0}\times \texttt{Different\hspace{0.8mm}Gender}$ & 0.137 & 0.574 & -0.024 & 0.476 & -0.151 & 0.491 & -0.532 & 0.438 & -0.723 & 0.455 \\
\midrule
$n$ & 729 &  & 961 &  & 1187 &  & 1412 &  & 1628 &  \\
$F$ & 87.78 &  & 128.99 &  & 145.23 &  & 127.24 &  & 103.55 &  \\
\midrule
\multicolumn{11}{l}{Panel B: Dependent variable: $\mathbf{W}\times \texttt{Editor}$} \\
\midrule
$\mathbf{W}_{0}\times \texttt{log(citation8)}$ & 0.009** & 0.004 & 0.013*** & 0.004 & 0.013*** & 0.004 & 0.016*** & 0.004 & 0.014*** & 0.003 \\
$\mathbf{W}_{0}\times \texttt{Editor}$ & 0.137 & 0.098 & 0.186* & 0.108 & 0.232** & 0.103 & 0.168** & 0.079 & 0.192*** & 0.071 \\
$\mathbf{W}_{0}\times \texttt{Different\hspace{0.8mm}Gender}$ & 0.033 & 0.046 & -0.051 & 0.041 & -0.029 & 0.038 & -0.043 & 0.029 & -0.033 & 0.03 \\
\midrule
$n$ & 729 &  & 961 &  & 1187 &  & 1412 &  & 1628 &  \\
$F$ & 9.82 &  & 13.00 &  & 12.58 &  & 10.09 &  & 10.39 &  \\
\midrule
\multicolumn{11}{l}{Panel C: Dependent variable: $\mathbf{W}\times \texttt{Different\hspace{0.8mm}Gender}$} \\
\midrule
$\mathbf{W}_{0}\times \texttt{log(citation8)}$ & 0.013*** & 0.003 & 0.018*** & 0.004 & 0.021*** & 0.003 & 0.024*** & 0.003 & 0.026*** & 0.003 \\
$\mathbf{W}_{0}\times \texttt{Editor}$ & 0.103 & 0.081 & 0.04 & 0.088 & 0.144 & 0.099 & 0.174* & 0.096 & 0.159* & 0.088 \\
$\mathbf{W}_{0}\times \texttt{Different\hspace{0.8mm}Gender}$ & 0.142** & 0.058 & 0.172*** & 0.062 & 0.111** & 0.054 & 0.044 & 0.046 & 0.085* & 0.049 \\
\midrule
$n$ & 729 &  & 961 &  & 1187 &  & 1412 &  & 1628 &  \\
$F$ & 18.85 &  & 37.56 &  & 47.12 &  & 45.98 &  & 84.54 &  \\
\bottomrule
\end{tabular}
}
\vspace{0.05cm}
\begin{flushleft}\footnotesize \hspace{1cm}Note: Standard errors (SE) are clustered at the \emph{Co-author} network. $^{*}p<0.10$, $^{**}p<0.05$, $^{***}p<0.01$.\end{flushleft}
\end{table}
\end{landscape}

\putbib[typ]
\end{bibunit}


\begin{thebibliography}{41}
\newcommand{\enquote}[1]{``#1''}
\providecommand{\natexlab}[1]{#1}
\providecommand{\url}[1]{\texttt{#1}}
\providecommand{\urlprefix}{URL }

\bibitem[{Aldasoro and Alves(2018)}]{Aldasoro_Alves_JFinStab}
Aldasoro, \protect{I\~{n}aki} and \protect{Iv\'{a}n} Alves. 2018.
\newblock \enquote{Multiplex Interbank Networks and Systemic Importance: An
  Application to European Data.}
\newblock \emph{Journal of Financial Stability} 35:17 -- 37.

\bibitem[{Ammermueller and \protect{J\"{o}rn-Steffen}
  Pischke(2009)}]{Ammermueller_Pischke_JLO}
Ammermueller, Andreas and \protect{J\"{o}rn-Steffen} Pischke. 2009.
\newblock \enquote{Peer Effects in European Primary Schools: Evidence from the
  Progress in International Reading Literacy Study.}
\newblock \emph{Journal of Labor Economics} 27~(3):315--348.

\bibitem[{An et~al.(2026)An, Estrada, Estrada, and
  Jacho-Chavez}]{An_et_al_2026}
An, Weihua, Pablo Estrada, Juan Estrada, and David Jacho-Chavez. 2026.
\newblock \enquote{Estimating Peer Influence in Multilayer Networks.}
\newblock \emph{Social Networks} 86:252--265.

\bibitem[{Atkisson et~al.(2020)Atkisson, G\'{o}rski, Jackson, Ho{\l{}}yst, and
  D'Souza}]{Jackson_Multiplex}
Atkisson, Curtis, Piotr~J. G\'{o}rski, Matthew~O. Jackson, Janusz~A.
  Ho{\l{}}yst, and Raissa~M. D'Souza. 2020.
\newblock \enquote{Why Understanding Multiplex Social Network Structuring
  Processes Will Help Us Better Understand the Evolution of Human Behavior.}
\newblock \emph{Evolutionary Anthropology} 29~(3):102--107.

\bibitem[{Baltagi and Deng(2015)}]{BD:2015}
Baltagi, Badi~H. and Ying Deng. 2015.
\newblock \enquote{EC3SLS Estimator for a Simultaneous System of Spatial
  Autoregressive Equations with Random Effects.}
\newblock \emph{Econometric Reviews} 34~(6-10):659--694.

\bibitem[{Baltagi and Liu(2011)}]{BL:2011}
Baltagi, Badi~H. and Long Liu. 2011.
\newblock \enquote{Instrumental Variable Estimation of a Spatial Autoregressive
  Panel Model with Random Effects.}
\newblock \emph{Economics Letters} 111~(2):135--137.

\bibitem[{Boccaletti et~al.(2014)Boccaletti, Bianconi, Criado, del Genio,
  G{\'{o}}mez-Garde{\~{n}}es, Romance, Sendi{\~{n}}a-Nadal, Wang, and
  Zanin}]{Boccaletti2014}
Boccaletti, S., G.~Bianconi, R.~Criado, C.~I. del Genio,
  J.~G{\'{o}}mez-Garde{\~{n}}es, M.~Romance, I.~Sendi{\~{n}}a-Nadal, Z.~Wang,
  and M.~Zanin. 2014.
\newblock \enquote{The Structure and Dynamics of Multilayer Networks.}
\newblock \emph{Physics Reports} 544:1--122.

\bibitem[{Bramoull{\'{e}}, Djebbari, and Fortin(2009)}]{Bramoulle2009}
Bramoull{\'{e}}, Yann, Habiba Djebbari, and Bernard Fortin. 2009.
\newblock \enquote{Identification of Peer Effects through Social Networks.}
\newblock \emph{Journal of Econometrics} 150~(1):41--55.

\bibitem[{Carrell, Sacerdote, and West(2013)}]{carrell2013}
Carrell, Scott~E, Bruce~I Sacerdote, and James~E West. 2013.
\newblock \enquote{From Natural Variation to Optimal Policy? The Importance of
  Endogenous Peer Group Formation.}
\newblock \emph{Econometrica} 81~(3):855--882.

\bibitem[{Chan et~al.(2024)Chan, Estrada, Huynh, Jacho-Chavez, Lam, and
  Sanchez-Aragon}]{Chan_et_al_social_effects}
Chan, TszKin~Julian, Juan Estrada, Kim Huynh, David Jacho-Chavez, Chungsang~Tom
  Lam, and Leonardo Sanchez-Aragon. 2024.
\newblock \enquote{Estimating Social Effects with Randomized and Observational
  Network Data.}
\newblock \emph{Journal of Econometric Methods} 13~(2):205--224.

\bibitem[{Chandrasekhar and Lewis(2016)}]{Chandrasekhar_unpub_2016}
Chandrasekhar, Arun~G. and Randall Lewis. 2016.
\newblock \enquote{Econometrics of Sampled Networks.}
\newblock Unpublished Manuscript.

\bibitem[{Colussi(2018)}]{Colussi2018}
Colussi, Tommaso. 2018.
\newblock \enquote{Social Ties in Academia: A Friend Is a Treasure.}
\newblock \emph{The Review of Economics and Statistics} 100~(1):45--50.

\bibitem[{de~Paula(2017)}]{Paula2017}
de~Paula, {\'{A}}ureo. 2017.
\newblock \emph{Econometrics of Network Models}, \emph{Econometric Society
  Monographs}, vol.~1.
\newblock Cambridge University Press, 268--323.

\bibitem[{Erd\"{o}s and R\'{e}nyi(1959)}]{Erdos1959}
Erd\"{o}s, P and A~R\'{e}nyi. 1959.
\newblock \enquote{On Random Graphs.}
\newblock \emph{Publicationes Mathematicae (Debrecen)} 6:290--297.

\bibitem[{Estrada(2022)}]{netivreg_g3sls}
Estrada, Juan. 2022.
\newblock \emph{Causal Inference in Multilayered Networks}.
\newblock Ph.d. dissertation, Emory University.
\newblock Available at
  \href{https://etd.library.emory.edu/concern/etds/3r074w158}{https://etd.library.emory.edu/concern/etds/3r074w158}.

\bibitem[{Estrada, Ding, and Rosales-Castillo(2025)}]{EstradaDingRosales2025}
Estrada, Juan, Cheng Ding, and Cristhian Rosales-Castillo. 2025.
\newblock \enquote{Social Interactions in Multilayer Observational Networks.}
\newblock Unpublished manuscript.

\bibitem[{Estrada et~al.(2025)Estrada, Estrada, Huynh, Jacho-Chavez, and
  Sanchez-Aragon}]{netivreg_stata_journal}
Estrada, Pablo, Juan Estrada, Kim Huynh, David Jacho-Chavez, and Leonardo
  Sanchez-Aragon. 2025.
\newblock \enquote{netivreg: Estimation of Peer Effects in Endogenous Social
  Networks.}
\newblock \emph{The Stata Journal} 25~(2):344--373.

\bibitem[{Falk and Ichino(2006)}]{falk2006}
Falk, Armin and Andrea Ichino. 2006.
\newblock \enquote{Clean Evidence on Peer Effects.}
\newblock \emph{Journal of labor economics} 24~(1):39--57.

\bibitem[{Fruehwirth(2014)}]{Fruehwirth_RevStat}
Fruehwirth, Jane~Cooley. 2014.
\newblock \enquote{Can Achievement Peer Effect Estimates Inform Policy? A View
  from Inside the Black Box.}
\newblock \emph{The Review of Economics and Statistics} 96~(3):514--523.

\bibitem[{Goldsmith-Pinkham and Imbens(2013)}]{Goldsmith-Pinkham2013}
Goldsmith-Pinkham, Paul and Guido~W. Imbens. 2013.
\newblock \enquote{Social Networks and the Identification of Peer Effects.}
\newblock \emph{Journal of Business and Economic Statistics} 31~(3):253--264.

\bibitem[{Graham(2017)}]{Graham2017}
Graham, Bryan~S. 2017.
\newblock \enquote{An Econometric Model of Network Formation With Degree
  Heterogeneity.}
\newblock \emph{Econometrica} 85~(4):1033--1063.

\bibitem[{Ji and Jin(2016)}]{AoAS_coauthorship}
Ji, Pengsheng and Jiashun Jin. 2016.
\newblock \enquote{Coauthorship and Citation Networks for Statisticians.}
\newblock \emph{Annals of Applied Statistics} 10~(4):1779--1812.

\bibitem[{Johnsson and Moon(2021)}]{Johnsson2019}
Johnsson, Ida and Hyungsik~Roger Moon. 2021.
\newblock \enquote{Estimation of Peer Effects in Endogenous Social Networks:
  Control Function Approach.}
\newblock \emph{The Review of Economics and Statistics} 103~(2):328--345.

\bibitem[{Kelejian and Piras(2014)}]{Kelejian_et_al_2014}
Kelejian, Harry~H and Gianfranco Piras. 2014.
\newblock \enquote{Estimation of spatial models with endogenous weighting
  matrices, and an application to a demand model for cigarettes.}
\newblock \emph{Regional Science and Urban Economics} 46:140--149.

\bibitem[{Kelejian and Prucha(1998)}]{Kelejian1998}
Kelejian, Harry~H. and Ingmar~R. Prucha. 1998.
\newblock \enquote{A Generalized Spatial Two-Stage Least Squares Procedure for
  Estimating a Spatial Autoregressive Model with Autoregressive Disturbances.}
\newblock \emph{Journal of Real Estate Finance and Economics} 17~(1):99--121.

\bibitem[{Kelejian and Prucha(1999)}]{Kelejian_Prucha_1999_ER}
---{}---{}---. 1999.
\newblock \enquote{A Generalized Moments Estimator for the Autoregressive
  Parameter in a Spatial Model.}
\newblock \emph{International Economic Review} 40~(2):509--533.

\bibitem[{Kivela et~al.(2014)Kivela, Arenas, Barthelemy, Gleeson, Moreno, and
  Porter}]{Kivela_multilayer_network_2014}
Kivela, M., A.~Arenas, M.~Barthelemy, J.~P. Gleeson, Y.~Moreno, and M.~A.
  Porter. 2014.
\newblock \enquote{Multilayer Networks.}
\newblock \emph{Journal of Complex Networks} 2~(3):203–271.

\bibitem[{K{\"o}nig, Liu, and Zenou(2019)}]{Konig_et_al_2019}
K{\"o}nig, Michael, Xiaodong Liu, and Yves Zenou. 2019.
\newblock \enquote{R\&D networks: Theory, empirics, and policy implications.}
\newblock \emph{The Review of Economics and Statistics} 101~(3):476--491.

\bibitem[{Lee(2003)}]{Lee2003}
Lee, Lung~Fei. 2003.
\newblock \enquote{Best Spatial Two-Stage Least Squares Estimators for a
  Spatial Autoregressive Model with Autoregressive Disturbances.}
\newblock \emph{Econometric Reviews} 22~(4):307--335.

\bibitem[{Lee(2007)}]{Lee2007}
---{}---{}---. 2007.
\newblock \enquote{Identification and Estimation of Econometric Models with
  Group Interactions, Contextual Factors and Fixed Effects.}
\newblock \emph{Journal of Econometrics} 140~(2):333--374.

\bibitem[{Lee et~al.(2021)Lee, Liu, Patacchini, and Zenou}]{Lee_et_al_2021}
Lee, Lung-Fei, Xiaodong Liu, Eleonora Patacchini, and Yves Zenou. 2021.
\newblock \enquote{Who is the key player? A network analysis of juvenile
  delinquency.}
\newblock \emph{Journal of Business \& Economic Statistics} 39~(3):849--857.

\bibitem[{Lewbel, Qu, and Tang(2023)}]{Lewbel_Qu_Tang}
Lewbel, Arthur, Xi~Qu, and Xun Tang. 2023.
\newblock \enquote{Social Networks with Unobserved Links.}
\newblock \emph{Journal of Political Economy} 131~(4):898--946.

\bibitem[{Liu, Patacchini, and Zenou(2014)}]{liu2014}
Liu, Xiaodong, Eleonora Patacchini, and Yves Zenou. 2014.
\newblock \enquote{Endogenous Peer Effects: Local Aggregate or Local Average?}
\newblock \emph{Journal of Economic Behavior \& Organization} 103:39 -- 59.

\bibitem[{Manta et~al.(2022)Manta, Ho, Huynh, and Jacho-Chavez}]{manta2021}
Manta, Alexandra, Anson~T.Y. Ho, Kim~P. Huynh, and David~T. Jacho-Chavez. 2022.
\newblock \enquote{Estimating Social Effects in a Multilayered Linear-in-Means
  Model with Network Data.}
\newblock \emph{Statistics \& Probability Letters} 183:109331.

\bibitem[{Mariano(2001)}]{mariano2001simultaneous}
Mariano, Roberto~S. 2001.
\newblock \enquote{Simultaneous Equation Model Estimators: Statistical
  Properties and Practical Implications.}
\newblock In \emph{A Companion to Theoretical Econometrics}, edited by Badi~H.
  Baltagi, chap.~6. Oxford: Blackwell Publishers, 122--143.

\bibitem[{Newman(2004)}]{Newman_2004a}
Newman, Mark E.~J. 2004.
\newblock \enquote{Coauthorship Networks and Patterns of Scientific
  Collaboration.}
\newblock \emph{Proceedings of the National Academy of Sciences}
  101~(90001):5200--5205.

\bibitem[{Pomeroy, Dasandi, and Mikhaylov(2019)}]{PoliSci_Multiplex}
Pomeroy, Caleb, Niheer Dasandi, and Slava~Jankin Mikhaylov. 2019.
\newblock \enquote{Multiplex Communities and the Emergence of International
  Conflict.}
\newblock \emph{PLoS One} 14~(10):1--17.

\bibitem[{Reza, Manchanda, and Chong(2021)}]{Reza2019}
Reza, Sadat, Puneet Manchanda, and Juin-Kuan Chong. 2021.
\newblock \enquote{Identification and Estimation of Endogenous Peer Effects in
  the Presence of Multiple Reference Groups.}
\newblock \emph{Management Science} ~(8):5070--5105.

\bibitem[{Rodriguez et~al.(2024)Rodriguez, Huynh, Jacho-Ch\'avez, and
  S\'anchez-Arag\'on}]{RHJS:2024}
Rodriguez, Belicia, Kim~P. Huynh, David~T. Jacho-Ch\'avez, and Leonardo
  S\'anchez-Arag\'on. 2024.
\newblock \enquote{Abstract Readability: Evidence from Top-5 Economics
  Journals.}
\newblock \emph{Economics Letters} 235:111541.

\bibitem[{Sacerdote(2001)}]{Sacerdote_QJE}
Sacerdote, Bruce. 2001.
\newblock \enquote{Peer Effects with Random Assignment: Results for Dartmouth
  Roommates.}
\newblock \emph{Quarterly Journal of Economics} 116~(2):681--704.

\bibitem[{Sheng and Sun(2025)}]{Sheng_et_al_2025}
Sheng, Shuyang and Xiaoxia Sun. 2025.
\newblock \enquote{Social interactions in endogenous groups.}
\newblock \emph{arXiv preprint arXiv:2306.01544}
  \urlprefix\url{https://arxiv.org/abs/2306.01544}.

\end{thebibliography}


\begin{thebibliography}{5}
\newcommand{\enquote}[1]{``#1''}
\providecommand{\natexlab}[1]{#1}
\providecommand{\url}[1]{\texttt{#1}}
\providecommand{\urlprefix}{URL }

\bibitem[{de~Paula(2017)}]{Paula2017}
de~Paula, {\'{A}}ureo. 2017.
\newblock \emph{Econometrics of Network Models}, \emph{Econometric Society
  Monographs}, vol.~1.
\newblock Cambridge University Press, 268--323.

\bibitem[{Estrada(2022)}]{netivreg_g3sls}
Estrada, Juan. 2022.
\newblock \emph{Causal Inference in Multilayered Networks}.
\newblock Ph.d. dissertation, Emory University.
\newblock Available at
  \href{https://etd.library.emory.edu/concern/etds/3r074w158}{https://etd.library.emory.edu/concern/etds/3r074w158}.

\bibitem[{Estrada, Ding, and Rosales-Castillo(2025)}]{EstradaDingRosales2025}
Estrada, Juan, Cheng Ding, and Cristhian Rosales-Castillo. 2025.
\newblock \enquote{Social Interactions in Multilayer Observational Networks.}
\newblock Unpublished manuscript.

\bibitem[{Goyal, van~der Leij, and Moraga-Gonz{\'{a}}lez(2006)}]{Goyal2006}
Goyal, Sanjeev, Marco~J. van~der Leij, and Jos{\'{e}}~Luis
  Moraga-Gonz{\'{a}}lez. 2006.
\newblock \enquote{Economics: An Emerging Small World.}
\newblock \emph{Journal of Political Economy} 114~(2):403--412.

\bibitem[{Humphries and Gurney(2008)}]{Humphries2008}
Humphries, {Mark D.} and Kevin Gurney. 2008.
\newblock \enquote{Network `small-world-ness': A quantitative method for
  determining canonical network equivalence.}
\newblock \emph{PLoS One} 3~(4).

\end{thebibliography}
\end{document}